\newif\ifprocs
\procstrue
\procsfalse %last one wins; comment this line for proceedings formatting

\ifprocs
\documentclass[sigconf,screen]{acmart}
\else
\documentclass[11pt]{article}
\fi

\usepackage{algorithm}
\usepackage[noend]{algpseudocode}

\ifprocs\else
\usepackage[margin=1in,letterpaper]{geometry}
\usepackage{amssymb}
\fi
\usepackage{amsmath}
\usepackage{bbm}
\usepackage{amsthm}
\usepackage{color,soul}
\usepackage{xcolor}
\usepackage{graphicx}
\usepackage{bm}
\usepackage{booktabs}
\usepackage{multirow}
\usepackage[export]{adjustbox}
\usepackage{thm-restate}

\ifprocs\else
\usepackage[colorlinks=true, linkcolor=red, urlcolor=blue, citecolor=gray]{hyperref}
\fi

\newtheorem{theorem}{Theorem}[section]
\newtheorem{lemma}[theorem]{Lemma}
\newtheorem{definition}[theorem]{Definition}

\newtheorem{corollary}[theorem]{Corollary}
\newtheorem{claim}[theorem]{Claim}

\newtheorem{observation}[theorem]{Observation}
\newtheorem{remark}[theorem]{Remark}

\usepackage{color}
\definecolor{darkgreen}{rgb}{0,0.5,0}
\usepackage{hyperref}
\hypersetup{
	unicode=false,          % non-Latin characters in AcrobatΓÇÖs bookmarks
	colorlinks=true,        % false: boxed links; true: colored links
	linkcolor=red,          % color of internal links (change box color with linkbordercolor)
	citecolor=darkgreen,        % color of links to bibliography
	filecolor=magenta,      % color of file links
	urlcolor=cyan           % color of external links
}

\usepackage[capitalize, nameinlink]{cleveref}
\crefname{theorem}{Theorem}{Theorems}
\Crefname{lemma}{Lemma}{Lemmas}
\Crefname{claim}{Claim}{Claims}
\Crefname{observation}{Observation}{Observations}

\DeclareMathOperator*{\argmax}{argmax}
\DeclareMathOperator*{\argmin}{argmin}

\theoremstyle{plain}

\theoremstyle{plain}
\theoremstyle{plain}

\theoremstyle{plain}

\theoremstyle{plain}

\theoremstyle{definition}

\theoremstyle{definition}

\theoremstyle{definition}

\theoremstyle{plain}
\makeatother

\providecommand{\lemmaname}{Lemma}
\providecommand{\propositionname}{Proposition}
\providecommand{\theoremname}{Theorem}
\providecommand{\corollaryname}{Corollary}
\providecommand{\definitionname}{Definition}
\providecommand{\assumptionname}{Assumption}
\providecommand{\remarkname}{Remark}

\global\long\def\RR{\mathbb{R}}

\newcommand{\EX}{{\mathbb E}}

\newcommand{\set}[1]{\{#1\}}
\newcommand{\tO}{\tilde{O}}

\allowdisplaybreaks

{\makeatletter
 \gdef\xxxmark{%
   \expandafter\ifx\csname @mpargs\endcsname\relax % in minipage?
     \expandafter\ifx\csname @captype\endcsname\relax % in figure/caption?
       \marginpar{xxx}% not in a caption or minipage, can use marginpar
     \else
       xxx % notice trailing space
     \fi
   \else
     xxx % notice trailing space
   \fi}
 \gdef\xxx{\@ifnextchar[\xxx@lab\xxx@nolab}
 \long\gdef\xxx@lab[#1]#2{{\bf [\xxxmark #2 ---{\sc #1}]}}
 \long\gdef\xxx@nolab#1{{\bf [\xxxmark #1]}}
}

\usepackage{tikz}
\usepackage{ifthen}
\usetikzlibrary{arrows,shapes}
\usetikzlibrary{intersections,positioning}

\usetikzlibrary{shapes}
\usetikzlibrary{plotmarks}

\newenvironment{proofof}[1]{\noindent{\bf Proof of #1:}}{$\qed$\par}

\title{Towards Tight Bounds for Spectral Sparsification of Hypergraphs}
\ifprocs
\titlenote{Full version available at \href{https://arxiv.org/abs/2011.06530}{arXiv:2011.06530}.}
\fi

\ifprocs
\author{Michael Kapralov}
\authornote{Supported in part by ERC Starting Grant 759471.}
\affiliation{%
  \institution{\'Ecole Polytechnique F\'ed\'erale de Lausanne}
  \city{}
  \country{}
}
\email{michael.kapralov@epfl.ch}

\author{Robert Krauthgamer}
\authornote{Work partially supported by ONR Award N00014-18-1-2364,
  the Israel Science Foundation grant \#1086/18,
  and a Minerva Foundation grant.
}
\affiliation{%
  \institution{Weizmann Institute of Science}
  \city{}
  \country{}
}
\email{robert.krauthgamer@weizmann.ac.il}

\author{Jakab Tardos}
\authornote{Supported by ERC Starting Grant 759471.}
\affiliation{%
  \institution{\'Ecole Polytechnique F\'ed\'erale de Lausanne}
  \city{}
  \country{}
}
\email{jakab.tardos@epfl.ch}

\author{Yuichi Yoshida}
\authornote{Supported in part by JSPS KAKENHI Grant Number 18H05291 and 20H05965.}
\affiliation{%
  \institution{National Institute of Informatics}
  \city{}
  \country{}
}
\email{yyoshida@nii.ac.jp}

\else

\author{
  Michael Kapralov\thanks{Supported in part by ERC Starting Grant 759471.}\\
  \'Ecole Polytechnique F\'ed\'erale de Lausanne\\
  \texttt{michael.kapralov@epfl.ch}
  \and
  Robert Krauthgamer%
  \thanks{Work partially supported by ONR Award N00014-18-1-2364,
  the Israel Science Foundation grant \#1086/18,
  and a Minerva Foundation grant.
  }\\
  Weizmann Institute of Science\\
  \texttt{robert.krauthgamer@weizmann.ac.il}
  \and
  Jakab Tardos\thanks{Supported by ERC Starting Grant 759471.}\\
  \'Ecole Polytechnique F\'ed\'erale de Lausanne\\
  \texttt{jakab.tardos@epfl.ch}
  \\
  \and
  Yuichi Yoshida\\
  National Institute of Informatics\\
  \texttt{yyoshida@nii.ac.jp}
}
\fi %\ifprocs

\ifprocs
\setcopyright{acmcopyright}
\acmPrice{15.00}
\acmDOI{10.1145/3406325.3451061}
\acmYear{2021}
\copyrightyear{2021}
\acmSubmissionID{stoc21main-p205-p}
\acmISBN{978-1-4503-8053-9/21/06}
\acmConference[STOC '21]{Proceedings of the 53rd Annual ACM SIGACT Symposium on Theory of Computing}{June 21--25, 2021}{Virtual, Italy}
\acmBooktitle{Proceedings of the 53rd Annual ACM SIGACT Symposium on Theory of Computing (STOC '21), June 21--25, 2021, Virtual, Italy}
\fi %\ifprocs

\begin{document}

\ifprocs\else
\maketitle
\fi

\begin{abstract}
Cut and spectral sparsification of graphs have numerous applications,
including e.g.~speeding up algorithms for cuts and Laplacian solvers.
These powerful notions have recently been extended to hypergraphs,
which are much richer and may offer new applications. However, the current bounds on the size of hypergraph sparsifiers are not as tight as the corresponding bounds for graphs.

Our first result is a polynomial-time algorithm that,
given a hypergraph on $n$ vertices with maximum hyperedge size $r$,
outputs an  $\epsilon$-spectral sparsifier
with $O^*(nr)$ hyperedges, where $O^*$ suppresses $(\epsilon^{-1} \log n)^{O(1)}$ factors.
This size bound improves the two previous bounds:
$O^*(n^3)$ [Soma and Yoshida, SODA'19]
and $O^*(nr^3)$ [Bansal, Svensson and Trevisan, FOCS'19]. Our main technical tool is a new method for proving concentration of the nonlinear analogue of the quadratic form of the Laplacians for hypergraph expanders.

We complement this with lower bounds on the bit complexity of any compression scheme that $(1+\epsilon)$-approximates all the cuts in a given hypergraph,
and hence also on the bit complexity of every $\epsilon$-cut/spectral sparsifier.
These lower bounds are based on Ruzsa-Szemer\'edi graphs,
and a particular instantiation yields an $\Omega(nr)$ lower bound on the bit complexity even for fixed constant $\epsilon$. In the case of hypergraph cut sparsifiers, this is tight up to polylogarithmic factors in $n$, due to recent result of [Chen, Khanna and Nagda, FOCS'20]. For spectral sparsifiers it narrows the gap to $O^*(r)$.

Finally, for directed hypergraphs,
we present an algorithm that computes an $\epsilon$-spectral sparsifier
with $O^*(n^2r^3)$ hyperarcs,
where $r$ is the maximum size of a hyperarc.
For small $r$, this improves over $O^*(n^3)$
known from [Soma and Yoshida, SODA'19],
and is getting close to the trivial lower bound of $\Omega(n^2)$ hyperarcs.
\end{abstract}

\ifprocs
\fi %\ifprocs

\ifprocs
\maketitle
\fi

\ifprocs\else
\thispagestyle{empty}
\newpage
\thispagestyle{empty}
\tableofcontents
\setcounter{page}{0}
\newpage
\fi

%!TEX root=./000-main.tex

\section{Introduction}

Sparsification is an algorithmic paradigm
where a dense object is replaced by a sparse one with similar features,
which often leads to significant improvements in efficiency of algorithms,
including running time, space complexity, and communication.
We study edge-sparsification of hypergraphs,
which replaces a hypergraph $G=(V,E,w)$ with a sparse hypergraph $\widetilde G$
that has the same vertex set $V$ and only a few hyperedges,
often a reweighted subset of $E$.
This is a natural extension of edge-sparsification of ordinary graphs,
which includes key concepts such as
cut sparsifiers,  
spectral sparsifiers, 
and flow sparsifiers. 
These were studied extensively from numerous angles, including various constructions, tight size bounds, related variants, and practical applications.
As this literature is too vast to cover here,
we quickly recap the basics for graphs before discussing hypergraphs, which are our focus here.

\paragraph{Graphs.}

Let $G=(V,E,w)$ be an edge-weighted graph, where $w \in \mathbb{R}_+^E$.
The \emph{energy} of a vector $x \in \mathbb{R}^V$ in $G$ is defined as
\[
  Q_G(x) = \sum_{uv \in E}w_{uv}{(x_u-x_v)}^2,
\]
and can also be written as $x^\top L_G x$,
where $L_G$ is the Laplacian matrix of $G$.
Spielman and Teng~\cite{Spielman2011} introduced the notion of
an \emph{$\epsilon$-spectral sparsifier} of $G$,
which is a graph $\widetilde{G}=(V,\widetilde{E},\widetilde{w})$
that satisfies (for $0\le \epsilon \leq 1/2$)
\begin{align}  \label{eq:spectral-sparsification}
  \forall x \in \mathbb{R}^V,
  \qquad
  Q_{\widetilde{G}}(x) \in (1\pm \epsilon) Q_G(x) .
\end{align}
 The \emph{size} of a spectral sparsifier $\widetilde{G}$ is $|\widetilde{E}|$.

We say that an edge $e \in E$ is \emph{cut} by $S \subseteq V$ if one endpoint of $e$ belongs to $S$ and another one belongs to $V\setminus S$.
The total weight of edges cut by $S$ is clearly $Q_G(1_S)$, where $1_S \in \mathbb{R}^V$ denotes the characteristic vector of a set $S \subseteq V$.

A spectral sparsifier $\widetilde{G}$ of a graph $G$ preserves many important properties of $G$: its cuts have approximately the same weight as those in $G$;
its Laplacian $L_{\widetilde G}$ approximates every eigenvalue of $L_G$;
electrical flows in $\widetilde{G}$ approximate those in $G$.
It is extremely useful to have a spectral sparsifier with a small number of edges
because algorithms that involve these quantities
can be applied on the sparsifier $\widetilde{G}$ instead of on $G$,
with only a small loss in accuracy.

A spectral sparsifier of size $O(n/\epsilon^2)$ can be computed
in almost linear time~\cite{Lee2018},
where $n$ is the number of vertices in $G$.

\paragraph{Hypergraphs.}
A hypergraph is a natural extension of a graph, which can represent relations between three or more entities, and has proved useful to solve problems in practical areas such as computer vision~\cite{Huang2009,Ochs2012}, bioinformatics~\cite{Klamt2009}, and information retrieval~\cite{Gibson2000}.
Many of those problems, such as semi-supervised learning~\cite{Hein2013,Yadati2019,Zhang2020} and link prediction~\cite{Yadati2020}, involve the notion of energy for hypergraphs, where the \emph{energy} of a vector $x \in \mathbb{R}^V$ in an edge-weighted hypergraph $G=(V,E,w)$ is defined as
\begin{align} \label{eq:hypergraph-energy}
  Q_G(x) = \sum_{e \in E}w_e \max_{u,v \in e}{(x_u - x_v)}^2.
\end{align}
This definition matches the one for graphs when every hyperedge in $G$ is of size two.
As before, $Q_G(1_S)$ gives the total weight of hyperedges cut by $S$,
where we regard a hyperedge $e \in E$ as \emph{cut} if $e \cap S \neq \emptyset$ and $e \cap (V\setminus S) \neq \emptyset$.

Spectral sparsification of hypergraphs was first defined
by Soma and Yoshida~\cite{Soma2019}, as follows.
Similarly to graphs, an \emph{$\epsilon$-spectral sparsifier} of $G$
is a hypergraph $G=(V,\widetilde{E},\widetilde{w})$
that satisfies~\eqref{eq:spectral-sparsification}. This is a strictly stronger notion than that of the hypergraph cut sparsifier which has been previously studied in~\cite{Newman2013} and~\cite{Kogan2015}.

Besides the applications mentioned above, spectral sparsifiers for hypergraphs
were used to show agnostic learnability of a certain subclass of submodular functions~\cite{Soma2019}.

Soma and Yoshida~\cite{Soma2019} showed that every hypergraph $G$
admits an $\epsilon$-spectral sparsifier
with $\widetilde{O}(n^3/\epsilon^2)$ hyperedges,%
\footnote{Throughout, we write $\widetilde{O}(\cdot)$ to suppress a factor of $\log^{O(1)} n$.}
which is non-trivial because a general hypergraph can have $2^n-1$ (non-empty) hyperedges.
Moreover, they provide an algorithm recovering this sparsifier, that runs in close to linear time
(in the input size).
Later, Bansal, Svensson and Trevisan~\cite{Bansal2019}
showed that every hypergraph $G$ admits a spectral sparsifier with $\widetilde{O}(nr^3/\epsilon^2)$ hyperedges,
where $r$ is the maximum size of a hyperedge in $G$.
Note that this bound is incomparable to~\cite{Soma2019} because $r$ could be as large as $n$.

\subsection{Results}\label{sec:results}
\paragraph{Spectral sparsification of undirected hypergraphs.}
Our first contribution is an algorithm that constructs an $\epsilon$-spectral sparsifier of a hypergraph with only $\widetilde{O}(nr/\epsilon^{O(1)})$ hyperedges,
which improves upon the previous constructions mentioned above.
(See Table~\ref{tab:results} for known bounds for hypergraph sparsification.)

\begin{theorem}\label{thm:general-sparsification}
Given an $r$-uniform hypergraph $G=(V,E,w)$ and $1/n\le\epsilon\le1/2$,
one can compute in polynomial time 
with probability $1-o(1)$
an $\epsilon$-spectral sparsifier of $G$ with $nr(\epsilon^{-1}\log n)^{O(1)}$ hyperedges. The running time is $O(mr^2) + n^{O(1)}$,
where $m=|E|$.
\end{theorem}

To simplify notation, our entire technical analysis considers
a hypergraph $G=(V,E)$ that is unweighted (i.e., unit weight hyperedges),
reserving the letter $w$ for the edge weights in the sparsifier.

This is actually without loss of generality,
see
\ifprocs
the full version of the paper.
\else
Section~\ref{subsec:weighted}.
\fi

We stress that Theorem~\ref{thm:general-sparsification} in fact applies
to hypergraphs with maximum size of a hyperedge at most $r$.
Indeed, in our analysis every hyperedge is a multiset of vertices,
and therefore a hyperedge with less than $r$ vertices can be trivially extended
to a multiset of exactly $r$ vertices by copying an arbitrary vertex,
without changing the energy (but it might affect vertex degrees).

\begin{table}[t]
  \centering
  \caption{Bounds on the size of hypergraph sparsifiers }
  \begin{tabular}[t]{lll}
    \toprule
    cut sparsification
    & spectral sparsification
    & reference
    \\
    \midrule
    $\widetilde{O}(n^2/\epsilon^2)$
    & & \cite{Newman2013} implicitly
    \\
    $\widetilde{O}(nr/\epsilon^2)$
    & & \cite{Kogan2015}
    \\
    & $\widetilde{O}(n^3/\epsilon^2)$
    & \cite{Soma2019}
    \\
    & $\widetilde{O}(nr^3/\epsilon^2)$
    & \cite{Bansal2019}
    \\
    $\widetilde{O}(n/\epsilon^2)$
    & & \cite{Chen20}
    \\
    & $\widetilde{O}(nr/\epsilon^{O(1)})$
    & Theorem~\ref{thm:general-sparsification}\\
    \bottomrule
    \end{tabular}
    \label{tab:results}
\end{table}

\paragraph{Bit-complexity lower bound.}

To complement Theorem~\ref{thm:general-sparsification}, we consider lower bounds on the bit complexity of sparsifiers.
Here, we consider \emph{$\epsilon$-cut sparsifiers},
which require that~\eqref{eq:spectral-sparsification} holds
only for vectors of the form $x=1_S$.
This notion actually predates spectral sparsification and was first defined
by Bencz{\'u}r and Karger~\cite{BK15} for graphs,
and by Kogan and Krauthgamer~\cite{Kogan2015} for hypergraphs.
Obviously, lower bounds for cut sparsifiers directly imply
the same lower bounds also for spectral sparsifiers.

The second contribution of this work is a surprising connection between \emph{a Ruzsa-Szemer\'{e}di (RS) graph}~\cite{ruzsa1978triple}, which is a well-studied notion in extremal graph theory, and a lower bound on the bit complexity of a hypergraph cut sparsifier.
Here, an (ordinary) graph is called a \emph{$(t,a)$-RS graph} if its edge set is the union of $t$ induced matchings of size $a$.
Then, we show the following.

\begin{theorem}\label{thm:lower-bound}
Suppose that there exists a $(t,a)$-Ruzsa-Szemer\'{e}di graph on $n$ vertices with $a \geq 6000\sqrt{n\log n}$.
Assume also one can compress unweighted $(t+1)$-uniform hypergraphs $G=(V,E)$ on $2n$ vertices into $k$ bits,
from which $Q_G(1_S)$ can be approximated for every $S \subseteq V$
within factor $1\pm \epsilon$, where $\epsilon = O(a/n)$.
Then, $k = \Omega(at)$.
\end{theorem}

For example, by instantiating Theorem~\ref{thm:lower-bound} with the  $(n^{\Omega(1/\log \log n)},\allowbreak n/3-o(n))$-Ruzsa-Szemer\'{e}di graphs known due to Fischer~et~al. \cite{Fischer2002}, we deduce that $\Omega(nr)$ bits
are necessary to encode all the cut values
of an arbitrary $r$-uniform hypergraph with $r = n^{O(1/\log \log n)}$,
even within a fixed constant ratio $1+\epsilon$.

This lower bound is in fact near-tight.
Indeed, Chen, Khanna, and Nagda~\cite{Chen20} showed very recently
that every hypergraph $G$ admits an $\epsilon$-cut sparsifier with $O(n\log n/\epsilon^2)$ hyperedges, which are actually sampled from $G$.
Applying this construction with fixed $\epsilon$ and $r = n^{O(1/\log \log n)}$
yields a sparsifier of $G$ with $O(n\log n)$ hyperedges;
encoding a hyperedge (including its weight, which is bounded by $n^r$)
takes at most $O(r\log n)$ bits,
and thus one can encode all the cuts of $G$ using $O(nr\log^2 n)$ bits.
It follows that our lower bound is optimal up to a lower order factor $O(\log^2 n)$. Instantiating our lower bound with the original construction of Ruzsa and Szemer\'{e}di~\cite{ruzsa1978triple}, we can rule out the possibility of compressing the cut structure of a hypergraph with $n$ vertices and maximum hyperedge size $r$ with significantly less than $nr$ space, and a polynomial scaling in the error (that is with $nr^{1-\Omega(1)}\varepsilon^{-O(1)}$ space), {\em for any $r$}. 
\ifprocs
See the full version of the paper for more details.
\else
See~Corollaries~\ref{corollary1},~\ref{corollary2} and~\ref{corollary3} in Section~\ref{sec:lower-bound} for more details.
\fi

In fact, our space lower bound for hypergraphs far exceeds
the $O(n\log n/\epsilon^2)$ bits that suffices to approximately represent
all the cuts of an (ordinary) graph by simply storing a cut sparsifier.
We thus obtain the first provable separation between the bit complexity
of approximating all the cuts of a graph vs.~of a hypergraph.

\paragraph{Spectral sparsification of directed hypergraphs.}
We also consider spectral sparsification of directed hypergraphs.
Here, a hyperarc $e$ consists of two disjoint sets,
called the \emph{head} $h(e) \subseteq V$
and the \emph{tail} $t(e) \subseteq V$,
and the \emph{size} of the hyperarc is $|t(e)|+|h(e)|$.
A directed hypergraph $G=(V,E)$ then consists of a vertex set $V$
and a hyperarc set $E$.
For an edge-weighted directed hypergraph $G=(V,E,w)$ and a vector $x \in \mathbb{R}^V$, the \emph{energy} of $x$ in $G$ is defined as
\begin{align}
  Q_G(x) = \sum_{e \in E}w_e \max_{u \in t(e),v \in h(e)}{(x_u-x_v)}_+^2, \label{eq:directed-hypergraph-energy}
\end{align}
where $(a)_+ = \max\{a,0\}$.
Again, it is defined so that $Q_G(1_S)$ is the total weight of hyperarcs that are cut by $S$, where a hyperarc $e$ is \emph{cut} if $t(e) \cap S \neq \emptyset$ and $h(e) \cap (V \setminus S) \neq \emptyset$.

It is not difficult to see that a spectral sparsifier might require
(in the worst-case) at least $\Omega(n^2)$ hyperarcs,
even for an ordinary directed graph.
Indeed, consider a balanced bipartite clique directed from one side of the bipartition towards the other.
Here, every arc is the unique arc crossing some particular directed cut, and hence a sparsifier must keep all the $\Omega(n^2)$ arcs
(see also~\cite{IT18, CPS20}).
However, Soma and Yoshida~\cite{Soma2019} showed that every directed hypergraph
admits an $\epsilon$-spectral sparsifier with $\widetilde{O}(n^3/\epsilon^2)$ hyperarcs.
We tighten this gap by showing that $\widetilde O(n^2/\epsilon^2)$ hyperarcs are sufficient when every hyperarc is of constant size.

\begin{theorem}\label{thm:directed-hypergraph-sparsification}
Given a directed hypergraph $G=(V,E)$ with maximum hyperarc size at most $r$ such that $11r\le\sqrt{\epsilon n}$, and a value $\epsilon\le1/2$,
one can compute in polynomial time with probability $1-o(1)$
an $\epsilon$-spectral sparsifier of $G$ with $O(n^2r^3\log^2 n/\epsilon^2)$ hyperarcs.
\end{theorem}

We note that Theorem~\ref{thm:directed-hypergraph-sparsification} is stated under the assumption $11 r\leq \sqrt{\epsilon n}$, which is useful for our analysis for technical reasons. For larger values of $r$ the result of~\cite{Soma2019} gives a better bound on the number of hyperedges in the sparsifier, and therefore this assumption is not restrictive.

\subsection{Related Work}\label{sec:related}

The first construction of cut sparsifiers for hypergraphs was given by Kogan and Krauthgamer~\cite{Kogan2015}
and uses $O(n(r+\log n)/\epsilon^2)$ hyperedges.
They also mention that an upper bound of $O(n^2 \log n/\epsilon^2)$ hyperedges
follows implicitly from the results of Newman and Rabinovich~\cite{Newman2013}.
Very recently (and independent of our work),
Chen, Khanna, and Nagda~\cite{Chen20} improved this bound
to $O(n\log n/\epsilon^2)$ hyperedges, which is near-optimal
because the current lower bound is $\Omega(n/\epsilon^2)$ edges,
and actually holds for (ordinary) graphs~\cite{ACKQWZ16,CKST19}.

Louis~\cite{Louis2015} (later merged with Chan~et~al.~\cite{Chan2018}) initiated the spectral theory for hypergraphs, in which the Laplacian operator $L:\mathbb{R}^V \to \mathbb{R}^V$ of a hypergraph is defined so that its ``quadratic form'' $x L(x)$ coincides with the energy~\eqref{eq:hypergraph-energy}.
As opposed to the graph case, here the Laplacian operator is merely piecewise linear, and hence computing its eigenvalues/vectors is hard.
He showed that $O(\log r)$-approximation is possible,
and that obtaining a better approximation ratio is NP-hard assuming the Small-Set Expansion (SSE) hypothesis~\cite{Raghavendra_2010}.
He further showed a Cheeger inequality for hypergraphs, which implies that, given a vector $x \in \mathbb{R}^V$ with a small energy, one can efficiently find a set $S \subseteq V$ of small expansion, which roughly means that the number of hyperedges cut by $S$ is small relative to the number of hyperedges incident to vertices in $S$ (see Section~\ref{sec:pre} for details).
Since then, several other algorithms for finding sets of small expansion have been proposed~\cite{Takai2020,ikeda2018finding}.

Yoshida~\cite{Yoshida2016} proposed another piecewise linear Laplacian for directed graphs and used it to study structures of real-world networks.
Generalizing the Laplacians for hypergraphs and directed graphs, Laplacian $L$ for directed hypergraphs was proposed~\cite{Li2018,Yoshida2019}, whose quadratic form $x^\top L(x)$ coincides with~\eqref{eq:directed-hypergraph-energy}.

\subsection{Discussion}
An obvious open question is the existence of a spectral sparsifier with $\widetilde{O}(n)$ hyperedges.
As we will see in Section~\ref{sec:technical-overview}, our overall strategy to construct a spectral sparsifier is decomposing the input hypergraph into good expanders (in a non-trivial way) and then sparsifying each expander. Here a \emph{good expander} is a hypergraph with the maximum possible expansion up to a constant factor (see Section~\ref{subsec:hypergraph-and-expansion} for the details).
However, we do not even know whether we can spectrally sparsify hypergraph expanders with $\widetilde{O}(n)$ hyperedges.
To see the difficulty, note that a graph expander has expansion $\Theta(1)$ whereas an $r$-uniform hypergraph expander has expansion $\Theta(1/r)$.
Let $x \in \mathbb{R}^V$ be a vector with $\sum_{v \in V} x_v^2d(v) = 1$, where $d(v)$ is the degree of a vertex $v \in V$.
Then by the Cheeger inequality for hypergraphs (Theorem~\ref{thm:hypergraph-cheeger}), the energy of $x$ in a graph expander is $\Omega(1)$ whereas that in an $r$-uniform hypergraph expander is merely $\Omega(1/r)$.
Hence preserving the latter energy is seemingly a harder problem.

%%% Local Variables:
%%% mode: latex
%%% TeX-master: "000-main"
%%% End:
%!TEX root=./000-main.tex

\section{Preliminaries}\label{sec:pre}
In the paper, we will often need to deal with additive or multiplicative errors of various approximations. For simplicity of notation we use $\widetilde A = A\pm\delta$ to denote $A-\delta \le\widetilde A\le A+\delta$, and we use $\widetilde A=(1\pm\epsilon)A$ to denote $(1-\epsilon)A\le\widetilde A\le(1+\epsilon)A$.

\subsection{Hypergraph and Expansion}\label{subsec:hypergraph-and-expansion}

A hypergraph $G=(V,E)$ on a vertex set $V$ is usually defined so that $E$ is a set of hyperedges, each of which is an arbitrary (non-empty) subset of $V$ (as opposed to ordinary graphs, where it is a subset of size two). In a slight departure from the norm, we allow the hyperedges in $E$ to be multisets instead.
That is, a hyperedge may contain certain vertices multiple times. This may be thought of as a generalization of the use of self-loops in ordinary graphs, which can be considered as multisets containing a single vertex with multiplicity two --- and thus having size two.
This slight change in the definition allows us to consider $r$-uniform hypergraphs throughout the paper without loss of generality, which makes the analysis
\ifprocs
\else
in Section~\ref{sec:general-sparsification}
\fi
 much simpler. We call a hypergraph \emph{$r$-uniform} if all of its hyperedges have size $r$.

Let us denote the multiplicity of a vertex $v\in V$ in $e\in E$ by $\mu_e(v)$.
Then the \emph{size} of $e$ is $\sum_{v\in V}\mu_e(v)$ (as is normal for multisets).
The \emph{degree} of a vertex $v$ is $d(v)=\sum_{e\in E}\mu_e(v)$.

Furthermore, we also allow hyperedges in $E$ to appear with multiplicity, i.e., parallel edges. This means that $E$ itself is a multiset. We call a hypergraph that has neither multiset edges nor multiple instances of the same edge a \textit{simple hypergraph}.

For a hypergraph $G=(V,E)$ and a set $S \subseteq V$, let $E(S) \subseteq E$ be the multiset of hyperedges $e \in E$ such that every vertex in $e$ belongs to $S$.
Then, let $G[S] = (S, E(S))$ denote the subgraph of $G$ induced by $S$.

Let $G=(V,E)$ be a hypergraph and $S \subseteq V$ be a vertex set.
The \emph{volume} of $S$, denoted by $\mathrm{vol}(S)$, is $\sum_{v \in S}d(v)$.
We say that a hyperedge $e \in E$ is \emph{cut} by $S$ if $e\cap S\neq\emptyset$ and $e\cap(V \setminus S)\neq\emptyset$.
In this context, we often call a pair $(S,V \setminus S)$ a \emph{cut}.
Let $E(S,V \setminus S)$ denote the set of hyperedges cut by $S$.
Then, the \emph{expansion} of $S$ (or a cut $(S,V \setminus S)$) is
$$\Phi(S)=\frac{|E(S,V \setminus S)|}{\min\left\{\mathrm{vol}(S),\mathrm{vol}(V \setminus S)\right\}}.$$
The expansion of a hypergraph $G=(V,E)$ is defined to be $\Phi(G) :=  \min_{S\subseteq V} \Phi(S)$.
For $\Phi \geq 0$, we say that $G$ is a \emph{$\Phi$-expander} if $\Phi(G) \geq \Phi$.

\subsection{Spectral Hypergraph Theory}\label{sec:prelim-spectral}
We briefly review spectral theory for hypergraphs.
See, e.g.,~\cite{Chan2018,Yoshida2019} for more details.

\begin{definition}
Let $G=(V,E)$ be a hypergraph and $x \in \mathbb{R}^V$ be a vector.
The \emph{energy} of a hyperedge $e\in E$ with respect to $x$ is defined as
$Q_x(e)=\max_{a,b\in e}{(x_a-x_b)}^2,$
 and the energy of a subset of hyperedges $E'\subseteq E$ is $Q_x(E')=\sum_{e\in E'}Q_x(e)$, respectively.
Finally, the \emph{entire energy} of $x$ is defined as the energy of all hyperedges combined, that is, $Q(x)=Q_x(E)$.
If the underlying hypergraph $G$ is unclear from context, we specify by writing $Q_G(x)$.
\end{definition}

\begin{definition}\label{def:spectral-sparsifier}
	Let $G=(V,E)$ be a hypergraph and $\epsilon > 0$.
	$\widetilde{G}=(V,\widetilde{E},w)$ is a weighted subgraph of $G$ if $w$ is a vector in $\mathbb R_+^E$, mapping each hyperedge $e\in E$ to a non-negative value, and $\widetilde E$ denotes $\{e\in E \mid w_e>0\}$. Such a weighted subgraph is called an $\epsilon$-spectral sparsifier if for any vector $x\in\mathbb R^V$, $\widetilde Q(x)=(1\pm\epsilon)\cdot Q(x)$, where $\widetilde Q$ denotes energy with respect to the graph $\widetilde G$, that is
	$$\widetilde Q(x)=\sum_{e\in\widetilde E}w_e\cdot Q_x(e).$$
	The size of such a sparsifier is $|\widetilde E|$.
\end{definition}

Given a hypergraph $G=(V,E)$ and a vector $x \in \mathbb{R}^V$, we can define an ordinary graph $G_x=(V,E_x)$ so that the energy of $x$ on $G$ and that on $G_x$ are equal.
Specifically, we define $E_x$	as the multiset
$$
E_x
= \left\{\Big(\argmax_{a\in e}x_a,\argmin_{b\in e}x_b\Big)\;\middle|\; e\in E\right\},
$$
where ties are broken arbitrarily.

The following Cheeger's inequality is a cornerstone of spectral hypergraph theory. Although a similar theorem has been proven in~\cite[Theorem 6.1]{Chan2018}, we include the proof in Appendix~\ref{app:A} for completeness because we do not know whether their proof goes through when we allow for multiset hyperedges.
\begin{theorem}[Hypergraph Cheeger's inequality]\label{thm:hypergraph-cheeger}
	Let $G=(V,E)$ be an $r$-uniform hypergraph with expansion at least $\Phi\le2/r$.
	Then for any vector $x \in \mathbb R^V$ with $\sum_{v\in V}x_vd(v)=0$, we have
	$$Q(x)\ge\frac{r\Phi^2}{32}\sum_{v\in V}x_v^2d(v).$$
\end{theorem}
\begin{remark}
	In fact, for simple hypergraphs the requirement $\Phi\le2/r$ is unnecessary and the statements holds in full generality. In our setting, this requirement is crucial, as non-simple $r$-uniform hypergraphs may have expansion $\omega(1/r)$, in which case the statement clearly does not hold.
\end{remark}

%%% Local Variables:
%%% mode: latex
%%% TeX-master: "000-main"
%%% End:

%!TEX root=./000-main.tex

\section{Technical Overview}\label{sec:technical-overview}

In this section we briefly outline the techniques used in the proofs of our main results.

\subsection{Spectral Sparsification of Expanders}

We begin by constructing spectral sparsifiers for ``good'' hypergraph expanders, where we call a hypergraph a good expander if it has expansion at least $\widetilde\Omega(1/r)$.
Even in this restricted case, no result better than $\widetilde O(nr^3/\epsilon^2)$~\cite{Bansal2019} was known previously.
Our plan will then be to partition general input hypergraphs into a series of good expanders.
The expansion $\widetilde \Omega(1/r)$ is in some sense the best we can hope for. In fact, $r$-unifrom simple hypergraphs cannot have an expansion better than $\Theta(1/r)$ and consequently no expander decomposition algorithm can guarantee expansion more than that.

To construct our spectral sparsifier for a good expander, we apply importance sampling to the input hypergraph. We sample each hyperedge $e$ independently with some probability $p_e$ and scale it up with weight $1/p_e$ if sampled. This guarantees that $\mathbb E(\widetilde G)=G$ and so for any vector $x\in\mathbb R^V$ we have $\mathbb E\widetilde Q(x)=Q(x)$, where $\widetilde Q$ denotes the energy with respect to the sparsifier.
In our case, $p_e$ is inversely proportional to $\min_{v\in e}d(v)$,
and then the expected number of sampled hyperedges is proportional to $n$
--- simply charge each hyperedge $e$ to a vertex $v\in e$ of minimum degree,
then each vertex is in charge of $O(1)$ sampled hyperedges in expectation.
It remains to prove that the random quantity $Q(x)$ concentrates well around its expectation \emph{for all vectors $x$ simultaneously}.

So far this is a known technique: similar approaches to constructing spectral sparsifiers in ordinary graphs have appeared in many works, starting from~\cite{SpielmanT11,Spielman2011a}. However, all of these rely on concentration inequalities for linear functions of independent random variables related to the matrix Bernstein inequality -- see, e.g.,~\cite{Tropp2015}\footnote{More precisely, the proof of the necessary concentration properties in~\cite{SpielmanT11} heavily relies on linearity of the graph Laplacian (specifically, the proof proceeded by bounding the trace of a high power of a corresponding matrix using combinatorial methods), and the analysis of ~\cite{Spielman2011a} relies on a concentration inequality for linear functions of independent random variables due to Rudelson and Vershynin~\cite{RudelsonV07}. Both of these proofs can also be reproduced using the matrix Bernstein inequality.}. Unfortunately, the energy of a hypergraph is not a linear transformation and such tools cannot be applied to it. Two recent works on spectral sparsification of hypergraphs developed methods for circumventing this problem, namely~\cite{Soma2019} and~\cite{Bansal2019}. The former uses a rather crude union bound plus Chernoff bound argument, and loses a factor of $n$ in the size of the sparsifier, both for undirected and directed hypergraphs. The latter, namely the recent work of~\cite{Bansal2019} uses Talagrand's comparison inequality and generic chaining to compare the hypergraph sampling process to effective resistance sampling of~\cite{Spielman2011a}, and loses a factor of $r^3$ in the size of the sparsifier. In this work we derive a simultaneous concentration inequality for $\widetilde Q(x)$ for all $x\in {\mathbb R}^V$ from more basic principles, and obtain a sparsifier with $\approx n r$ hyperedges as a result -- a bound that is seemingly best that can be obtained through the expander decomposition route.

Note that for a single, fixed vector $x\in\mathbb R^V$, the concentration inequality $\widetilde Q(x)=(1\pm\epsilon)Q(x)$ holds with high probability by the Chernoff bound (Theorem~\ref{thm:chernoff}). Our broad strategy will be to prove concentration over individual choices of $x$, and combine these results through a union bound.
An obvious issue is that $x$ is a continuous variable, making a direct union bound infeasible. We therefore have to discretize it, rounding each $x$ to some $\widetilde x$ from a finite net. Our plan then becomes to prove the chain of approximations
$$Q(x)\cong Q(\widetilde x)\cong\widetilde Q(\widetilde x)\cong\widetilde Q(x),$$
where the second approximation ($Q(\widetilde x)\cong\widetilde Q(\widetilde x)$) utilizes the idea above of a Chernoff bound for each $\widetilde x$ plus a union bound over the net.

This turns out to be too simplistic, and the analysis requires a more technical discretization of $x$. Recall that the energy of the whole hypergraph can be written as a sum of the energies of the individual hyperedges:
$$Q(x)=\sum_{e\in E}Q_x(e).$$
We categorize hyperedges based on a carefully chosen metric $\max_{v\in e}x_v^2\cdot\min_{v\in e}d(v)$, which we will call the hyperedge's \emph{power}. If a hyperedge's power is approximately $2^{-i}$, then it resides in the $i^\text{th}$ category $E_i$ (see Section~\ref{sec:expander-sparsification-correctness}). We have in total a logarithmic number of categories. This categorization is important, because the power of a hyperedge turns out to be closely related to the strength of the Chernoff bound applicable to it, as well as to the required accuracy of the approximation $\widetilde x$. That is, some cruder approximation $\widetilde x$ may be sufficient to guarantee $Q_x(E_1)\cong Q_{\widetilde x}(E_1)$, but it might not be able to guarantee the same for a later category.
Conversely, the Chernoff bound is stronger (i.e., the failure probability is smaller) at larger values of $i$.
Thus, for each $i$ we discretize $x$ into a different vector $x^{(i)}$
(rather than the same $\tilde x$) and we prove individually for each $i$ that
$$Q_x(E_i)\cong Q_{x^{(i)}}(E_i)\cong\widetilde Q_{x^{(i)}}(E_i)\cong\widetilde Q_x(E_i).$$

Here, ``$\cong$'' necessarily covers both multiplicative \emph{and additive} errors. Indeed, we have no guarantee on the sizes of these categories. Some $E_i$ could contain only a single hyperedge, in which case a simple Chernoff bound would yield no concentration whatsoever. This is where we utilize the additive-multiplicative version (Theorem~\ref{thm:am-Chernoff}).
Since we have $\Theta(\log n)$ categories to sum over,
we naturally allow additive error $\Theta(\epsilon Q(x)/\log n)$.

Note that $x^{(i)}$ is a discretization of $x$ specialized to preserve the energies of hyperedges in $E_i$. Intuitively, the energy of such a hyperedge $e$ is dictated by the largest value of $x_v^2$ within it. This value necessarily belongs to a vertex satisfying $x_v^2d(v)\gtrapprox2^{-i}$. Thus, it should be enough for our rounding to preserve the $x$-values of vertices that satisfy this. To this end, we round the $x$-values of vertices with $x_v^2d(v)\gtrapprox2^{-i}$ \emph{carefully} --- by an inverse polynomial amount in $n$. However, we round the $x$-values of all other vertices to $0$ ---
which is obviously a crude (non-careful) rounding.
Thus, if there are only $k_i$ vertices we have to be careful about, the number of possible settings of $x^{(i)}$ becomes $\approx\exp(k_i)$.

Recall the formula of the additive-multiplicative Chernoff bound from Theorem~\ref{thm:am-Chernoff}. In our case, the allowable multiplicative error is always $\approx 1+\epsilon$, while the allowable additive error is always $\approx\epsilon Q(x)$. The only quantity that varies from level to level is the range of the random variables involved. If a specific hyperedge is sampled, it is scaled up by $1/p_e\approx\min_{v\in e}d(v)$, and the energy of this weighted hyperedge can be upper bounded by $\approx\max_{v\in e}x_v^2\cdot\min_{v\in e}d(v)$ --- exactly the power of the hyperedge. Thus, at level $i$, the additive-multiplicative Chernoff bound guarantees a failure probability of $\approx\exp(-2^i Q(x))$. (Here we omit the $\epsilon$ terms, along with others, for simplicity.)

Finally, we want to equate the terms in the exponents of the Chernoff bound with the enumeration of $x^{(i)}$'s, so as to bound the total failure probability. We use hypergraph Cheeger (Theorem~\ref{thm:hypergraph-cheeger}) to relate $k_i$ to $Q(x)$. Suppose that $x$ is normalized in the sense that $\sum_{v\in V}x_v^2d(v)=1$. This immediately gives that $k_i\le2^i$ by definition. On the other hand, we can finally use our assumption that the input hypergraph $G$ was a good expander, since hypergraph Cheeger gives us that $Q(x)\gtrapprox1/r$. This makes the error probability for individual $x^{(i)}$'s $\approx\exp(-2^i/r)$ (from Chernoff bounds), while the enumeration of all $x^{(i)}$ becomes $\approx\exp(2^i)$. To bridge this gap, we must sacrifice a factor $r$ in the sampling ratio $p_e$, and correspondingly in the size of the output sparsifier (see proof of Claim~\ref{claim:2}).

The formal proof is far more involved, and can be found in Section~\ref{sec:expander-sparsification}.

\subsection{General Spectral Sparsification of Hypergraphs}\label{sec:tech-overview-general}

Having constructed spectral sparsifiers for good expanders, we move our attention to arbitrary input hypergraphs. We decompose the vertex set of the input hypergraph $G=(V,E)$ into clusters of good expansion, while being careful not to cut too many hyperedges between the clusters.
We adapt well-known techniques to the setting of hypergraphs,
and is detailed for completeness in
\ifprocs
the full version of the paper.
\else
 Section~\ref{sec:expander-decomposition}.
 \fi
As is common for expander decompositions,  we partition $V$ into clusters $C_1,\ldots,C_k$ such that the internal expansion of each cluster (along with its induced hyperedges) is at least $\widetilde\Omega(1/r)$ while cutting only a constant fraction of the hyperedges between the clusters.

In ordinary graphs, this would immediately yield the desired result: We could simply decompose $G$ into expanders and sparsify these, then repeat this process on the discarded hyperedges. Since the number of hyperedges decreases by a constant factor at each level, this process terminates after $O(\log n)$ levels of expander decomposition; each vertex only participates in $O(\log n)$ expanders, and thus the size bound of the overall sparsifier only suffers a logarithmic factor compared to the sparsifiers of expanders. For hypergraphs, this is not the case. Even simple, $r$-uniform hypergraphs may have up to nearly $n^r$ hyperedges. This means that such a decomposition process could require $r\log n$ levels to terminate, introducing another factor $r$ in the size of the sparsifier.

To combat this problem, we contract clusters into individual supernodes after sparsifying them
\ifprocs
(see the full version of the paper).
\else
(see Algorithm~\ref{alg:main}).
\fi
This allows us to simply bound the number of clusters a single vertex can participate in, and consequently the size of the output sparsifier. However, proving the correctness of this more complicated algorithm introduces new challenges.

We denote the contracted version of the input hypergraph $G$  by $G/\approx$, where $u\approx v$ if the two vertices $u$ and $v$ have been contracted into the same supernode.
We can equate between the hyperedges of $G$ and those of $G/\approx$
using the natural bijection between them (this means that a hyperedge $e$ in $G$
refers also to the corresponding hyperedge in $G/\approx$, and vice versa).
Note that this operation can produce multiple parallel hyperdges, as well as \emph{vertices appearing within the same hyperedge with multiplicity}, even if these phenomena were not allowed in the input hypergraph. It is important to note that our expander sparsification algorithm from Section~\ref{sec:expander-sparsification} works equally well in this setting. Furthermore, by allowing hyperedges to contain vetices with multiplicity higher than $1$, we may continue to work with $r$-uniform hypergraphs throughout this process of repeatedly contracting vertices. This technicality is crucial, since our expander decomposition algorithm is designed for this setting, and does not work when hyperedges have different sizes (by more than a constant factor).

The main technical contribution of
\ifprocs
this section
\else
Section~\ref{sec:general-sparsification}
\fi
is to show that a sparsifier computed \emph{after contraction} still sufficiently approximates the energy of the input hypergraph \emph{before contraction}. Here we take a simplified example: Suppose we wish sparsify a cluster $C\subset V$ and subsequently contract it into a supernode $v_C$. At a later level we might wish to sparsify some other cluster $C'$ that contains $v_C$ as one of its vertices (see Figure~\ref{fig:contraction}). The result is a (weighted) subset of hyperedges that well-approximates the spectral structure of $C'$, but will this still be the case when we un-contract $v_C$?

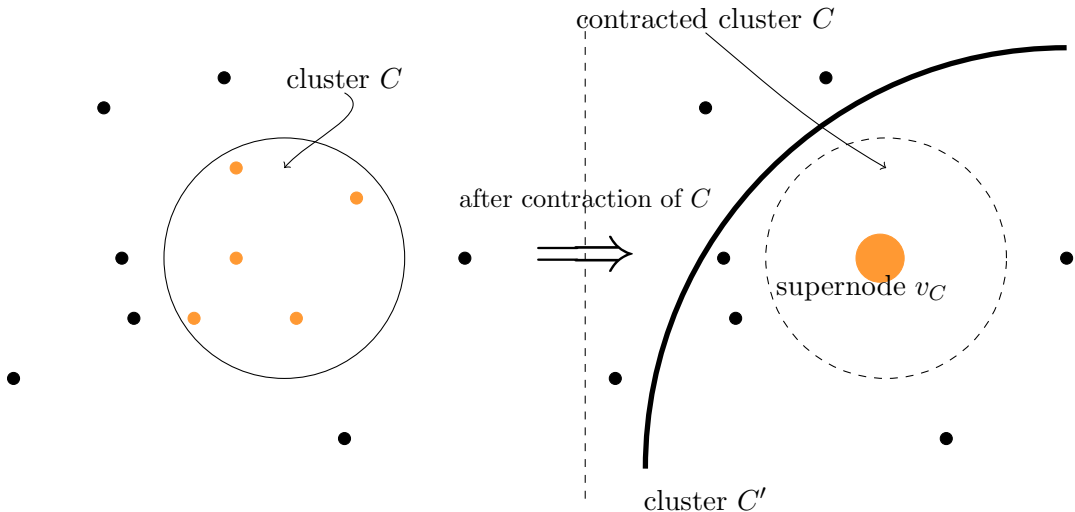
\begin{figure*}[h]
	\begin{center}
		\begin{tikzpicture}[scale=0.8]

		\draw (-5, 0) ellipse (2 and 2);
		\draw[fill] (-7.7, 0) circle (0.1);
		\draw[fill] (-9.5, -2) circle (0.1);
		\draw[fill] (-4, -3) circle (0.1);
		\draw[fill] (-2, 0) circle (0.1);
		\draw[fill] (-7.5, -1) circle (0.1);
		\draw[fill] (-6, +3) circle (0.1);
		\draw[fill] (-8, +2.5) circle (0.1);

		\draw[fill,orange!80] (-5.8, 0) circle (0.1);
		\draw[fill,orange!80] (-3.8, +1) circle (0.1);
		\draw[fill,orange!80] (-4.8, -1) circle (0.1);
		\draw[fill,orange!80] (-5.8, 1.5) circle (0.1);
		\draw[fill,orange!80] (-6.5, -1) circle (0.1);

		\draw (-4, +3) node {cluster $C$};
		\draw[->](-4, 2.75) to[out=-30,in=50] (-5, 1.5);

		\draw[dashed,-] (0, -4) --(0, 4);

		\draw (0, 1) node {\small after contraction of $C$};
		\draw (0, 0) node {\Huge $\Longrightarrow$};

%%%%%%

\begin{scope}[shift={(+10,0)}]
		\draw[dashed] (-5, 0) ellipse (2 and 2);
		\draw[fill] (-7.7, 0) circle (0.1);
		\draw[fill] (-9.5, -2) circle (0.1);
		\draw[fill] (-4, -3) circle (0.1);
		\draw[fill] (-2, 0) circle (0.1);
		\draw[fill] (-7.5, -1) circle (0.1);
		\draw[fill] (-6, +3) circle (0.1);
		\draw[fill] (-8, +2.5) circle (0.1);

		\draw[fill,orange!80] (-5.1, 0) circle (0.4);

		\draw (-8, +4) node {contracted cluster $C$};
		\draw[->](-8, 3.75) to[out=-40,in=150] (-5, 1.5);

		\draw[line width=2pt](-8-1, 0-4+0.5) to[out=90,in=+180] (-1-1, 7-4+0.5);

		\draw (-8, -4) node {cluster $C'$};

		\draw (-5.4, -0.5) node {supernode $v_C$};

\end{scope}
  		\end{tikzpicture}

		\caption{Illustration of the contraction process. Vertices inside $C$ are contracted into a single supernode $v_C$. This is then contained in a later cluster $C'$.}\label{fig:contraction}

	\end{center}

\end{figure*}

Denote the hyperedges of $C'$ by $E'$, and let their sparsifier be $\widetilde E'$ (which is a weighted subset of $E'$). Being a sparsifier with respect to the contracted hypergraph can be viewed as being a sparsifier on the original hypergraph \emph{only when $x\in\RR^V$ is uniform, i.e., takes the same value, on all verices of $C$}, as in this case we can simply assign that same value to $v_C$, and the energy of the original and contracted hypergraphs will be the same. Unfortunately, we have to deal with general vectors $x\in\mathbb R^V$, so we quantify how far $x$ is from satisfing that uniformity requirement.
We consider the maximum discrepancy between the $x$-values of $C$,
defined as $\delta = \max_{u,v\in C}|x_u-x_v|$. We show
\ifprocs
\else
in Section~\ref{sec:general-sparsification-proofs}
\fi
that the additive error introduced by taking $\widetilde E'$ as a sparsifier to $E'$ in the original hypergraph -- as opposed to the contracted hypergraph where it is guaranteed to be a good sparsifier -- is proportional to $\delta^2$ per hyperedge (see the
\ifprocs
full version of the paper).
\else
proof of Claim~\ref{claim:g0}).
\fi

We handle this additive error by arguing that it is dwarfed by energy of $x$ with respect to $C$. On the one hand, we introduce $\delta^2$ error per hyperedge of $C'$ for a total of at most $\approx\delta^2d'n$, where $d'$ is the typical degree in $C'$. On the other hand, we know that the range of $x$ within $C$ is $\delta$, so by hypergraph Cheeger (Theorem~\ref{thm:hypergraph-cheeger}) the energy of $C$ is at least $\approx\delta^2d/r$, where $d$ is the typical degree in $C$. (Here we assume that there are no outlier vertices with extremely low degree, which can be guaranteed by a slight adaptation of the expander decomposition
\ifprocs
subroutine.)
\else
subroutine, Lemma~\ref{lem:partition}.)
\fi Recall that the number of hyperedges --- and therefore the typical degree --- decreases by a constant factor per level. If we can simply guarantee that the sparsifiaction of $C$ precedes the sparsification of $C'$ by at least $\Omega(\log n)$ levels, then $d$ will dwarf $d'$ by an arbitrarily large $n^{\Theta(1)}$ factor. We accomplish this by simply waiting $\Omega(\log n)$ levels to contract a cluster after sparsifying
\ifprocs
it.
\else
it (see Algorithm~\ref{alg:main}).
\fi

The formal proof is far more involved, but relies on the same concept of charging additive errors to previous clusters, until we ultimately achieve the desired overall error of $\epsilon Q(x)$.
The details appear in
\ifprocs
the full version of the paper.
\else
Section~\ref{sec:general-sparsification}.
\fi

\subsection{Lower Bounds}\label{sec:tech-overview-lower-bound}

The most common method for approximating the Laplacian of a (hyper)graph is to take a weighted subset of the original (hyper)edges. While asympotically optimal for graphs~\cite{ACKQWZ16,CKST19}, this method has obvious limitations as a data structure: it is not hard to come up with an example where $\Omega(n)$ hyperedges are required even for the sparsifier to be connected, and if the input hypergraph is $r$-uniform, this translates into $\Omega(nr\log n)$ bit complexity, a linear loss in the arity $r$ of the hypergraph. It is therefore natural to ask whether there are more efficient ways of storing a spectral approximation to a hypergraph. As concrete example, we could permit the inclusion of hyperedges not in the original hypergraph -- could this or another scheme lead to a data structure that can approximate the spectral structure of a hypergraph using $\widetilde O(n)$ space, avoiding a dependence on $r$?

\ifprocs
We
\else
In~Section~\ref{sec:lower-bound}, we
\fi
study this question in full generality:

\begin{center}
\ifprocs
\fbox{
\parbox{0.45\textwidth}{
\begin{center}
Is it possible to compress a hypergraph into a $o(n\cdot r)$ size data structure that can approximate the energy $Q_G(x)$ (defined in~\eqref{eq:hypergraph-energy}) simultaneously for all $x\in {\mathbb R}^V$?
\end{center}
}}
\else
\fbox{
\parbox{0.9\textwidth}{
	\begin{center}
		Is it possible to compress a hypergraph into a $o(n\cdot r)$ size data structure that can approximate the energy $Q_G(x)$ (defined in~\eqref{eq:hypergraph-energy}) simultaneously for all $x\in {\mathbb R}^V$?
	\end{center}
}}
\fi
\end{center}

\ifprocs
We
\else
In Section~\ref{sec:lower-bound}, we
\fi
show a space lower bound of $\Omega(nr)$ for sparsifying a hypergraph on $n$ vertices with maximum hyperedge-size $r$\footnote{With some limits on the range of $r$. For more formal statements of our results see
\ifprocs
the full version of the paper.
\else
Section~\ref{sec:lower-bound-construction}.
\fi}. In fact, our lower bound applies even to the weaker notion of cut sparsification (where one only wants to approximate $Q_G(x)$ for all $x\in \{0, 1\}^V$), and is tight by the recent result of~\cite{Chen20}, who gave a sampling-based cut sparsification algorithm that produces hypergraph sparsifiers with $O(n \log^{O(1)} n)$ hyperedge. In what follows we give an outline of our lower bound.

We start by formally defining the data structure for approximating the cut structure of a hypergraph that we prove a lower bound for. A \emph{hypergraph cut sparsification scheme (HCSS)} is  an algorithm for compressing the cut structure of a hypergraph such that queries on the size of cuts can be answered within a small multiplicative error:
\begin{restatable}{definition}{HCSSdef}
	Let $\mathfrak H(n,r)$ be the set of hypergraphs on a vertex set $[n]$ with each hyperedge having size at most $r$.
A pair of functions $\textsc{Sparsify}:\mathfrak H(n,r)\to\{0,1\}^k$ and $\textsc{Cut}:\{0,1\}^k\times 2^{[n]}\to\mathbb N$ is said to be an $(n,r,k,\varepsilon)$-HCSS if for all inputs $G=(V,E) \in \mathfrak H(n,r)$ the following holds.
	\begin{itemize}
		\item For every query $S\in 2^{[n]}$, $\left|\textsc{Cut}(\textsc{Sparsify}(G),S) - |E(S,\overline S)|\right|\le\varepsilon\cdot|E(S,\overline S)|$.
	\end{itemize}
\end{restatable}

To argue a lower bound on the space requirement (parameter $k$ above), we use a reduction to string compression. It is known that $\{0,1\}$-strings of length $\ell$ cannot be significantly compressed to a small space data structure that allows even extremely crude additive approximations to subset sum queries --- see, e.g., the LP decoding paper of~\cite{DinurN03} (here we only need a lower bound for computationally unbounded adversaries), or
\ifprocs
the full version of this paper.
\else
Section~\ref{sec:lower-bound-string}.
\fi
We manage to encode a $\{0,1\}$-string of length $\ell$ into the cut structure of a hypergraph $H$ with \emph{fewer hyperedges than $\ell$} --- a testament to the higher complexity of hypergraph cut structures, as opposed to the cut structures of ordinary graphs.

Our string encoding construction utilizes Ruzsa-Szemer\'edi graphs. These are (ordinary) graphs whose edge-sets are the union of \textit{induced} matchings. Our construction works generally on any Ruzsa-Szemer\'edi  graphs and as a result we get several lower bounds in various parameter regimes (values of the hyperedge arity $r$ and the precision parameter $\epsilon$) based on the specific Ruzsa-Szemer\'edi  graph constructions we choose to utilize. In particular, for the setting where $r=n^{O(1/\log\log n)}$ we are able to conclude that any hypergraph cut sparsification scheme requires $\Omega(rn)$ bits of space even for constant $\epsilon$, matching the upper bound of~\cite{Chen20} to within logarithmic factors. For larger $r$ we get a lower bound of $n^{1-o(1)} r$ bits of space for $\epsilon=n^{-o(1)}$. The latter in particular rules out the possibility of an $\epsilon$-sparsifier that can be described with asymptotically fewer than $(\epsilon^{-1})^{O(1)} nr$ bits of space.

Here we briefly describe how we encode strings into hypergraphs generated from Ruzsa-Szemer\'edi graphs. Let $G$ be a bipartite Ruzsa-Szemer\'edi-graph (with bipartition $P\cup Q$) composed of $t$ induced matchings of size $a$ each. We can then use the $a\cdot t$ edges of the graph to encode a string $s$ of length $\ell=at$: simply order the edges of $G$ and remove any edges corresponding to $0$ coordinates in $s$, while keeping edges corresponding to $1$'s. This graph --- which we call $G_s$ --- already encodes $s$ when taken as a whole. However, its cut structure is not sufficient for decoding it. For that we need to turn $G_s$ into a hypergraph $H_s$ as follows: For each vertex $u$ on one side of the bipartition, say $P$, we combine all edges adjacent on $u$ into one hyperedge containing $\{u\}\cup\Gamma(u)$. This means that each hyperedge will have only a single vertex in $P$, but many vertices in $Q$ (see Figure~\ref{fig:lower-bounds}).

To decode the original string $s$ from the cut structure of $H$, we must be able to answer subset sum queries $q\subseteq[at]$, that is return how many $1$-coordinates $s$ has, restricted to $q$. (For more details see
\ifprocs
the full version of the paper.)
\else
the definition of string compression -- Definition~\ref{def:SCS} in Section~\ref{sec:lower-bound-string}.)
\fi
To do this, consider each induced matching one at a time and decode $s$ restricted to the corresponding coordinates. We measure the size of a carefully chosen cut in $H_s$. Consider Figure~\ref{fig:lower-bounds}: We restrict our view to a single matching $M_j$ supported on $P_j$ and $Q_j$ in the two sides of the bipartition. Suppose for simplicity that $q$ is entirely contained in this matching, and we are interested in the Hamming-weight of $s$ restricted to a subset of coordinates $q$. To create our cut, in the top half of the hypergraph ($P$), we take the endpoints of edges corresponding to $q$ -- we call this set $A$. In the bottom half ($Q$), we take everything except for $Q_j$. The cut, which we call $S$, is depicted in red in Figure~\ref{fig:lower-bounds}.

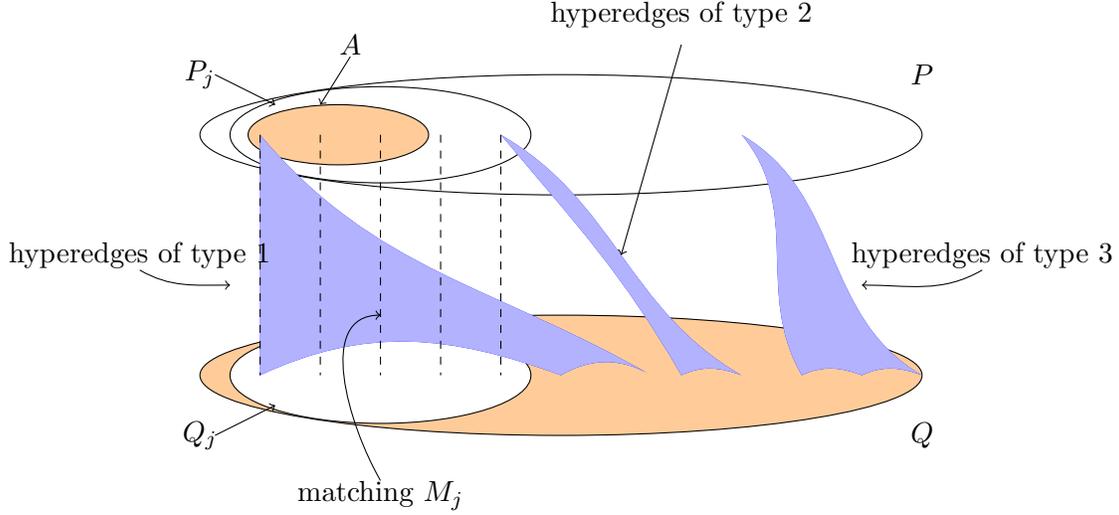
\begin{figure*}[h]
	\begin{center}
		\begin{tikzpicture}[scale=0.8]

		\draw (0, 0) ellipse (6 and 1);
		\draw[fill = orange!40] (0, -4) ellipse (6 and 1);

		\draw (-3, 0) ellipse (2.5 and 0.8);
		\draw[fill = white] (-3, -4) ellipse (2.5 and 0.8);

		\draw[fill = orange!40] (-4+0.3, 0) ellipse (1.5 and 0.5);

		\fill[blue!30,preaction={fill=blue}]  (-5, 0) to[out=-90,in=90] (-5, -4) to [out=+25,in=+160]  (+0, -4)
		     to[out=+30,in=+150] (+1.5, -4) to[out=+150,in=-50] (-5, 0);

		\fill[blue!30,preaction={fill=blue}]  (-1, 0) to[out=-50,in=120] (+2, -4) to [out=+25,in=+160]  (+3, -4) to[out=+150,in=-30] (-1, 0);

		\fill[blue!30,preaction={fill=blue}]  (+3, 0) to[out=-50,in=120] (+4, -4) to [out=+25,in=+160]  (+5, -4) to [out=+25,in=+160]  (+6, -4) to[out=+150,in=-30] (+3, 0);

		\draw (+6, 1) node {$P$};
		\draw (+6, -5) node {$Q$};

		\draw (-6, 1) node {$P_j$};
		\draw (-6, -5) node {$Q_j$};
		\draw[->] (-5.75, 1) --(-4.75, 0.5);
		\draw[->] (-5.75, -5) --(-4.75, -4-0.5);

		\draw (-3.5, 1.5) node {$A$};
		\draw[->] (-3.5, 1.3) --(-4, 0.5);

		\draw (-7, -2) node {hyperedges of type 1};
		\draw[->](-7, -2.25) to[out=-30,in=+180] (-5.5, -2.5);

		\draw (+2, +2) node {hyperedges of type 2};
		\draw[->] (+2, 1.5) --(+1, -2);

		\draw (+7, -2) node {hyperedges of type 3};
		\draw[->](+7, -2.25) to[out=210,in=0] (+5.0, -2.5);

		\draw[dashed, name path=p1] (-5, 0) -- (-5, -4);
		\draw [dashed, name path=p2](-4, 0) -- (-4, -4);
		\draw [dashed, name path=p3] (-3, 0) -- (-3, -4);
		\draw [dashed, name path=p4] (-2, 0) -- (-2, -4);
		\draw [dashed, name path=p5] (-1, 0) -- (-1, -4);

		\draw (-3, -6) node {matching $M_j$};
		\draw[->](-3, -5.75) to[out=120,in=+180] (-3, -3);

		\end{tikzpicture}

		\caption{Illustration of the decoding process. One side of the cut $S$ is depicted in orange.}\label{fig:lower-bounds}
	\end{center}

\end{figure*}

Informally, the crux of the decoding is the observation that the number of hyperedges crossing from $A$ to $Q_j$ is exactly the quantity we want to approximate. Indeed, consider a coordinate in $q$. If it has value $1$ in $s$, the corresponding hyperedge crosses from $A$ to $Q_j$, thus crossing the cut $S$. If however this coordinate is $0$ in $s$, the corresponding hyperedge does not cross to $Q_j$, thus not crossing the cut. These types of hyperedges are denoted by $1$ in Figure~\ref{fig:lower-bounds}.

Unfortunately, there are more hyperedges crossing $S$, adding noise to our measurement of $s$. One might hope to prove that the noise is small, i.e., can be attributed to measurement error, but this is not the case. Instead, we show that while this noise is not small, it is predictable enough to subtract accurately without knowing $s$. Hyperedges denoted $2$ in Figure~\ref{fig:lower-bounds} cross from $P_j\setminus A$ to $Q\setminus Q_j$. Here we observe that nearly all hyperedges from $P_j\setminus A$ do in fact cross the cut, for almost all choices of $s$. Hyperedges denoted $3$ in Figure~\ref{fig:lower-bounds} cross from $P\setminus P_j$ to $Q\setminus Q_j$. Here we cannot say much about the quantity of such hyperedges crossing the cut. However, we observe that this quantity does not depend on $q$, and therefore we can use Chernoff bounds (Theorem~\ref{thm:chernoff}) to prove that it concentrates around its expectation with high probability \textit{over $s$}. This allows us to predict and subtract the noise caused by type $3$ hyperedges, for whatever instance of Ruzsa-Szemer\'edi-graph we use
\ifprocs
.
\else
(see the proof of Theorem~\ref{thm:lower-bound-construction}).
\fi

Ultimately, we show that efficient cut sparsification for such hypergraphs would result in an equally efficient compression of $\{0,1\}$-strings, which implies our lower bounds. For more details see
\ifprocs
the full version of the paper.
\else
Section~\ref{sec:lower-bound}.
\fi

\subsection{Directed Spectral Sparsification of Hypergraphs}

\ifprocs
In the full version of the paper, we also apply
\else
In Section~\ref{sec:directed-hypergraph-sparsification}, we apply
\fi
our discretization technique from Section~\ref{sec:expander-sparsification} to the spectral sparsification of directed hypergraphs. As a testiment to the versitility of this technique, we are able to produce an $O(n^2r^3\log^2n/\epsilon^2)$-sized $\epsilon$-spectral sparsifier. This is a factor $n$ better than the previous state of the art by~\cite{Soma2019}, and nearly optimal in the setting where $r$ is constant.

The broad arc of the proof is very similar to that of Section~\ref{sec:expander-sparsification}: We construct our sparsifier using importance sampling. We then divide the set of hyperarcs into a logarithmic number of categories, $E_i$. For each category separately, we show using discretization that the energy of the proposed sparsifier approximates the energy of the input hypergraph with respect to \emph{all} $x\in\mathbb R^V$ simultaneously with high probability.

However, the details of each of these steps differ from their corresponding step in Section~\ref{sec:expander-sparsification}. Here we mention only a few key differences. Instead of looking at degrees or expansion, we define a novel quantity characterizing each hyperarc we call its \emph{overlap}. Intuitively, this denotes the highest density of an induced subgraph in which the paericular hyperarc resides. We then sample each hyperarc with probability inverse proportional to its overlap. We show that this produces a sufficiently small sparsifier with high probability
\ifprocs
.
\else
(see Lemma~\ref{lem:overlap}).
\fi

Perhaps the most crucial departure from Section~\ref{sec:expander-sparsification} occurs during the discretization step when proving $Q_x(E_i)=\widetilde Q_x(E_i)$. Instead of discretizing the vector $x\in\mathbb R^V$, we discretize the derived vector of energies on the hyperarcs, that is $Q_x\in\mathbb R^E$. So for each $x$ and $i$ we define a vector $Q_x^{(i)}$ --- from a finite set of possibilities --- such that, informally
$$Q_x(E_i)\cong Q_x^{(i)}(E_i)\cong\widetilde Q_x^{(i)}(E_i)\cong\widetilde Q_x(E_i).$$
\ifprocs
\else
For more details on the definition of $Q_x^{(i)}$, see the proof of Lemma~\ref{lem:directed-main}.
\fi
This additional trick is necessary; we do not know of a way to make the discretization argument work by rounding $x$ itself.

For more details on the construction of directed hypergraph sparsifiers and their analysis see
\ifprocs
the full version of the paper.
\else
Section~\ref{sec:directed-hypergraph-sparsification}.
\fi

%%% Local Variables:
%%% mode: latex
%%% TeX-master: "000-main"
%%% End:

%!TEX root=./000-main.tex

\section{Spectral Sparsification of Expanders}\label{sec:expander-sparsification}

In this section, we prove the following.
\begin{theorem}\label{thm:expander-sparsification}
  There is an algorithm that, given a parameter $n$, given $100/n\le\epsilon\le 1/2$ and an $r$-uniform hypergraph $G=(V,E)$ with $|V|\le n$ and expansion at least $350\sqrt{(\log n)/(\epsilon rn)}\le\Phi\le2/r$, outputs an $\epsilon$-spectral sparsifier of $G$ with $O(|V|\cdot(\tfrac1\epsilon \log n)^{O(1)} /(\Phi^2r))$ hyperedges with probability $1-O((\log n)/n^2)$ in $O(r|E|)$ time.

\end{theorem}

\begin{remark}
	  Note that $n$ here does not denote the size of $V$ but an arbitrary parameter larger than that. $n$ serves only as an indirect error parameter, as the failure probability of the algorithm is allowed to be $1-O((\log n)/n^2)$. The reason for this notation is that later on,
	  \ifprocs
	  \else
	  in Section~\ref{sec:general-sparsification},
	  \fi
	  we apply Theorem~\ref{thm:expander-sparsification} to subgraphs of the input hypergraph. In this context, $n$ will denote the the size of the input hypergraph, whereas $|V|$ will denote the (potentially much smaller) size of the cluster within it, to be sparsified. Note that the size of the sparsifier scales linearly in the size of the cluster, but only logarithmically in the size of the input hypergraph. The latter is because the desired failure probability is always defined in terms of $n$.
	  
	  \ifprocs
	  For more details on this, see the full version of the paper.
	  \fi
\end{remark}

\begin{remark}
The guarantee of Theorem~\ref{thm:expander-sparsification} translates to $|\widetilde E| = \widetilde O(|V|r)$ when $\Phi(G) = \Omega(1/r)$, i.e., when $G$ is a nearly-optimal expander.
\end{remark}
We show our construction of the sparsifier in Section~\ref{sec:expander-sparsification-construction} and discuss its correctness in Section~\ref{sec:expander-sparsification-correctness}, where some proofs are deferred to Section~\ref{sec:expander-sparsification-claims}.

The following lemma is useful throughout this section.
\begin{lemma}\label{lem:sum-min-d}
  For any hypergraph $G=(V, E)$, we have
  \[
    \sum_{e\in E}\frac1{\min_{v\in e}d(v)}\le|V|.
  \]
\end{lemma}
\begin{proof}
  Consider each hyperedge $e \in E$ to be directed towards its vertex with the lowest degree, i.e., $\argmin_{v\in e}d(v)$, breaking ties arbitrarily.
  Each vertex $v \in V$ has at most $d(v)$ incoming hyperedges,
  and each such hyperedge contributes to the above sum by $1/d(v)$.
  Hence the total contribution of all the incoming hyperedges to $v$ is at most $1$.
  It follows that the overall summation is at most $|V|$.
\end{proof}

\subsection{Construction}\label{sec:expander-sparsification-construction}
The construction of $\widetilde G$ is quite simple.
Sample each hyperedge $e \in E$ with probability $p_e=\min\left(\frac{\lambda}{\min_{v\in e}d(v)},1\right)$
for
\begin{equation}\label{eq:def-lambda}
\lambda = (\epsilon^{-1}\log n)^{O(1)}/(\Phi^2r).
\end{equation}
Each sampled hyperedge $e$ is given weight $w_e=1/p_e$,
and for every non-sampled hyperedge $e$ define $w_e=0$.
Let $\widetilde G$ contain the sampled hyperedges,
i.e., $\widetilde E = \set{e\in E \mid w_e>0 }$.
Notice that each random variable $w_e$ has expectation $\EX[w_e]=1$,
and thus informally $\mathbb E[\widetilde G]=G$.

Clearly we can compute the output in time $O(r|E|)$.
Also, we can bound the size of the sparsifier with high probability as follows.
\begin{lemma}\label{lem:size-bound}
  We have
  \[
    \mathbb P[|\widetilde E|\ge 2\lambda|V|]
    \le O(1/n^2),
  \]
  when $|\widetilde E|=\Omega(\log n)$.
\end{lemma}
\begin{proof}
  First, we have
  \[
    \EX[|\widetilde E|]
    \le \sum_{e\in E} p_e
    \le \lambda\sum_{e\in E}\frac1{\min_{v\in e}d(v)}
    \le \lambda|V|,
  \]
  where the last inequality is due to Lemma~\ref{lem:sum-min-d}.
  Noting that $|\widetilde{E}|$ is a sum of independent indicator random variables, the claimed inequality is a direct consequence of the Chernoff bound (Theorem~\ref{thm:chernoff}).
\end{proof}

\subsection{Correctness}\label{sec:expander-sparsification-correctness}
Let us now consider the spectral properties of $\widetilde G$. We must prove that with high probability
\begin{equation}\label{eq:main-spectral}
  \forall x\in\mathbb R^V,
  \quad
  \widetilde Q(x)=(1\pm\epsilon)\cdot Q(x).
\end{equation}
We stress that this gives an error bound that holds for all $x$ \textit{simultaneously}.
We may assume without loss of generality that $\sum_{v\in V}x_vd(v)=0$ and $\sum_{v\in V}x_v^2d(v)=1$, because Equation~\eqref{eq:main-spectral} is invariant under translation and scaling of $x$. Let the set of such centered and normalized vectors be $\overline{\mathbb R^V}$.
This guarantees that every non-isolated vertex $v$ has
$x_v^2 \leq 1/d(v) \leq 1$,
and by Theorem~\ref{thm:hypergraph-cheeger} we get $Q(x)\ge\frac{r\Phi^2}{32}$.

Now fix one such vector $x\in\overline{\mathbb R^V}$,
and use it to partition the hyperedge multiset $E$ into $O(\log n)$ subsets as follows.
For each $i=1,\ldots,i^*$, where $i^*=\lceil2\log n\rceil$,
let
\begin{align*}
  E_i
  &= \Big\{e\in E\mid \max_{v\in e}x_v^2\cdot \min_{v\in e}d(v)\in(2^{-i},2^{-i+1}] \Big\},
    \intertext{and let}
    E_*
  &= E\setminus \bigcup_{i=1}^{i^*} E_i
   = \Big\{e\in E\mid\max_{v\in e}x_v^2\cdot\min_{v\in e}d(v)\le2^{-i^*}\Big\} .
\end{align*}
To justify the second equality in the equation above, note that  $\sum_{v\in V}x_v^2d(v)=1$ implies $x(v)^2\leq 1/d(v)$, and therefore for every $e\in E$
$$
\max_{v\in e}x_v^2\cdot \min_{v\in e}d(v)\leq \max_{v\in e} 1/d(v)\cdot  \min_{v\in e}d(v)=1.
$$

Informally, we would like to show that with high probability,
for all $x$ and all $i$ we have $\widetilde Q_x(E_i)\cong Q_x(E_i)$. Note that the multisets $E_i$ and $E_*$ are dependent on $x$, but we omit this from the notation for better readability.
Our plan is to define another vector $x^{(i)}\in\mathbb R^V$
by rounding the coordinates of $x$,
that preserves $Q(E_i)$ up to small multiplicative and additive error.
Using this rounded vector, we will then show
$$\widetilde Q_x(E_i)\cong\widetilde Q_{x^{(i)}}(E_i)\cong Q_{x^{(i)}}(E_i)\cong Q_x(E_i),$$
and similarly also $\widetilde Q_x(E_*)\cong Q_x(E_*)$.

Formally, for each $v \in V$ define $x^{(i)}_v$ as follows:
\begin{itemize}
    \item If $x_v^2 d(v)\ge\epsilon^2 2^{-i}/2500$, then round $x_v$ to the nearest integer multiple of $1/(n^2\sqrt{d(v)})$.
    \item If $x_v^2d(v)<\epsilon^2 2^{-i}/2500$, then round $x_v$ to $0$.
\end{itemize}
We implement the above plan using the following four claims.

First, we show that for every scale $i$ the energy of $E_i$ with respect to the rounded vector $x^{(i)}$ is quite close to the energy of $E_i$ with respect to the original vector $x$:
\begin{restatable}{claim}{claimone}\label{claim:1}
For all $x \in \overline{\mathbb R^V}$ and all $i=1,\ldots,i^*$,
$$
  Q_{x^{(i)}}(E_i)
  = \Big(1\pm\frac{\epsilon}{10}\Big) Q_x(E_i)\pm \frac{20}{n}.
$$
\end{restatable}

Next, we show that for every scale $i$ our sampling process preserves energy of $E_i$ on rounded version of all $x$ simultaneously:
\begin{restatable}{claim}{claimtwo}\label{claim:2}
For all $i=1,\ldots,i^*$,
$$\mathbb P\left[\forall x\in\overline{\mathbb R^V},\ \widetilde Q_{x^{(i)}}(E_i)
  = \Big(1\pm\frac{\epsilon}{10}\Big) Q_{x^{(i)}}(E_i)\pm\frac{\epsilon Q(x)}{10\log n}\right]
  \ge 1-\frac{1}{n^2}.
$$
\end{restatable}
We then relate the energy of the sampled $E_i$  on rounded versions of $x$ to the corresponding energy on original $x$:
\begin{restatable}{claim}{claimthree}\label{claim:3}
For all $i=1,\ldots,i^*$,
$$\mathbb P\left[\forall x \in \overline{\mathbb R^V},\ \widetilde Q_{x^{(i)}}(E_i)
  = \Big(1\pm\frac{\epsilon}{10}\Big) \widetilde Q_x(E_i)\pm \frac{60}{n}\right]
  \ge 1-\frac{1}{n^2}.$$
\end{restatable}
Finally we bound the error introduced on the hyperedges of $E_*$.
\begin{restatable}{claim}{claimfour}\label{claim:4}
$$
  \mathbb P\left[\forall x\in\overline{\mathbb R^V},\ \widetilde Q_x(E_*)
  = Q_x(E_*)\pm\frac{12}n\right]
  \ge  1-\frac{1}{n^2}.
$$
\end{restatable}

Before proving these claims,
which we do in the next section,
let us show how to use them to show the correctness of the sparsifier.
\begin{lemma}\label{lem:spectral-property}
  The hypergraph $\widetilde{G}$ is an $\epsilon$-spectral sparsifier of $G$ with probability $1-O((\log n)/n^2)$.
\end{lemma}
\begin{proof}
Assume henceforth that the events in Claims~\ref{claim:2},~\ref{claim:3}, and~\ref{claim:4} all hold simultaneously for every $i$ ---
we know that this happens with probability $1-O(\log n/n^2)$ ---
and let us compare $Q_x(E_i)$ with $\widetilde Q_x(E_i)$ for each $i$.
If the above claims had no additive error, we could conclude that
$\widetilde Q_x(E_i) = (1\pm{4\epsilon}/{10}) Q_x(E_i)$.
Similarly, if they had no multiplicative error, we could conclude that
$
  |\widetilde Q_x(E_i) - Q_x(E_i)|
  \leq \frac{80}{n} + \frac{\epsilon Q(x)}{10\log n}
$;
we could then use the assumed lower bound on $\Phi$
to bound $\frac{80}{n} \leq \frac{\epsilon Q(x)}{10\log n}$,
and sum up these additive errors over all $i=1,\ldots,i^*$
to a total that is bounded by $\frac{4}{10}\epsilon Q(x)$.
These arguments extend easily also to $E_*$.

For the formal calculation, consider first one direction,
\ifprocs
\begin{align*}
&\widetilde Q_x(E_i)\\
&\le{\left(1 - \frac{\epsilon}{10}\right)}^{-1}
\left[\frac{60}{n} +  \widetilde Q_{x^{(i)}}(E_i)\right]
%& \text{By Claim~\ref{claim:3}}
\\
&\le{\left(1 - \frac{\epsilon}{10}\right)}^{-1}
\left[ \frac{60}{n}
+ \frac{\epsilon Q(x)}{10\log n} + \left(1+\frac{\epsilon}{10}\right) Q_{x^{(i)}}(E_i) \right]
%& \text{By Claim~\ref{claim:2}}
\\
&\le{\left(1 - \frac{\epsilon}{10}\right)}^{-1}
\left[\frac{60}{n}
+  \frac{\epsilon Q(x)}{10\log n} + \left(1+\frac{\epsilon}{10}\right) \left[ \frac{20}{n} + \left(1+\frac{\epsilon}{10}\right) Q_x(E_i) \right] \right]
% & \text{By Claim~\ref{claim:1}}
\\
&\le\frac{120}{n} + \frac{2\epsilon Q(x)}{10\log n}
+ \left(1+\frac{4\epsilon}{10}\right) Q_x(E_i).
\end{align*}
Here the second line follows by Claim~\ref{claim:3}, the third line follows by Claim~\ref{claim:2}, the fourth line follows by Claim~\ref{claim:1}, and the last line follows since $\epsilon\le1/2$.
\else
\begin{align*}
\widetilde Q_x(E_i)
&\le {\left(1 - \frac{\epsilon}{10}\right)}^{-1}
\left[\frac{60}{n} +  \widetilde Q_{x^{(i)}}(E_i)\right]
& \text{By Claim~\ref{claim:3}}
\\
&\le {\left(1 - \frac{\epsilon}{10}\right)}^{-1}
\left[ \frac{60}{n}
+ \frac{\epsilon Q(x)}{10\log n} + \left(1+\frac{\epsilon}{10}\right) Q_{x^{(i)}}(E_i) \right]
& \text{By Claim~\ref{claim:2}}
\\
&\le {\left(1 - \frac{\epsilon}{10}\right)}^{-1}
\left[\frac{60}{n}
+  \frac{\epsilon Q(x)}{10\log n} + \left(1+\frac{\epsilon}{10}\right) \left[ \frac{20}{n} + \left(1+\frac{\epsilon}{10}\right) Q_x(E_i) \right] \right]
& \text{By Claim~\ref{claim:1}}
\\
&\le \frac{120}{n} + \frac{2\epsilon Q(x)}{10\log n}
+ \left(1+\frac{4\epsilon}{10}\right) Q_x(E_i) ,
\end{align*}\\
since $\epsilon\le1/2$.
\fi
Now sum this over all $i$ and combine it
with the bound from Claim~\ref{claim:4} on the error introduced by $E_*$,
to get
\begin{align*}
  \widetilde Q(x)
  &=   \widetilde Q_x(E_*) + \sum_{i=1}^{i^*} \widetilde Q_x(E_i) \\
  &\le \left[ Q_x(E_*) + \frac{12}{n}\right]
    + i^* \left[ \frac{120}{n} + \frac{2\epsilon Q(x)}{10\log n} \right]
    + \left(1+\frac{4\epsilon}{10}\right) \sum_{i=1}^{i^*} Q_x(E_i)
  \\
  &\le \frac{250\log n}{n} + \frac{5\epsilon Q(x)}{10} + \left(1+\frac{4\epsilon}{10}\right) Q(x)
  \\
  &\le (1+\epsilon) Q(x).
\end{align*}
The last inequality follows from $Q(x)\ge\Phi^2r/32$ (by Theorem~\ref{thm:hypergraph-cheeger}), and the theorem's assumption that $\Phi\ge350\sqrt{(\log n)/(\epsilon rn)}$.

The other direction,
$\widetilde Q(x) \ge (1-\epsilon) Q(x)$, follows similarly.
\end{proof}

Theorem~\ref{thm:expander-sparsification} then follows by Lemmas~\ref{lem:size-bound} and~\ref{lem:spectral-property} and a union bound.

\subsection{Proofs of Claims~\ref{claim:1},~\ref{claim:2},~\ref{claim:3}, and~\ref{claim:4}}\label{sec:expander-sparsification-claims}
We begin by presenting a preliminary lemma about the effects of approximating $x$ on a general quadratic form, which will be useful in proving the four
\ifprocs
claims.
\else
claims, and will be useful later on in Section~\ref{sec:general-sparsification}.
\fi

\begin{lemma}\label{lem:additive-error}
Let $G=(V,E)$ be a hypergraph and let $x, \widetilde x$ be two vectors in $\mathbb R^V$ such that $|x_v-\widetilde x_v|\le\delta$ on every coordinate $v\in V$ for some $\delta\ge0$.
Then for any $e \in E$,
\[
  \big|Q_x(e)-Q_{\widetilde x}(e)\big|
  \le 4\delta\left(\sqrt{Q_x(e)}+\delta\right).
\]
\end{lemma}

\begin{proof}
Given $e\in E$, we begin by finding two vertices $u^*,v^*\in e$ such that
$$\big|Q_x(e)-Q_{\widetilde x}(e)\big| \le \big|(x_{u^*}-x_{v^*})^2-(\widetilde x_{u^*}-\widetilde x_{v^*})^2 \big|.$$
It is indeed possible to find such vertices.
If $Q_x(e)\ge Q_{\widetilde x}(e)$, set $u^*,v^*$ such that $Q_x(e) = (x_{u^*}-x_{v^*})^2$,
and we get
\[
  \big|Q_x(e)-Q_{\widetilde x}(e)\big|
  = Q_x(e)-Q_{\widetilde x}(e)
  \le (x_{u^*}-x_{v^*})^2- (\widetilde x_{u^*}-\widetilde x_{v^*})^2,
\]
since $Q_{\widetilde x}(e) \geq (\widetilde x_u-\widetilde x_v)^2$ for every $u, v\in e$, and in particular $Q_{\widetilde x}(e) \geq (\widetilde x_{u^*}-\widetilde x_{v^*})^2$. Otherwise, i.e., $Q_{\widetilde x} > Q_x(e)$, similarly set $u^*,v^*$ such that $Q_{\widetilde x}(e) = (\widetilde x_{u^*}-\widetilde x_{v^*})^2 $.

Using these $u^*,v^*$, we have
\begin{align*}
  \big|Q_x(e)-Q_{\widetilde x}(e)\big|
  &\le\big|(x_{u^*}-x_{v^*})^2-(\widetilde x_{u^*}-\widetilde x_{v^*})^2\big|
  \\
  &= |x_{u^*}+\widetilde x_{u^*}-x_{v^*}-\widetilde x_{v^*}|
    \cdot|x_{u^*}-\widetilde x_{u^*}-x_{v^*}+\widetilde x_{v^*}|
\end{align*}
Let us now bound each of these two factors.
The second one is clearly bounded by $2\delta$ by the lemma's assumption.
To bound the first term,
we use that $Q_x(e)=\max_{u,v\in e}(x_u-x_v)^2\ge(x_{u^*}-x_{v^*})^2$,
and therefore
\ifprocs
\begin{align*}
|x_{u^*}+\widetilde x_{u^*}-x_{v^*}-\widetilde x_{v^*}|&\le2|x_{u^*}-x_{v^*}|+|x_{u^*}-\widetilde x_{u^*}|+|x_{v^*}-\widetilde x_{v^*}|\\
&\le2\sqrt{Q_x(e)}+2\delta.
\end{align*}
\else
\begin{align*}
	|x_{u^*}+\widetilde x_{u^*}-x_{v^*}-\widetilde x_{v^*}|&\le2|x_{u^*}-x_{v^*}|+|x_{u^*}-\widetilde x_{u^*}|+|x_{v^*}-\widetilde x_{v^*}|
	\le2\sqrt{Q_x(e)}+2\delta.
\end{align*}
\fi
Putting these two bounds together, we obtain the result of the lemma.
\end{proof}

The following claim examines the effects of rounding from $x$ to $x^{(i)}$ (from the previous section) on a single hyperedge of $E_i$. This is the main technical claim that allows as to then easily prove both Claims~\ref{claim:1} and~\ref{claim:3}.

\begin{claim}\label{claim:0}
For all $x \in \mathbb{R}^V$, all $i=1,\ldots,i^*$, and every hyperedge $e\in E_i$,
$$
  Q_{x^{(i)}}(e)
  = \Big(1\pm\frac{\epsilon}{10}\Big) Q_x(e)\pm\frac{20}{n^2\min_{v\in e}d(v)}.
$$
\end{claim}

\begin{proof}% [Proof of Lemma~\ref{lemma:0}]
We examine the difference $\big|Q_x(e)-Q_{x^{(i)}}(e)\big|$. Recall that $x^{(i)}_v$ is either rounded to the nearest multiple of $1/(n^2\sqrt{d(v)})$ or rounded to zero.
We consider two cases:
	\begin{enumerate}
		\item No vertex in $e$ is rounded to zero.
		\item At least one vertex in $e$ is rounded to zero.
	\end{enumerate}
For simplicity, denote $x_+ = \max_{v\in e}|x_v|$ and $d_-=\min_{v\in e}d(v)$. Recall that by definition of $E_i$,
	\begin{equation}\label{eq:Ei}
	x_+^2d_-\in\left(2^{-i},2^{-i+1}\right].
	\end{equation}

	In the first case, the value of $x$ on every vertex $v\in e$ changes by at most $1/(n^2\sqrt{d(v)})\le1/(n^2\sqrt{d_-})$. Thus, we can apply Lemma~\ref{lem:additive-error} with $\delta = 1/(n^2\sqrt{d_-})$ to get
        \begin{align*}
\big|Q_x(e)-Q_{x^{(i)}}(e)\big|&\le\frac{4}{n^2\sqrt{d_-}}\cdot\left(\sqrt{Q_x(e)}+\frac{1}{n^2\sqrt{d_-}}\right).
	\end{align*}
	We can use~\eqref{eq:Ei} to bound $Q_x(e)\le 4x_+^2 \le 4\cdot2^{1-i}/d_-\le4/d_-$. Substituting this in, we get
	\begin{align*}
	\big|Q_x(e)-Q_{x^{(i)}}(e)\big|&\le\frac{4}{n^2\sqrt{d_-}}\cdot\left(\frac{2}{\sqrt{d_-}}+\frac{1}{n^2\sqrt{d_-}}\right)
	\le\frac{20}{n^2d_-}.
	\end{align*}

In the second case, the value of $x$ on a vertex in $e$ can still change by at most $1/(n^2\sqrt{d_-})$ by rounding to a non-zero value.
It can additionally be rounded to a zero,
as long as $x_v^2 d(v) < \epsilon^2 2^{-i}/2500$,
which amounts to additive error per coordinate of at most
$|x_v| < \epsilon/\sqrt{2500\cdot2^i d(v)} \le \epsilon/\sqrt{2500\cdot 2^i d_-}$.
Therefore we can apply Lemma~\ref{lem:additive-error} with
$\delta = \epsilon/\sqrt{2500\cdot2^i d_-} \geq 1/(n^2\sqrt{d_-})$,
which gives us that
$$\big|Q_x(e)-Q_{x^{(i)}}(e)\big|\le\frac{4\epsilon}{\sqrt{2500\cdot2^i d_-}}\cdot\left(\sqrt{Q_x(e)}+\frac{\epsilon}{\sqrt{2500\cdot2^i d_-}}\right).$$

This time we use a lower bound on $Q_x(e)$. Recall that we assumed that at least one vertex in $e$ is rounded to zero. Let one such vertex be $v_0$. This means that $x_{v_0}^2 d(v_0) \le \epsilon^2 2^{-i}/2500$,
but at the same time $x_+^2 d(v_0) \ge x_+^2 d_- \ge 2^{-i}$.
Using these two facts, we get our lower bound
\begin{align*}
  \sqrt{Q_x(e)}&\ge x_+-|x_{v_0}|
  \ge x_+-\epsilon x_+/50
  \ge\frac{49}{50\sqrt{2^i d_-}}.
\end{align*}
Substituting this in, we get
\ifprocs
\begin{align*}
\big|Q_x(e)-Q_{x^{(i)}}(e)\big|
&\le	\frac{4\epsilon\sqrt{Q_x(e)}}{49} \cdot\left(\sqrt{Q_x(e)}+\frac{\epsilon\sqrt{Q_x(e)}}{49} \right)\\
&\le\frac{4\epsilon Q_x(e)}{49}\cdot\left(1+\frac{\epsilon}{49}\right)\\
&\le\frac{\epsilon}{10} Q_x(e).
\end{align*}
\else
\begin{align*}
    \big|Q_x(e)-Q_{x^{(i)}}(e)\big|
    &\le	\frac{4\epsilon\sqrt{Q_x(e)}}{49} \cdot\left(\sqrt{Q_x(e)}+\frac{\epsilon\sqrt{Q_x(e)}}{49} \right)
\le\frac{4\epsilon Q_x(e)}{49}\cdot\left(1+\frac{\epsilon}{49}\right)
\le\frac{\epsilon}{10} Q_x(e).
\end{align*}
\fi
	In conclusion, in the first case we get the claimed additive error, while in the second case we get the claimed multiplicative error.
\end{proof}

We are now ready to proceed to proving Claims~\ref{claim:1} and~\ref{claim:3}.

\claimone*
\begin{proof}%[Proof of Claim~\ref{claim:1}]
We can bound
\ifprocs
\begin{align*}
\big|Q_x(E_i)-Q_{x^{(i)}}(E_i)\big|
&\le\sum_{e\in E_i}\big|Q_x(e)-Q_{x^{(i)}}(e)\big| \\
&\le\sum_{e\in E_i}\left[\frac{\epsilon}{10} Q_x(e)+\frac{20}{n^2\min_{v\in e}d(v)}\right]% & \text{by Claim~\ref{claim:0}}
\\
&\le \frac{\epsilon}{10} Q_x(E_i) + \frac{20}{n^2}\sum_{e\in E}\frac{1}{\min_{v\in e}d(v)} \\
&\le \frac{\epsilon}{10} Q_x(E_i) + \frac{20}{n}% & \text{by Lemma~\ref{lem:sum-min-d} and $n\ge|V|$},
\end{align*}
by Lemma~\ref{lem:sum-min-d} and $n\ge|V|$, as claimed. The second line follows by Claim~\ref{claim:0}.
\else
\begin{align*}
  \big|Q_x(E_i)-Q_{x^{(i)}}(E_i)\big|
  &\le\sum_{e\in E_i}\big|Q_x(e)-Q_{x^{(i)}}(e)\big| \\
  &\le\sum_{e\in E_i}\left[\frac{\epsilon}{10} Q_x(e)+\frac{20}{n^2\min_{v\in e}d(v)}\right] & \text{by Claim~\ref{claim:0}} \\
  &\le \frac{\epsilon}{10} Q_x(E_i) + \frac{20}{n^2}\sum_{e\in E}\frac{1}{\min_{v\in e}d(v)} \\
  &\le \frac{\epsilon}{10} Q_x(E_i) + \frac{20}{n} & \text{by Lemma~\ref{lem:sum-min-d} and $n\ge|V|$},
\end{align*}
as claimed.
\fi
\end{proof}

\claimthree*
\begin{proof}%[Proof of Claim~\ref{claim:3}]
Similarly to the previous proof, we first bound
$$\big|\widetilde Q_x(E_i)-\widetilde Q_{x^{(i)}}(E_i)\big|\le\sum_{e\in E_i}\big|\widetilde Q_x(e)-\widetilde Q_{x^{(i)}}(e)\big|.$$
Recall that $\widetilde Q_x(e) = w_e Q_x(e)$ where $w_e$ is a random variable (independent from all others) that takes value $1/p_e$ with probability $p_e$, and value $0$ otherwise. Similarly, $\widetilde Q_{x^{(i)}}(e)=w_e\widetilde Q_{x^{(i)}}(e)$. Applying this along with Claim~\ref{claim:0}, we get
\begin{align*}
	\big|\widetilde Q_x(E_i)-\widetilde Q_{x^{(i)}}(E_i)\big|&\le\sum_{e\in E_i}w_e\big|Q_x(e)-Q_{x^{(i)}}(e)\big|\\
	&\le \sum_{e\in E_i}\left[\frac{\epsilon}{10} w_e Q_x(e) + w_e\cdot\frac{20}{n^2\min_{c\in e}d(v)}\right]\\
	&= \frac{\epsilon}{10}\widetilde Q_x(e) + \sum_{e\in E_i}\frac{20w_e}{n^2\min_{v\in e}d(v)}.
\end{align*}
Note that in the sum here the term corresponding to $e$ is zero unless $e$ is sampled to $\widetilde E$,
in which case $w_e = 1/p_e \leq 1+\min_{v\in e}d(v)/\lambda$. (Recall $\lambda$ from equation~\ref{eq:def-lambda}.)
Using also Lemmas~\ref{lem:sum-min-d} and~\ref{lem:size-bound}, and the fact that $n\ge|V|$,
we have that with high probability
\ifprocs
\begin{align*}
\sum_{e\in E_i}\frac{20w_e}{n^2\min_{v\in e}d(v)}
& \le \sum_{e\in\widetilde  E_i} \frac{20}{n^2\min_{v\in e}d(v)} + \sum_{e\in \widetilde E} \frac{20}{\lambda n^2}\\
&\le \sum_{e\in E_i}\frac{20}{n^2\min_{v\in e}d(v)} + |\widetilde E|\cdot\frac{20}{\lambda n^2}\\
&\le\frac{60}{n} .
\qedhere
\end{align*}
\else
\begin{align*}
  \sum_{e\in E_i}\frac{20w_e}{n^2\min_{v\in e}d(v)}
  & \le \sum_{e\in\widetilde  E_i} \frac{20}{n^2\min_{v\in e}d(v)} + \sum_{e\in \widetilde E} \frac{20}{\lambda n^2}
  \le \sum_{e\in E_i}\frac{20}{n^2\min_{v\in e}d(v)} + |\widetilde E|\cdot\frac{20}{\lambda n^2}
  \le\frac{60}{n} .
  \qedhere
\end{align*}
\fi
\end{proof}

\claimfour*
\begin{proof}%[Proof of Claim~\ref{claim:4}]
Note that
\ifprocs
\begin{align*}
|\widetilde Q_x(E_*)- Q_x(E_*)|&\le\widetilde Q_x(E_*)+Q_x(E_*)\\
=\sum_{e\in E_*} & \max_{u,v\in e}{(x_u-x_v)}^2+\sum_{e\in E_*}w_e\cdot\max_{u,v\in e}{(x_u-x_v)}^2.
\end{align*}
\else
\begin{align*}
    |\widetilde Q_x(E_*)-Q_x(E_*)|&\le\widetilde Q_x(E_*)+Q_x(E_*)
    =\sum_{e\in E_*}\max_{u,v\in e}{(x_u-x_v)}^2+\sum_{e\in E_*}w_e\cdot\max_{u,v\in e}{(x_u-x_v)}^2.
\end{align*}
\fi
Now, we bound each term using that $\max_{v\in e}x_v^2\cdot\min_{v\in e}d(v)\le1/n^2$ by definition of $E_*$.
For the first term, we use Lemma~\ref{lem:sum-min-d},
\begin{align*}
  \sum_{e\in E_*}\max_{u,v\in e}{(x_u-x_v)}^2
  \le 4\sum_{e\in E_*}\max_{v\in e}x_v^2
  \le4\sum_{e\in E^*}\frac{1}{n^2\min_{v\in e}d(v)}
  \le\frac4n .
\end{align*}
For the second term, we have by Lemma~\ref{lem:size-bound}, and the fact that $n\ge|V|$, that with high probability,
\ifprocs
\begin{align*}
\sum_{e\in E_*}w_e\cdot\max_{u,v\in e}{(x_u-x_v)}^2
& \le 4\sum_{e\in E_*}w_e\cdot\max_{v\in e}x_v^2\\
&\le 4\sum_{e\in E_*}\frac{w_e}{n^2\min_{v\in e}d(v)}\\
&\le 4|\widetilde E|\cdot\frac{1}{\lambda n^2}\\
&\le \frac{8}{n} .
\qedhere
\end{align*}
\else
\begin{align*}
  \sum_{e\in E_*}w_e\cdot\max_{u,v\in e}{(x_u-x_v)}^2
  & \le 4\sum_{e\in E_*}w_e\cdot\max_{v\in e}x_v^2
  \le 4\sum_{e\in E_*}\frac{w_e}{n^2\min_{v\in e}d(v)}
  \le 4|\widetilde E|\cdot\frac{1}{\lambda n^2}
  \le \frac{8}{n} .
  \qedhere
\end{align*}
\fi
\end{proof}

Finally, we prove the technical crux of the theorem, Claim~\ref{claim:2}.

\claimtwo*
\begin{proof}%[Proof of Claim~\ref{claim:2}]
	We shall prove the stronger claim
	$$\mathbb P\left[\forall x\in\overline{\mathbb R^V},\ \widetilde Q_{x^{(i)}}(E_i)
	= \Big(1\pm\frac{\epsilon}{10}\Big) Q_{x^{(i)}}(E_i)\pm\frac{\epsilon r\Phi^2/32 }{10\log n}\right]
	\ge 1-\frac{1}{n^2}.
	$$
	This is indeed stronger, since for all $x\in\overline{\mathbb R^V}$, we know that $Q(x)\ge r\Phi^2/32$ by the Hypergraph Cheeger inequality (Theorem~\ref{thm:hypergraph-cheeger}). This allows us to argue that the probabilistic claim depends on $x$ only through $x^{(i)}$ and $E_i$. These are discrete which will allow for the use of union bound later on. We will first prove a deviation bound for a single instance of $(x^{(i)},E_i)$ using an additive-multiplicative Chernoff bound, and then extend it to hold for all instances simultaneously using a union bound.

Fix $i$, $x^{(i)}$, and $E_i$.
Notice that $\widetilde Q_{x^{(i)}}(E_i)=\sum_{e\in E_i}w_e\cdot\max_{a,b\in e}(x_a^{(i)}-x_b^{(i)})^2$ is a sum of independent random variables whose expectation is $Q_{x^{(i)}}(E_i)$.
Let us bound the maximum range of one summand, for some $e\in E_i$.
If $p_e=1$ the range is $0$, and otherwise the range is bounded by
\ifprocs
\begin{align*}
w_e\cdot\max_{a,b\in e}{\left(x_a^{(i)}-x_b^{(i)}\right)}^2
&\le\max_{a,b\in e} \frac{2(x_a^2+x_b^2)}{p_e}\\
&\le\frac{4}{\lambda}\max_{v\in e}x_v^2\cdot\min_{v\in e}d(v)\\
&\le \frac{2^{-i+3}}{\lambda} .
\end{align*}
\else
\begin{align*}
	w_e\cdot\max_{a,b\in e}{\left(x_a^{(i)}-x_b^{(i)}\right)}^2
	\le\max_{a,b\in e} \frac{2(x_a^2+x_b^2)}{p_e}
	\le\frac{4}{\lambda}\max_{v\in e}x_v^2\cdot\min_{v\in e}d(v)
	\le \frac{2^{-i+3}}{\lambda} .
\end{align*}
\fi
We can thus apply Theorem~\ref{thm:am-Chernoff} and get
\ifprocs
\begin{align*}
\mathbb P&\left[|\widetilde Q_{x^{(i)}}(E_i)-Q_{x^{(i)}}(E_i)|
\ge \frac{\epsilon}{10} Q_{x^{(i)}}(E_i) + \frac{\epsilon r\Phi^2/32}{10\log n}\right]\\
&\le 2\exp\left(-\frac{\epsilon/10\cdot(\epsilon r\Phi^2/32)/(10\log n)}{3\cdot2^{-i+3}/\lambda}\right)\\
&= 2\exp\left(-\frac{\lambda 2^i\epsilon^2r \Phi^2}{32\cdot2400\log n}\right) .
\end{align*}
\else
\begin{align*}
  \mathbb P\left[|\widetilde Q_{x^{(i)}}(E_i)-Q_{x^{(i)}}(E_i)|
  \ge \frac{\epsilon}{10} Q_{x^{(i)}}(E_i) + \frac{\epsilon r\Phi^2/32}{10\log n}\right]
  &\le 2\exp\left(-\frac{\epsilon/10\cdot(\epsilon r\Phi^2/32)/(10\log n)}{3\cdot2^{-i+3}/\lambda}\right)\\
  &= 2\exp\left(-\frac{\lambda 2^i\epsilon^2r \Phi^2}{32\cdot2400\log n}\right) .
\end{align*}
\fi

Now we wish to extend this high-probability bound to hold simultaneously for all possible $x^{(i)}$ and $E_i$. How many possible settings of $(x^{(i)},E_i)$ are there? Each non-zero coordinate $v$ of $x^{(i)}$ has $x_v^2d(v)\ge\epsilon^2 2^{-i}/2500$,
so there are at most $2500\cdot2^i/\epsilon^2$ such coordinates.
Furthermore, each such coordinate $x_v^{(i)}$ is an integer multiple of $1/(n^2\sqrt{d(v)})$ in the range $[-1/\sqrt{d(v)},1/\sqrt{d(v)}]$,
so there are only $2n^2$ possibilities per non-zero coordinate.
Thus, the total number of vectors $x^{(i)}$ is at most
\[
  \binom{|V|}{2500\cdot2^i/\epsilon^2}
  \cdot {(2n^2)}^{2500\cdot2^i/\epsilon^2}
  \leq {(2n^3)}^{2500\cdot2^i/\epsilon^2}
\]
We still need to enumerate the number of possible hyperedge multisets $E_i$ given $x^{(i)}$. To know whether some hyperedge $e\in E$ is in $E_i$, we must know whether the value of $\max_{v\in e}x_v^2\min_{v\in e}d(v)$ is in $(2^{-i},2^{-i+1}]$.
Unfortunately, this depends on $\max_{v\in e}x_v^2$, which is not determined by $x^{(i)}$, due to the rounding error between $x$ and $x^{(i)}$.
Let $D=\{d(v) \mid v\in V\}$ be the set of all degrees in $G$.
It suffices to know for each $v$ corresponding to a non-zero coordinate of $x^{(i)}$ the two values
\[
  \min\{d\in D \mid x_v^2d > 2^{-i}\}
  \quad \text{ and } \quad
  \max\{d\in D \mid x_v^2d \le 2^{-i+1}\}.
\]
Indeed, we need not worry about zero coordinates of $x^{(i)}$,
i.e., vertices $v$ with $x_v^2d(v)<\epsilon^2 2^{-i}/2500$,
as these cannot attain $\max_{u\in e} x_u^2$ for a hyperedge $e\in E_i$.
Thus, the total number of possible multisets $E_i$ \textit{given $x^{(i)}$} is at most $(|V|^2)^{2500\cdot2^i/\epsilon^2}\le(n^2)^{2500\cdot2^i/\epsilon^2}$.

We are now ready to apply a union bound,
	\begin{align*}
        \mathbb P& \left[\forall x,\ |\widetilde Q_{x^{(i)}}(E_i)-Q_{x^{(i)}}(E_i)|\le \frac{\epsilon}{10} Q_{x^{(i)}}(E_i)+ \frac{\epsilon(r\Phi^2/32)}{10\log n}\right] \\
          &\le \left(n^2\cdot (2n^3)\right)^{2500\cdot2^i/\epsilon^2}
            \cdot2\exp\left( -\frac{\lambda 2^i\epsilon^2 r\Phi^2} {32\cdot2400\log n} \right)\\
	&\le 2\exp\left(\frac{15000\cdot2^i\log n}{\epsilon^2} - \frac{\lambda 2^i\epsilon^2r \Phi^2}{32\cdot2400\log n}\right)
	\le \frac{1}{n^2},
	\end{align*}
where the last inequality holds as long as
$\lambda \geq 24\cdot10^8\cdot\log^2 n/(\epsilon^4\Phi^2r)$,
which is indeed how we set $\lambda$.
\end{proof}

%%% Local Variables:
%%% mode: latex
%%% TeX-master: "000-main"
%%% End:

\ifprocs
\else
%!TEX root=./000-main.tex

%!TEX root=./000-main.tex

\section{Expander Decomposition}\label{sec:expander-decomposition}

This section provides a procedure to decompose an input hypergraph 
into expanders while cutting a small number of hyperedges.
This is stated in the following lemma, which we fully prove for completeness,
as we cannot find a useful reference for it. 
It is based on the standard technique of iteratively removing a sparse cut,
with slight adaptations like a minimum-degree guarantee in each expander,
and an approximation algorithm for sparsest cut (minimal expansion)
in a hypergraph.

\begin{lemma} [Expander Decomposition] 
  \label{lem:partition}
There exists a polynomial-time algorithm that, given an $r$-uniform hypergraph $G=(V,E)$ with $n$ vertices and $m$ hyperedges, outputs disjoint vertex subsets $C_1,\ldots, C_k\subseteq V$ (not necessarily a partition) that satisfy
\begin{enumerate}
\item \label{item:expansion}
  $\Phi(G[C_j]) = \Omega(1/(r\log^2 n))$ for all $j=1,\ldots,k$;
\item \label{item:min-degree}
  $\min_{v\in C_j}d_{G[C_j]}(v)\ge m/(4n)$ for all $j=1,\ldots,k$; and
\item \label{item:size}
  $|E\setminus\bigcup_{j=1}^k E(G[C_j])|\le m/2$.
\end{enumerate}
\end{lemma}

We first provide (in Section~\ref{subsubsec:sparse-cut})
an approximation algorithm for the sparsest cut problem on hypergraphs
by slightly modifying a known algorithm from~\cite{Feige2008}.
We then use it (in Section~\ref{subsubsec:partition})
to prove Lemma~\ref{lem:partition},
where we decompose an input hypergraph into expanders
by iteratively deleting sparse cuts and low-degree vertices.

\subsection{Approximating Sparsest Cut}
\label{subsubsec:sparse-cut}

\begin{lemma}\label{lem:sparse-cut}
  There exists a polynomial-time algorithm that, given a hypergraph $G=(V,E)$, computes a cut $S \subseteq V$ with $\Phi(S) = O(\log n \cdot \Phi(G))$.
\end{lemma}

We remark that a better approximation ratio $O(\sqrt{\log n})$
was shown by Louis and Makarychev~\cite{LM17}.
Strictly speaking, their definition of expansion is slightly different,
using $|S|$ rather than $\mathrm{vol}(S)$,
but their results probably extend also to our setting. 

Our algorithm is an extension of known approximation algorithms for the sparsest cut problem on ordinary graphs.
Specifically, we follow the rounding procedure techniques of
Feige, Hajiaghayi and Lee~\cite{Feige2008},
which in turn build on the linear programming relaxation approach
introduced by Leighton and Rao~\cite{Leighton1999}
and on the rounding procedures based on metric embeddings devised by
Aumann and Rabani~\cite{AR98}
and Linial, London and Rabinovich~\cite{Linial1995}.

For a hypergraph $G=(V,E)$ and a vertex set $S \subseteq V$, we define
\[
  \phi(S) = \frac{|E(S,V \setminus S)|}{\mathrm{vol}(S) \cdot \mathrm{vol}(V \setminus S)}
\]
(notice the difference in the denominator from $\Phi(S)$),
and let $\phi(G) = \min_{S \subseteq V} \phi(S)$.
We note that
\[
  \frac{\mathrm{vol}(S)\cdot \mathrm{vol}(V \setminus S)}{\mathrm{vol}(G)}  \leq \min\{\mathrm{vol}(S), \mathrm{vol}(V \setminus S)\} \leq \frac{2 \cdot \mathrm{vol}(S)\cdot \mathrm{vol}(V \setminus S)}{\mathrm{vol}(G)},
\]
and hence a $\rho(n)$-approximation algorithm for $\phi(G)$
immediately gives a $2\rho(n)$-approximation for $\Phi(G)$.

For a vertex set $S \subseteq V$, let $1_S \colon V \to \{0,1\}$ be the indicator function of $S$.
Then, we have
\[
  \phi(S) = \frac{\sum_{e \in E}\max_{u,v \in e}|1_S(u) - 1_S(v)|}{\sum_{u,v \in V}d(u) d(v) |1_S(u)-1_S(v)|}.
\]
By relaxing this optimization over cut metrics (distances induced by indicator functions) to optimization over all pseudo-metrics, we obtain the following linear program (LP).
\begin{align*}
  \begin{array}{lll}
  \text{minimize} & \displaystyle \sum_{e \in E} z(e), \\
  \text{subject to}     & \displaystyle \sum_{u, v \in V} d(u) d(v) \ell(u,v) = 1, \\
      & z(e) \geq \ell(u,v) & \forall e \in E, u, v \in e,\\
      & \ell(u,w) \leq \ell(u,v) + \ell(v,w) & \forall u,v,w \in V,\\
      & \ell(u,v) = \ell(v,u) \geq 0 & \forall u,v \in V .
  \end{array}
\end{align*}
This LP has variables $\set{\ell(u,v)}_{u,v\in V}$ and $\set{z(e)}_{e\in E}$,
and its size is $O(n^2+mr^2)$.
Our algorithm solves this LP and then rounds the solution as explained next.

Our rounding procedure is similar to~\cite{Feige2008},
who designed an approximation algorithm for the sparsest vertex-cut problem. They use the following embedding result due to Bourgain.
\begin{theorem}[Bourgain’s embedding~\cite{Bourgain1985}]\label{thm:Bourgain}
Let $\ell : V \times V \to \mathbb{R}$ be a pseudo-metric on an $n$-point set $V$.
Then there exists an embedding $f:V\to\RR^k$
with distortion $D=O(\log n)$ in the following sense:
\begin{align}
  \forall u, v \in V,
  \qquad
  &\max_{i\in[k]} |f_i(u) - f_i(v)|  \leq \ell(u, v), \label{eq:BourgainUB}
  \\
  \forall u, v \in V,
  \qquad
  &\frac{1}{k}\sum_{i\in[k]} |f_i(u) - f_i(v)| \geq \frac{\ell(u, v)}{D}. \label{eq:BourgainLB}
\end{align}
\end{theorem}
Linial, London and Rabinovich~\cite{Linial1995} showed how to compute
such an embedding with high probability in time $\tO(n^2)$.

\begin{remark}
The above form of Bourgain's embedding is slightly stronger than
the usual statement of embedding into $\ell_1$ (see also~\cite{MR01}).
Indeed, the usual statement follows easily by viewing $f$
as an embedding into $\ell_1$, then
\[
  \forall u, v \in V,
  \qquad
  \frac{\ell(u, v)}{D}
  \leq \frac{1}{k} \|f(u) - f(v)\|_1
  \leq \ell(u, v) ,
\]
which means that scaling $f$ by factor $1/k$ achieves distortion $D$.
\end{remark}

\begin{lemma}\label{lem:rounding}
Given an embedding $f: V \to \mathbb{R}^k$,
one can find in polynomial time a cut $S^* \subseteq V$ such that
\[
  \phi(S^*)
  \leq \min_{i\in[k]} \frac{\sum_{e \in E}\max_{u,v \in e} |f_i(u) - f_i(v)|} {\sum_{u,v \in V}d(u) d(v)|f_i(u)-f_i(v)|} .
\]
\end{lemma}

\begin{proof}
Let $i\in [k]$ be the index that minimizes the above ratio,
and define the embedding $g : V \to \mathbb{R}$ by scaling and translating the corresponding $f_i$ such that $\min_v g(v) = 0$ and $\max_v g(v) = 1$.

Pick uniformly at random a threshold $s\in [0, 1]$,
and consider the (random) set $S = \set{v \in V \mid g(v) > s}$.
Then by simple calculations
\begin{align*}
& \forall u,v\in V,
\quad
\EX |1_{S}(u) - 1_{S}(v)|
=
|g(u) - g(v)| ,
\\
& \forall e \in E,
\qquad
\EX \max_{u,v \in e}|1_S(u) - 1_S(v)|
= \max_{u \in e}g(u) - \min_{v \in e}g(v).
\end{align*}

It follows that there must exist $S^*= \set{v\in V \mid g(v)>s^*}$ that is non-trivial (i.e., $S^*\neq \emptyset,V$) for which
\begin{align*}
  \phi(S^*)
  & =  \frac{\sum_{e \in E}\max_{u,v \in e}  |1_{S^*}(u) - 1_{S^*}(v)|} {\sum_{u,v \in V}d(u) d(v)  |1_{S^*}(u)-1_{S^*}(v)|}
  \le \frac{\EX \sum_{e \in E}\max_{u,v \in e}  |1_S(u) - 1_S(v)|} {\EX \sum_{u,v \in V}d(u) d(v)  |1_S(u)-1_S(v)|} \\
  &= \frac{\sum_{e \in E}\max_{u,v \in e} |g(u) - g(v)|} {\sum_{u,v \in V}d(u) d(v)|g(u)-g(v)|}
  = \frac{\sum_{e \in E}\max_{u,v \in e} |f_i(u) - f_i(v)|} {\sum_{u,v \in V}d(u) d(v)|f_i(u)-f_i(v)|} .
\end{align*}
A polynomial-time implementation can simply take the best sweep cut.
\end{proof}

\begin{proof}[Proof of Lemma~\ref{lem:sparse-cut}]
  The algorithm computes an optimal LP solution
  and then applies to it Bourgain's embedding (Theorem~\ref{thm:Bourgain})
  and then Lemma~\ref{lem:rounding} to find a cut $S^*\subset V$.
  This is a (randomized) polynomial-time algorithm,
  and with high probability its output $S^*\subseteq V$ satisfies
  \begin{align*}
    \phi(S^*)
    & \leq  \min_{i\in[k]} \frac{\sum_{e \in E}\max_{u,v \in e} |f_i(u) - f_i(v)|} {\sum_{u,v \in V}d(u) d(v)|f_i(u)-f_i(v)|}
    \leq \min_{i\in[k]} \frac{\sum_{e \in E} \max_{u,v \in e} \ell(u,v)} {\sum_{u,v \in V}d(u) d(v)|f_i(u)-f_i(v)|} \\
    &\leq \frac{\sum_{e \in E} z(e)} {\max_{i\in[k]} \sum_{u,v \in V}d(u) d(v)|f_i(u)-f_i(v)|}
    \leq \frac{\sum_{e \in E} z(e)} {\EX_{i\in[k]} \sum_{u,v \in V}d(u) d(v)|f_i(u)-f_i(v)|} \\
    &\leq \frac{\sum_{e \in E} z(e)} {\sum_{u,v \in V}d(u) d(v) \ell(u,v)/ D}
    = D \cdot \operatorname{LP}
      \leq O(\log n) \cdot \phi(G).
  \end{align*}
  We thus conclude an $O(\log n)$-approximation algorithm for $\phi(G)$.
\end{proof}

\subsection{Proof of Lemma~\ref{lem:partition}}\label{subsubsec:partition}

We now prove Lemma~\ref{lem:partition}. 
We shall refer to the algorithm given in Lemma~\ref{lem:sparse-cut} as \Call{SparseCut}{}.

\begin{algorithm}
  \caption{}\label{alg:partition}
  \begin{algorithmic}[1]
    \Procedure{ExpanderDecomposition}{$G=(V,E)$}
      \State $\mathcal{C} \leftarrow \{V\}$
      \While{\textbf{true}}
        \If{$\exists C \in \mathcal{C}$ and $v \in C$ such that $d_{G[C]}(v) < m/(4n)$}\label{line:min-degree}
          \State remove $v$ from $C$.\label{line:discard}
        \ElsIf{$\exists C \in \mathcal{C}$ s.t.\ \Call{SparseCut}{$G[C]$} finds a cut $S$ with $\Phi_{G[C]}(S) \leq 1/(4r \log n)$}
        \State $\mathcal{C} \leftarrow (\mathcal{C} \setminus \{C\}) \cup \{S,C \setminus S\}$ \label{line:split}
        \Else
          \State \Return $\mathcal{C}$
        \EndIf
      \EndWhile
    \EndProcedure
  \end{algorithmic}
\end{algorithm}

\begin{proof}[Proof of Lemma~\ref{lem:partition}]
  Our algorithm is given in Algorithm~\ref{alg:partition}.
By Lemma~\ref{lem:sparse-cut}, there exists a constant $K >0$ such that, if there is a cut of expansion at most $1/(Kr\log^2 n)$, then \Call{SparseCut}{} finds a cut of expansion at most $1/(4r\log n)$, and hence Guarantee~\ref{item:expansion} holds. Guarantee~\ref{item:min-degree} also holds, since the algorithm only returns when when no vertices violate the min-degree condition, due to Line~\ref{line:min-degree}.

  We now show Guarantee~\ref{item:size}, which claims that at most half the hyperedges are omitted from this clustering.
  First, we bound the number of hyperedges cut by splitting $C$ on Line~\ref{line:split}.
  To this end, we charge some weight to vertices for each cut $(S,C \setminus S)$ found throughout the algorithm.
  Specifically, suppose $|S|\le|C\setminus S|$ --- then we charge cut hyperedges to vertices in $S$ \emph{proportionally to their degree}: Each vertex $v\in S$ is charged
  $$E(S,C\setminus S)\cdot\frac{d_{G[C]}(v)}{\mathrm{vol}_{G[C]}(S)}\le d_{G[C]}(v)\cdot\frac{E(S,C\setminus S)}{\min(\mathrm{vol}_{G[C]}(S),\mathrm{vol}_{G[C]}(C\setminus S))}\le\frac{d_{G[C]}(v)}{4r\log n}\le\frac{d(v)}{4r\log n}.$$
  If $|C\setminus S|\le|S|$, vertices of $C\setminus S$ get charged similarly. Since each vertex is charged only when the size of its containing component is decreased by at least a factor of $2$, each vertex can get charged a maximum of $\log n$ times. This means that in total $v$ is charged at most $d(v)/(4r)$ cut hyperedges.

  Hence, the total number of hyperedges cut by splitting components is bounded by
  \[
    \sum_{v \in V}\frac{d(v)}{4r}  = \frac{m}{4}.
  \]
  The total number of hyperedges cut by discarding low-degree vertices in Line~\ref{line:discard} is at most
  \[
    \frac{m}{4 n} \cdot n = \frac{m}{4}.
  \]
  Altogether, the total number of hyperedges cut by the algorithm is at most $m/4+m/4=m/2$.
\end{proof}

%%% Local Variables:
%%% mode: latex
%%% TeX-master: "000-main"
%%% End:

%!TEX root=./000-main.tex

\section{General Spectral Sparsification of Hypergraphs}\label{sec:general-sparsification}
In this section, we prove Theorem~\ref{thm:general-sparsification}.
We describe our construction in Section~\ref{subsec:general-sparsification-construction} and then prove its correctness in Section~\ref{subsec:general-sparsification-correctness}.

\subsection{Construction}\label{subsec:general-sparsification-construction}

We shall call a hyperedge a \emph{self-loop} if all of its vertices are identical
(i.e., the number of distinct vertices in it is one).
We explicitly prohibit self-loops in the next lemma for a technical reason
inside the proof of Theorem~\ref{thm:general-sparsification}.

\begin{lemma}\label{lem:expander-sparsify}
There is an algorithm that, given a parameter $n$, given $100/n\le\epsilon\le 1/2$ and an $r$-uniform hypergraph $G=(V,E)$ with $|V|\le n$ and expansion $\Phi(G)\ge\Omega(\tfrac{1}{r\log^2n})$ and $r\le \epsilon n/\log^6n$, outputs an $\epsilon$-spectral sparsifier of $G$ with at most $V|r(\epsilon^{-1}\log n)^{O(1)}$ hyperedges and no self-loops with probability at least $1-O((\log n)/n^2)$ in $O(r|E|)$ time.
Let $\Call{ExpanderSparsify}{G,\epsilon}$ be such an algorithm.
\end{lemma}

\begin{proof}
  This is a simple corollary of Theorem~\ref{thm:expander-sparsification} when applied with $\Phi=\Omega(\tfrac1{r\log^2n})$, the only difference being the exclusion of self-loops. One can simply remove all the self-loops from the sparsifier obtained from Theorem~\ref{thm:expander-sparsification},
  as this does not change the spectral properties at all.
\end{proof}
\begin{definition}\label{def:tr-cl}
Let $\sim$ be a relation on a ground set $V$. The transitive closure $\approx$ of $\sim$ is a relation on $V$ defined as follows
$$a\approx b\ \Longleftrightarrow\ \exists c_1,\ldots,c_k\in V, a\sim c_1\sim\cdots\sim c_k\sim b.$$
\end{definition}

\begin{definition}\label{def:contract}
Let $G=(V,E)$ be a hypergraph and let $\approx$ be an equivalence relation on $V$.
Then the \emph{contraction} of $G$ with respect to $\approx$, denoted by $G/\approx$, is a hypergraph on the vertex set $V'=V/\approx$ and the multiset of hyperedges
$$E'=\{\{[v] \mid v\in e\} \mid e\in E\},$$
where $[v]$ denotes the equivalence class of $v$ with respect to $\approx$.
We stress that each hyperedge in the contracted hypergraph is itself a multiset,
i.e., for each $e\in E$ there is a corresponding hyperedge in $E'$ of
total multiplicity $|e|$.
Moreover, $E'$ is a multiset and thus may have copies of the same hyperedge (i.e., parallel hyperedges).
\end{definition}

\begin{remark}\label{rem:important-remark}
Since there is a bijection between $E$ and $E'$,
we shall slightly abuse notation and equate the hyperedges of $G$ and of $G/\approx$ (i.e., use one as a shorthand for the other).
This occurs already in the next observation.
\end{remark}

\begin{observation}\label{obs:contraction-sparsification}
Let $G=(V,E)$ be a hypergraph and let $\approx$ be an equivalence relation on $V$.
Let $\widetilde G$ be an $\epsilon$-spectral sparsifier of the contracted hypergraph $G/\approx$ (not of $G$).
Then it need not be true that the energy of $G$ is always approximated by the energy of $\widetilde G$ in the sense that for all $x\in\mathbb R^V$,
\[
  \widetilde Q(x)=(1\pm\epsilon)Q(x).
\]
Here $\widetilde Q$ is the energy of the sparsifier $\widetilde G$ \emph{when interpreted as a subgraph of $G$, not $G/\approx$}. That is, we take the weighted hyperedges found in $\widetilde G$ and interpret them as hyperedges over the vertex set of $G$ (see Remark~\ref{rem:important-remark}).

However, the equation above \emph{does} hold if $x$ is constant on each equivalence class of $\approx$, that is whenever
$$u\approx v\ \Rightarrow\ x_u=x_v.$$
\end{observation}

We are now ready to define our main algorithm for sparsifying arbitrary $r$-uniform input hypergraphs $G=(V,E)$. We use \textsc{ExpanderDecomposition} produce subsets of $V$ that are good expanders and \textsc{ExpanderSparsify} to sparsify them. Since we get rid of at least half the edges each turn, this process, repeated until no hyperedges remain, would take potentially $\Omega(\log m)=\widetilde\Omega(r\log n)$ rounds. To reduce the number of rounds, we contract clusters into supernodes shortly after sparsifying them. For more intuition on the workings of Algorithm~\ref{alg:main} see Section~\ref{sec:tech-overview-general}.

\begin{algorithm}[H]
\caption{Algorithm sparsifying an arbitrary hypergraph}\label{alg:main}
\begin{algorithmic}[1]
\Procedure{Sparsify}{$G,\epsilon$}
    \State $E^{(0)}\gets E$
    \State $\widetilde E\gets\emptyset$
    \For{$i=0,\ldots,\log m$}
        \If{$i\ge10\log n$}
            \State $\sim_i\ \gets$ the relation on $V$ where $u\sim_i v$ iff $\exists i'\le i-10\log n,\ \exists j, \ [u],[v]\in C^{(i')}_j$\label{line:sim-def}
            \State $\approx_i\ \gets\text{transitive closure of}\sim_i$
            \State $G^{(i)}\gets(V,E^{(i)})/\approx_i$
        \Else
            \State $G^{(i)}\gets(V,E^{(i)})$
        \EndIf
        \State $\left(C^{(i)}_1\ldots,C^{(i)}_{k_i}\right)\gets\textsc{ExpanderDecomposition}(G^{(i)})$
        \For{$j=1,\ldots, k_i$}
        \State $G_j^{(i)}=(C_j^{(i)},E_j^{(i)}) \gets G^{(i)}[C_j^{(i)}]$
            \State $(C_j^{(i)},\widetilde E_j^{(i)},w_j^{(i)}) \gets \textsc{ExpanderSparsify}(G_j^{(i)},\epsilon/10)$\label{line:def-tilde-Eij}
            \State $\widetilde E\gets \widetilde E\cup \widetilde E_j^{(i)}$\label{line:update-tilde-E}
            \State $w|_{E_j^{(i)}}=w_j^{(i)}$\label{line:set-w}
        \EndFor
        \State $E^{(i+1)}\gets E^{(i)}\setminus\cup_{j=1}^{k_i}E_j^{(i)}$\label{line:remove-edges}
    \EndFor
    \State \Return $\widetilde G=(V,\widetilde E)$
\EndProcedure
\end{algorithmic}
\end{algorithm}

$G$ is assumed to be $r$-regular, and $\epsilon$ is assumed to be in $[1000/n,1/2]$.
In Line~\ref{line:sim-def}, $[u],[v]\in C_j^{(i')}$ means that the \textit{supernodes containing $u$ and $v$} were in the same expander of the decomposition at an earlier level $i'$.
Note also the abuse of notation in Lines~\ref{line:def-tilde-Eij} and~\ref{line:update-tilde-E} (in accordance with Remark~\ref{rem:important-remark}): In Line~\ref{line:def-tilde-Eij}, $\widetilde E_j^{(i)}$ is defined as a multiset of hyperedges on the \textit{contracted} vertex set $C_j^{(i)}\subseteq V/\approx_i$, but in Line~\ref{line:update-tilde-E} we treat it as containing hyperedges supported on $V$.
This is justified because elements of $\widetilde E_j^{(i)}$
have clearly corresponding counterparts in $E$
(recall our definition of a sparsifier as a weighted subgraph and
that a contraction maintains a bijection between the hyperedges),
and it would only complicate the notation to make this distinction formal.

Also note that throughout the contractions, our graphs always remain $r$-regular, thanks to the use of multisets as hyperedges (see Definition~\ref{def:contract}). This --- along with Guarantee~\ref{item:expansion} of Lemma~\ref{lem:partition} that the expansion of $G_j^{(i)}$ is $\Omega(\tfrac1{r\log^2n})$ --- allows the use of \textsc{ExpanderSparsify} in Line~\ref{line:def-tilde-Eij}.

Line~\ref{line:set-w} simply means that we set the weights of hyperedges in $E_j^{(i)}$ as in the sparsifier computed in Line~\ref{line:def-tilde-Eij}. This is consistent with the update of $\widetilde E$ in Line~\ref{line:update-tilde-E}.

\subsection{Correctness}\label{subsec:general-sparsification-correctness}
First we bound the total size and number of clusters $C_j^{(i)}$:

\begin{claim}\label{claim:all-clusters}
	$$\sum_{i=0}^{\log m}\sum_{j=1}^{k_i}\left|C_j^{(i)}\right|\le21n\log n.$$
\end{claim}
\begin{proof}
	We first bound the number of \emph{distinct} vertices (including supernodes)
	that appear throughout the execution of Algorithm~\ref{alg:main}, i.e.,
	$$V^*=\bigcup_{i=0}^{\log m}\bigcup_{j=1}^{k_i}C_j^{(i)}.$$
	This set includes vertices from $V$, as well as some contracted supernodes that are subsets of $V$. Note however, that these sets form a laminar family, since supernodes are only constructed by merging previous supernodes --- no supernodes ever get broken apart. Therefore, $|V^*|\le2n-1$.

	Let us first bound the number of \textit{non-singleton} clusters a single $v\in V^*$ can participate in. As soon as $v$ participates in some cluster $C$ of size at least $2$ at level $i$, we know that at level $i+10\log n$ $C$ will be contracted into a different vertex in $V^*$. Therefore, $v$ can participate in at most $10\log n$ clusters of size $2$ or more.

	On the other hand, any $v\in V^*$ can participate in at most one singleton cluster. Indeed as soon as $v$ forms a cluster $\{v\}$ all of its self-loop hyperedges are removed, and it will not have self-loops until it is contracted again. Therefore, $\{v\}$ cannot appear again as a cluster as this would violate Guarantee~\ref{item:size} of Lemma~\ref{lem:partition}.

	Therefore, the total of the sum is at most
	$
	(2n-1)\cdot(10\log n+1) \le 21n\log n
	$, for large enough $n$.
\end{proof}
As an immediate corollary to this, we can conclude that are there at most $21n\log n$ clusters considered throughout the entire algorithm, since the size of a cluster is always at least one. Therefore, \textsc{ExpanderSparsify} --- the only non-deterministic step of our algorithm --- gets called at most $21n\log n$ time, and succeeds \textit{all} of these times with combined probability $1-O((\log^2n)/n)$. From this point on we consider only the event that all of these calls indeed succeed.

Let us consider the size of the sparsifier output by $\textsc{Sparsify}(G,\epsilon)$.
\begin{lemma}\label{lem:general-sparsification-size}
  Let $G$ be $r$-regular hypergraph and let $(r\log^6n)/n\le\epsilon\le1/2$. The hypergraph returned by $\textsc{Sparsify}(G,\epsilon)$ has $nr(\epsilon^{-1}\log n)^{O(1)}$ hyperedges with probability $1-O((\log^2n)/n)$.
\end{lemma}
\begin{proof}
The hyperedge multiset of the hypergraph $\widetilde G=(V,\widetilde E)$
is simply the union of $\widetilde E_j^{(i)}$ for all $i$ and $j$,
and by Lemma~\ref{lem:expander-sparsify} the size of each of them is bounded by
$|\widetilde E_j^{(i)}|
  \le |C_j^{(i)}|r(\epsilon^{-1}\log n)^{O(1)}$ with probability $1-O(\log n^2/n)$.
So by Claim~\ref{claim:all-clusters}
\begin{align*}
	|\widetilde E|\le\sum_{i=0}^{\log m}\sum_{j=1}^{k_i}\left|C_j^{(i)}\right|r{\left(\epsilon^{-1}\log n\right)}^{O(1)}\le nr{\left(\epsilon^{-1}\log n\right)}^{O(1)},
\end{align*}
as claimed.
\end{proof}

\begin{lemma}\label{lem:general-sparsification-spectral}
  Let $G$ be $r$-regular hypergraph and let $(r\log^6n)/n\le\epsilon\le1/2$. Then  hypergraph $\textsc{Sparsify}(G,\epsilon)$ an $\epsilon$-spectral sparsifier to $G$ in polynomial time with probability at least $1-O((\log^2n)/n)$.
\end{lemma}
\begin{proof}
We will consider each cluster $C_j^{(i)}$ separately, it suffices to prove that for each $i$ and $j$, the energy of $(V,\widetilde E_j^{(i)})$ approximates the energy of $(V,E_j^{(i)})$ up to small additive and multiplicative errors. We must first verify that each hyperedge of $E$ appears in exactly one of $E_j^{(i)}$. Indeed, any hyperedge can appear in at most one of them: $E_j^{(i)}$ gets removed from $E^{(i+1)}$ at the end of the main for loop in Line~\ref{line:remove-edges}. Furthermore, all hyperedges eventually get removed this way: By Guarantee~\ref{item:size} of Lemma~\ref{lem:partition}, after each round no more than half of the hyperedges remain, so after $\log m+1$ rounds all hyperedges must be gone and $E^{(\log m+1)}=\emptyset$.

Let us fix a single cluster $C_j^{(i)}$ and compare the spectral properties of $E_j^{(i)}$ and $\widetilde E_j^{(i)}$. Let $\widehat C_j^{(i)}$ be the set of vertices from $V$ that make up $C_j^{(i)}$ after contraction by $\approx_i$, that is, formally
$$\widehat C_j^{(i)}=\{v\in V \mid [v]\in C_j^{(i)}\},$$
where $[v]$ denotes the equivalence class of $v$ with respect to $\approx_i$.
At this point it is important to keep in mind the differences between the graphs $\widehat G_j^{(i)}=\left(\widehat C_j^{(i)},E_j^{(i)}\right)$, $G_j^{(i)}=\left(C_j^{(i)},E_j^{(i)}\right)$, $\widetilde G_j^{(i)}=\left(C_j^{(i)},\widetilde E_j^{(i)}\right)$, and $\breve G_j^{(i)}=\left(\widehat C_j^{(i)},\widetilde E_j^{(i)}\right)$. Though $\widehat G_j^{(i)}$ and $G_j^{(i)}$ share their hyperedge multiset, they are not the same, in fact $G_j^{(i)}=\widehat G_j^{(i)}/\approx_i$. Similarly, $\widetilde G_j^{(i)}=\breve G_j^{(i)}/\approx_i$.
Lemma~\ref{lem:expander-sparsify} says that $\widetilde G_j^{(i)}$ is an $\epsilon$-spectral sparsifier to $G_j^{(i)}$, it makes no such guarantee about $\widehat G_j^{(i)}$ and $\breve G_j^{(i)}$.

We will show that for every vector $x\in\mathbb R^V$,
\begin{equation}\label{eq:cluster-approx}
\breve Q(x)=\left(1\pm\frac{4\epsilon}{10}\right)\widehat Q(x)\pm\frac{3\epsilon Q(x)}{n^2},
\end{equation}
where $\widehat Q$ and $\breve Q$ denote the energy with respect to $\widehat G_j^{(i)}$ and $\breve G_j^{(i)}$, respectively. Recall from Observation~\ref{obs:contraction-sparsification} that the equation
\begin{equation}\label{eq:general-main}
\breve Q(x)=\left(1\pm\frac{\epsilon}{10}\right)\widehat Q(x)
\end{equation}
holds when $x$ is constant within all supernodes of $C_j^{(i)}$.
Informally, our plan is to round $x$ to $\widetilde x$ such that it satisfies this requirement and then show that
$$\widehat Q(x)\cong\widehat Q(\widetilde x)\cong\breve Q(\widetilde x)\cong\breve Q(x).$$

We may assume without loss of generality that $i\ge10\log n$, as otherwise no contraction takes place and Equation~\eqref{eq:general-main} holds trivially. Denote the maximum discrepancy between the values of $x$ within supernodes of $C_j^{(i)}$ by
$$\delta=\max_{u,v\in\widehat C_j^{(i)}:\ u\approx_i v}|x_u-x_v|.$$

Our rounding of $x$ to $\widetilde x$ is mostly arbitrary, we enforce only that no coordinate changes by more than an additive $\delta$.

We will show the following three claims in the next section.
\begin{restatable}{claim}{claimgone}\label{claim:g1}
	$$\widehat Q(\widetilde x)=\left(1\pm\frac{\epsilon}{10}\right)\widehat Q(x)\pm\frac{\epsilon Q(x)}{n^2}. $$
\end{restatable}

\begin{restatable}{claim}{claimgtwo}\label{claim:g2}
$$\breve Q(\widetilde x)=\left(1\pm\frac{\epsilon}{10}\right)\widehat Q(\widetilde x). $$
\end{restatable}

\begin{restatable}{claim}{claimgthree}\label{claim:g3}
$$\breve Q(\widetilde x)=\left(1\pm\frac{\epsilon}{10}\right)\breve Q(x)\pm\frac{\epsilon Q(x)}{n^2}. $$
\end{restatable}

By combining Claims~\ref{claim:g1},~\ref{claim:g2}, and~\ref{claim:g3} we get that the quadratic form $\breve Q$ is indeed close to the quadratic form $\widehat Q$, with small additive and multiplicative error, as claimed.
As in the proof of Lemma~\ref{lem:spectral-property}, first note that the multiplicative error between $\widehat Q(x)$ and $\breve Q(x)$ by itself would only amount to $(1\pm4\epsilon/10)$. Similarly, the additive error by itself would be exactly $2\epsilon Q(x)/n^2$, small enough even when summed over all possible clusters $C_j^{(i)}$. This is because the number of clusters throughout the whole algorithm is bounded by $21n\log n$ due to Claim~\ref{claim:all-clusters}.

Formally, we consider one direction of the bound first:
\begin{align*}
	\breve Q(x)&\le{\left(1-\frac{\epsilon}{10}\right)}^{-1}\left[\frac{\epsilon Q(x)}{n^2}+\breve Q(\widetilde x)\right]&\text{By Claim~\ref{claim:g3}}\\
	&\le{\left(1-\frac{\epsilon}{10}\right)}^{-1}\left[\frac{\epsilon Q(x)}{n^2}+\left(1+\frac{\epsilon}{10}\right)\widehat Q(\widetilde x)\right]&\text{By Claim~\ref{claim:g2}}\\
	&\le{\left(1-\frac{\epsilon}{10}\right)}^{-1}\left[\frac{\epsilon Q(x)}{n^2}+\left(1+\frac{\epsilon}{10}\right)\left[\frac{\epsilon Q(x)}{n^2}+\left(1+\frac{\epsilon}{10}\right)\widehat Q(x)\right]\right]&\text{By Claim~\ref{claim:g1}}\\
	&\le \left(1+\frac{4\epsilon}{10}\right)\widehat Q(x) +\frac{3\epsilon Q(x)}{n^2},
\end{align*}
since $\epsilon\le1/2$. The other direction that
$$\breve Q(x)\ge\left(1-\frac{\epsilon}{10}\right)\widehat Q(x)-\frac{3\epsilon Q(x)}{n^2}$$
 follows similarly, which concludes the proof of Equation~\eqref{eq:cluster-approx}.

Finally, we can sum over all clusters $C_j^{(i)}$, noting that their number cannot exceed $21n\log n$, to get that $Q_{\widetilde E}(x)=(1\pm\epsilon)Q(x)$.
\end{proof}

Combining Lemmas~\ref{lem:general-sparsification-size} and~\ref{lem:general-sparsification-spectral}, we get a polynomial-time algorithm that, given an unweighted $r$-uniform hypergraph, constructs an $\epsilon$-spectral sparsifier with $nr (\epsilon^{-1}\log n)^{O(1)}$ hyperedges.
We will discuss how to handle weighted hypergraphs and reduce the running time to $O(mr^2) + n^{O(1)}$ in Section~\ref{subsec:weighted}.

\subsection{Proofs of Claims~\ref{claim:g1},~\ref{claim:g2}, and~\ref{claim:g3}}
\label{sec:general-sparsification-proofs}

Before proceeding on to Claims~\ref{claim:g1},~\ref{claim:g2}, and~\ref{claim:g3} from the previous section, we prove a simple supporting lemma. This allows us to relate the total energy of $G$ to the energies of the various $C_j^{(i)}$ clusters.

\begin{lemma}\label{lem:contract-compare}
	Let $G=(V,E)$ be an arbitrary hypergraph and let $x$ be a vector in $\mathbb R^V$. Let $\approx$ be an equivalence relation on $V$ and define the contraction $G'=(V',E')=G/\approx$. Let $x'\in\mathbb R^V$ be a specification of $x$ on vertices of $G'$ such that
	$$\forall v'\in V',\exists v\in v':\ x'_{v'}=x_v.$$
	That is each vertex in $V'$ takes the value of one of the vertices in $V$ from which it was contracted. Then
	$$Q(x)\ge Q'(x),$$
	where $Q$ and $Q'$ denote the energy with respect to $G$ and $G'$ respectively.
\end{lemma}
\begin{proof}
	We examine each hyperedge of $E$ separately. Let $e\in E$ and let the corresponding hyperedge in $E'$ be $e'$. By definition $Q_x(e)=\max_{a,b\in e}(x_a-x_b)^2$ and $Q'_{x'}(e')=\max_{a,b\in e'}(x_a-x_b)^2$. By definition of $x'$, each value of $x'$ in $e'$ also appears as a value of $x$ in $e$. Therefore, $Q_x(e)$ and $Q'_{x'}(e')$ are maximizations of the same formula, with the former having more choice in $x_a$ and $x_b$, so $Q_x(e)\ge Q'_{x'}(e')$. Summing this over all hyperedges we get
	\[
		Q(x)\ge Q'(x'). \qedhere
	\]
\end{proof}

We are now ready to prove a claim bounding the effect of rounding from $x$ to $\widetilde x$ on the energy of a single hyperedge. This is the main technical result of the section that allows us to do contraction in Algorithm~\ref{alg:main}.

\begin{claim}\label{claim:g0}
	For all $e\in E_j^{(i)}$,
	\[
	Q_{\widetilde x}(e)
	= \left(1\pm\frac\epsilon{10}\right) Q_x(e)\pm\frac{\epsilon Q(x)}{2n^3m_i},
	\]
	where $m_i$ is the size of $E^{(i)}$.
\end{claim}

\begin{proof}

	Recall that the additive error between $x$ and $\widetilde x$ is at most $\delta$, thus by Lemma~\ref{lem:additive-error}
	$$\big|Q_x(e)-Q_{\widetilde x}(e)\big|
	\le 4\delta\left(\sqrt{Q_x(e)}+\delta\right).$$

	We distinguish between two cases based on the size of $\sqrt{Q_x(e)}$ relative to $\delta$. First suppose that $\delta\le\epsilon\sqrt{Q_x(e)}/50$. This is the simpler case, because we immediately get
	$$\big|Q_x(e)-Q_{\widetilde x}(e)\big|
	\le \frac{4\epsilon}{50}\sqrt{Q_x(e)}\cdot\left(\sqrt{Q_x(e)}+\frac{\epsilon}{50}\sqrt{Q_x(e)}\right)
	\le \frac{\epsilon}{10} Q_x(e). $$

	Now consider the second case, $\delta\ge\epsilon\sqrt{Q_x(e)}/50$. This time we get
	$$\big|Q_x(e)-Q_{\widetilde x}(e)\big|
	\le4\delta\left(\frac{50\delta}{\epsilon} + \delta\right)
	\le \frac{204\delta^2}{\epsilon}.$$

	Based on the definition of $\delta$, let $u,v\in\widehat C_j^{(i)}$ be such that $u\approx_i v$ and $|x_u-x_v|=\delta$.
	By Definition~\ref{def:tr-cl} there must exist a sequence of vertices, or \emph{path},
	$u=w_0,\ldots,w_k=v$ such that $w_{\ell-1}\sim_i w_{\ell}$ for all $\ell \in [k]$ ($\sim_i$ is defined in Line~\ref{line:sim-def}).
	Without loss of generality, assume the path length $k$ is minimal and therefore at most $n$.
	Then by averaging, there exists $\ell'\in[k]$ such that  $|x_{w_{\ell'-1}}-x_{w_{\ell'}}| \ge\delta/k \ge\delta/n$.
	By definition of $\sim_i$, there exist $i'\le i-10\log n$ and $j'$ such that $w_{\ell'-1},w_{\ell'}\in C_{j'}^{(i')}$ but $w_{\ell'-1}\not\approx_{i'}w_{\ell'}$.

	We now wish to relate $Q(x)$ to $\delta^2$. We will accomplish this by lower bounding $Q(x)$ by the energy of $G_{j'}^{(i')}$ with respect to some vector $x'$, as per Lemma~\ref{lem:contract-compare}. Let us define $x'\in\mathbb R^{G_{j'}^{(i')}}$ as in Lemma~\ref{lem:contract-compare} such that the supernode of $w_{\ell'-1}$ retains the $x$-value of $w_{\ell'-1}$ and the supernode of $w_{\ell'}$ retains the $x$-value of $w_{\ell'}$. Formally,
	\begin{align*}
		x'_{[w_{\ell'-1}]}&=x_{w_{\ell'-1}},\\
		x'_{[w_{\ell'}]}&=x_{w_{\ell}},
	\end{align*}
	where $[w]$ denotes the equivalence class with respect to $\approx_{i'}$. All other coordinates of $x'$ are defined arbitrarily. By applying Lemma~\ref{lem:contract-compare} and discard unnecessary hyperedges we can conclude that
	$$Q(x)\ge Q'(x'),$$
	the energy of $x$ on $G_{j'}^{(i')}$.

To lower bound the energy of $x$ on $G_{j'}^{(i')}$, we can apply the hypergraph Cheeger inequality (Theorem~\ref{thm:hypergraph-cheeger}).
  We know by Guarantee~\ref{item:expansion} of \textsc{ExpanderDecomposition} that $\Phi\left(G_j^{(i')}\right)\ge\Omega(\tfrac{1}{r\log^2n})$.
	Let $m_{i'}$ be the number of hyperedges in $E^{(i')}$.
	Then by Guarantee~\ref{item:min-degree} of \textsc{ExpanderDecomposition}, the minimum degree of $G_{j'}^{(i')}$ is at least $m_{i'}/4n$. In general, the hypergraph Cheeger inequality states that when $x$ is centered, that is $\sum_{v\in V}x_vd(v)=0$, we have
	$$Q(x)\ge\frac{r\Phi^2}{32}\sum_v x_v^2 d(v).$$
	Our vector $x'$ is not centered with respect to the cluster $C_{j'}^{(i')}$. However, since the difference between $x_{[w_\ell]}$ and $x_{[w_{\ell+1}]}$ is at least $\delta/n$, at least one of them will have absolute value at least $\delta/(2n)$ even when $x$ \textit{is} centered.
	Therefore, the terms corresponding to $[w_{\ell'-1}]$ and $[w_{\ell'}]$ already give
	$$Q(x)\ge Q_{G_{j'}^{(i')}}(x)\ge \frac{r{\Omega\left(\tfrac1{r\log^2n}\right)}^2}{32}\cdot\frac{m_{i'}}{4n}\cdot{\left(\frac{\delta}{2n}\right)}^2.$$
	Putting these together, we have
	\begin{align*}
	\big|Q_x(e)-Q_{\widetilde x}(e)\big|&\le \frac{204\delta^2}{\epsilon}
	\le \frac{204}{\epsilon}\cdot{(2n)}^2\cdot\frac{32}{r\cdot\Omega{\left(\tfrac1{r\log^2n}\right)}^2}\cdot\frac{4n}{m_{i'}}\cdot Q(x)
	\le\frac{O(1)\cdot n^3r\log^4n}{\epsilon m_{i'}}\cdot Q(x)
	\end{align*}

	To relate this to $m_i$, as stated in the claim, recall that by Guarantee~\ref{item:size} of \textsc{ExpanderDecomposition} the total number of hyperedges decreases at least by a factor of $2$ during every iteration of the outer for-loop. So $i'\le i-10\log n$ implies that $m_i\le m_{i'}n^{-10}$, and we get
	$$\big|Q_x(e)-Q_{\widetilde x}(e)\big|
	\le\frac{O(\log^4n)\cdot n^3r}{\epsilon m_i n^{10}} \cdot Q(x)
	\le\frac{\epsilon Q(x)}{2n^3m_i},$$
	where the last inequality is because $r$, $\epsilon^{-1}$, and the $O(\log^4n)$ term are all smaller than $n$, by the assumptions of Lemma~\ref{lem:general-sparsification-spectral} and for large enough $n$.

	Putting the two cases together gives us the additive and multiplicative error terms and completes the proof of Claim~\ref{claim:g0}.
\end{proof}

Claims~\ref{claim:g1} and~\ref{claim:g3} follow as a result of this, while Claim~\ref{claim:g2} is a simple consequence of Observation~\ref{obs:contraction-sparsification}.

\claimgone*

\begin{proof}
	This follows immediately from Claim~\ref{claim:g0}, as $\widehat G_j^{(i)}$ has at most $m_i$ hyperedges,
	\[
	\Big|\widehat Q(x)-\widehat Q(\widetilde x)\Big|\le\sum_{e\in E_j^{(i)}}\Big|Q_x(e)-Q_{\widetilde x}(e)\Big|\le\sum_{e\in E_j^{(i)}}\left[\epsilon Q_x(e)+\frac{\epsilon Q_x(e)}{2n^3m_i}\right]\le\epsilon\widehat Q(x)+\frac{\epsilon Q(x)}{n^2}.
	\qedhere
	\]
\end{proof}

\claimgtwo*

\begin{proof}
	This follows immediately from Equation~\eqref{eq:general-main} because $\widetilde x$ is constant on all equivalence classes of $\approx_i$ in $C_j^{(i)}$.
\end{proof}

\claimgthree*

\begin{proof}
	The proof follows similarly to the proof of Claim~\ref{claim:g1}. However,
	$$\breve Q(x)=\sum_{e\in\widetilde E_j^{(i)}}w_e\cdot Q_x(e),$$
	so we must prove that the total weight $\sum_{e\in\widetilde E_j^{(i)}}w_e$ of the sparsifier is not too large. Note that this is not the same as the size of the sparsifier, which is guaranteed to be small by Lemma~\ref{lem:partition}.

	To do this, note that we can bound the total weight of hyperedges adjacent on a specific vertex, say $v$, by looking at the energy of the vector $\mathbbm1_v$, which has value $1$ on $v$ and $0$ everywhere else. (Here it is important that there are no self-loop hyperedges in $\widetilde E_j^{(i)}$). So we have
	\begin{align*}
	\sum_{e\in\widetilde E_j^{(i)}}w_e&\le\sum_{v\in C_j^{(i)}}\sum_{e\in\widetilde E_j^{(i)}:\ v\in e}w_e
	=\sum_{v\in C_j^{(i)}}Q_{\widetilde G_j^{(i)}}(\mathbbm1_v)
	\le\sum_{v\in C_j^{(i)}}(1+\epsilon)Q_{G_j^{(i)}}(\mathbbm1_v)
	\le2\sum_{v\in C_j^{(i)}}|\widetilde E_j^{(i)}|
	\le2nm_i.
	\end{align*}
	Therefore,
	\[
	\Big|\breve Q(x)-\breve Q(\widetilde x)\Big|\le\sum_{e\in\widetilde E_j^{(i)}}w_e\cdot\Big|Q_x(e)-Q_{\widetilde x}(e)\Big|\le\sum_{e\in\widetilde E_j^{(i)}}w_e\cdot\left[\epsilon Q_x(e)+\frac{\epsilon Q_x(e)}{2n^3m_i}\right]=\epsilon\breve Q(x)+\frac{\epsilon Q(x)}{n^2}.
	\qedhere
	\]
\end{proof}

\subsection{Weighted Hypergraphs and Proof of Theorem~\ref{thm:general-sparsification}}\label{subsec:weighted}

We have so far dealt only with unweighted hypergraphs, so as not to further complicate our algorithms and notation. However, our techniques extend essentially unchanged to weighted ones as well.

One way to see this is to replace weighted hyperedges with parallel hyperedges. Our proofs throughout Sections~\ref{sec:pre}--\ref{sec:general-sparsification} apply to hypergraphs that may contain hyperedges with multiplicity. Given a weighted graph, we can simply scale the weights up (or down if necessary) and approximate them with integer weights arbitrarily closely. A weighted hypergraph where the ratio between the largest and smallest weights is $w_{\max}/w_{\min}$ can be approximated to within a multiplicative $1\pm\eta$ error using $\log(\eta^{-1})\cdot w_{\max}/w_{\min}$ parallel hyperedges to replace each weighted hyperedge. The parameter $\eta$ can be set to $o(\epsilon)$ so as to still produce a good spectral approximation.

One might worry that this increases the running time since the number of hyperedges has technically increased. However, this turns out not to be the case: the running times of all of our algorithms scale polynomially with the number of \emph{distinct} hyperedges. Indeed, parallel edges can be consider simultaneously at every step. In the expander sparsification algorithm of Section~\ref{sec:expander-sparsification} the sampling probability of parallel hyperedges is the same, and at most $r(\epsilon^{-1}\log n)^{O(1)}$ of them are sampled. In Algorithm~\ref{alg:partition} of Section~\ref{sec:general-sparsification} each parallel instance of the same hyperedge gets cut by the same cuts and ends up in the same component. Consequently parallel edges end up on the same level in the same component in Algorithm~\ref{alg:main} of Section~\ref{sec:general-sparsification} and are sampled at the same rate.

In fact, one can verify that our proofs extend even more directly to weighted hypergraphs, without the need for approximating hyperedge weights by integers.

Until now, all of our algorithms have claimed only polynomial running time. Surprisingly it is the above extension that allows us to accelerate the runtime to nearly linear --- even in the case of unweighted input graphs.

\begin{proof}[Proof of Theorem~\ref{thm:general-sparsification}]
	Given a hypergraph $G$, we can apply the previously known hypergraph sparsification algorithm of~\cite{Soma2019} to get a polynomial (in $n$) size sparsifier in nearly linear (that is $O(mr^2) + n^{O(1)}$) time. We can then further sparsify this using our own Algorithm~\ref{alg:main} in time $n^{O(1)}$. Setting the error parameters of both algorithms to $\epsilon/3$ allows us to recover an $\epsilon$-spectral sparsifier of $G$, as desired.

Note that we may drop the $(r\log^6n)/n\le\epsilon$ requirement of Lemma~\ref{lem:general-sparsification-spectral} without loss of generality.
\end{proof}

%%% Local Variables:
%%% mode: latex
%%% TeX-master: "000-main"
%%% End:

%%% Local Variables:
%%% mode: latex
%%% TeX-master: "000-main"
%%% End:

%!TEX root=./000-main.tex

\section{Lower Bounds}\label{sec:lower-bound}

In this section we prove our space lower bound for an arbitrary compression of the cut structure of a hypergraph. In Section~\ref{sec:lower-bound-string} we introduce string compression, and reprove the corresponding lower bound result for completeness. In section~\ref{sec:lower-bound-construction} we construct our generic hard example in Theorem~\ref{thm:lower-bound-construction}. We then state Corollaries~\ref{corollary1},~\ref{corollary2}, and~\ref{corollary3} which result from applying Theorem~\ref{thm:lower-bound-construction} to various specific Ruzsa-Szemer\'edi graph constructions.

\subsection{String Compression}\label{sec:lower-bound-string}

A \emph{string compression scheme (SCS)} is an algorithm for compressing a long string into a short string, such that any subset sum query can be answered with small additive error.
Formally, we define it as follows.
\begin{definition}\label{def:SCS}
For positive integers $\ell,k$ and $\epsilon,g>0$, a pair of functions $\textsc{Encode}:\{0,1\}^\ell\to\{0,1\}^k$ and $\textsc{Decode}:\{0,1\}^k\times2^{[\ell]}\to\mathbb N$ is considered to be an $(\ell,k,\epsilon,g)$-SCS, if there exists a set of strings $\mathcal G\subseteq\{0,1\}^\ell$, such that the following holds.
\begin{itemize}
    \item $|\mathcal G|\ge g\cdot2^\ell$.
    \item For every string $s\in\mathcal G$ and every query $q\in 2^{[\ell]}$, $\left|\textsc{Decode}(\textsc{Encode}(s),q) - |s\cap q|\right|\le\epsilon\ell/2$.
\end{itemize}

\end{definition}

\begin{remark}
In general we use subsets of $[\ell]$ and elements of $\{0,1\}^\ell$ interchangeably. For instance, in the above definition, in $|s\cap q|$, $s$ is considered as a set.
\end{remark}

\begin{remark}
It is important that although a compression scheme may only work on a subset of strings ($\mathcal G$), it must work on \emph{all} queries. In fact, it is trivial to answer almost all queries on all inputs, by simply outputting $|q|\cdot|s|/\ell$.
\end{remark}

The lower following lower bound on the space requirement of string compression schemes has been known, and appears, for example, in~\cite{DinurN03}. We reprove it here for completeness.

\begin{theorem}\label{thm:string-compression}
Suppose $(\textsc{Encode}, \textsc{Decode})$ is an $(\ell,k,\epsilon,g)$-SCS, where $\epsilon\le1/10$. Then
$$k\ge\frac{\log g+3\ell/50}{\log2}-1.$$
\end{theorem}
\begin{proof}
We know that $\textsc{Encode}$ maps $\mathcal G$ into $\{0,1\}^k$. Therefore, by pigeonhole principle, there must be some set of inputs $\mathcal G_0$ of size at least $|\mathcal G|\cdot2^{-k}\ge g\cdot2^{\ell-k}$ that maps to the same output, say $c_0$. Let $s_0$ be an arbitrary string in $\mathcal G_0$.

Define $B_H(s_0,2\epsilon\ell)$ as the ball of radius $2\epsilon\ell$ in Hamming distance around $s_0$, that is, the set of strings $s \in \{0,1\}^\ell$ such that the number of coordinates where $s$ and $s_0$ differ is at most $2\epsilon\ell$.

\begin{claim}\label{claim:one}
$\mathcal G_0\subseteq B_H(s_0,2\epsilon\ell)$.
\end{claim}
\begin{proof}
Suppose there exists $s\in\mathcal G_0\backslash B_H(s_0,2\epsilon\ell)$, that is $s$ and $s_0$ differ on more than $2\epsilon\ell$ coordinates. Without loss of generality, we may assume that there are more than $\epsilon\ell$ coordinates where $s_0$ is $0$ but $s_1$ is $1$; let the set of such coordinates be $q$. By the definition of a string compression scheme
$$\textsc{Decode}(\textsc{Encode}(s_0),q)=\textsc{Decode}(c_0,q)\le|s_0\cap q|+\epsilon\ell/2=\epsilon\ell/2,$$
but
$$\textsc{Decode}(\textsc{Encode}(s),q)=\textsc{Decode}(c_0,q)\ge|s\cap q|-\epsilon\ell/2=|q|-\epsilon\ell/2>\epsilon\ell/2.$$
This is a contradiction.
\end{proof}

\begin{claim}\label{claim:two}
$|B_H(s_0,2\epsilon\ell)|<2^\ell\cdot2\exp\left(-\frac{\ell(1-4\epsilon)^2}{6}\right)$.
\end{claim}
\begin{proof}
Indeed,
\begin{align*}
    B_H(s_0,2\epsilon\ell)&=B_H(0^\ell,2\epsilon\ell)=2^\ell\cdot\mathbb P(w_H(x)\le2\epsilon\ell),
\end{align*}
where $x$ is a uniformly random vector in $\{0,1\}^\ell$. By Chernoff's bound
$$\mathbb P\left[w_H(x)\le2\epsilon\ell\right]\le\mathbb P\left[\left|w_H(x)-\frac{\ell}{2}\right|\ge\frac{\ell}2\left(1-4\epsilon\right)\right]\le2\exp\left(-\frac{\ell(1-4\epsilon)^2}{6}\right),$$
since $\epsilon\le1/4$, and the claim holds.
\end{proof}
Combining Claims~\ref{claim:one} and~\ref{claim:two} we get that
\begin{align*}
& g\cdot2^{\ell-k}\le2^\ell\cdot2\exp\left(-\frac{\ell{(1-4\epsilon)}^2}{6}\right), \\
\Rightarrow & \log g - k\log2 \le \log2 - \frac{\ell{(1-4\epsilon)}^2}{6}, \\
\Rightarrow & k\ge\frac{\log g -\log2+\frac{\ell{(1-4\epsilon)}^2}{6}}{\log 2}\ge\frac{\log g+3\ell/50}{\log2}-1,
\end{align*}
since $\epsilon\le1/10$.
\end{proof}

\begin{corollary}\label{cor:scs}
For $\ell\ge200$, there does not exist an $(\ell,k,1/10,1/2)$-SCS with $k<\ell/20$.
\end{corollary}

\subsection{Construction}\label{sec:lower-bound-construction}

We will derive a lower bound on $k$ from the existence of a \emph{Ruzsa-Szemer\'edi (RS) graph}, defined as follows.
\begin{definition}[Ruzsa-Szemer\'edi graph]
	We call an (ordinary) graph a $(t,a)$-RS graph if its edge set is the union of $t$ induced matchings of size $a$.
\end{definition}

Recall the definition of hypergraph cut sparsification schemes from Section~\ref{sec:tech-overview-lower-bound}:
\HCSSdef*

\begin{theorem}\label{thm:lower-bound-construction}
	Suppose there exists a $(t,a)$-RS graph on $n$ vertices where $a\ge6000\sqrt{n\log n}$ and $at\ge480n$. Then, any $(2n,t+1,k,\varepsilon)$-HCSS where $\varepsilon\le a/(60n)$ must have
	$$k = \Omega(at).$$
\end{theorem}
This is equivalent to Theorem~\ref{thm:lower-bound}.
\begin{proof}
	Let us fix such a $(t,a)$-RS graph $G$ on $n$ vertices, and a $(2n,t+1,k,\varepsilon)$-HCSS $(\textsc{Sparsify},\textsc{Cut})$.
	We will use this HCSS as a black box to construct a string compression scheme using $k$ bits of space, then bound $k$ by Corollary~\ref{cor:scs}.
	First let us convert $G$ into a bipartite graph $G'$. Let the vertex set of $G'=(V',E')$ be $V\times\{0,1\}$ where $P=V\times\{0\}$ and $Q=V\times\{1\}$ are the two sides of the bipartition. For each edge $e=(u,v)\in E$, we add two edges to $E'$: $((u,0),(v,1))$ and $((v,0),(u,1))$, ensuring that $G'$ is indeed bipartite.
	Note that $E'$ is the union of $t$ induced matchings of size $2a$. Let us call these matchings $M_1,\ldots M_t$ and let each $M_j$ be supported on $P_j$ in $P$ and $Q_j$ in $Q$. The maximum degree in $G'$ is $t$.

	We will use $G'$ to design a compression of strings of length $\ell = 2ta$. Note that there are exactly $2ta$ edges of $G'$. Let $\phi$ be an arbitrary bijection from $E'$ to $[\ell]$.
    For a string $s \in \{0,1\}^\ell$, Let $E_s$ be the subset of $E'$ defined as
	$$E_s=\left\{e\in E':s_{\phi(e)}=1\right\}.$$
	Thus the graph $G_s=(P\cup Q,E_s)$ encodes the string $s$. We then transform $G_s$ into the hypergraph $H_s=(P\cup Q,E^H_s)$. Let $E^H_s$ consist of one hyperedge corresponding to each vertex $u\in P$:
	$$E^H_s=\left\{\{u\}\cup \Gamma_s(u)\mid u\in P\right\},$$
	where $\Gamma_s$ denotes the neighborhood in $G_s$.

	Our compression function $\textsc{Encode}$ is then simply to sparsify $H_s$ using $\textsc{Sparsify}$.
	This can indeed be done, since $H_s$ is a hypergraph with $2n$ vertices and each edge has cardinality at most $t+1$. It remains to define the decoding function $\textsc{Decode}$.

	Given a query $q\subseteq[\ell]$, we must estimate the size of $s\cap q$, the number of coordinates of $s$ within $q$ having value $1$.
	To do this, we partition $q$ into segments $q^1,\ldots,q^t$, and then estimate the size of each $s\cap q^j$. Specifically, let
	$$q^j=\{i\in q \mid \phi^{-1}(i)\in M_j\}.$$
	We can then define
	$$\textsc{Decode}(\textsc{Sparsify}(H_s),q)=\sum_{j=1}^t\textsc{Decode}^j(\textsc{Sparsify}(H_s),q^j).$$
	Here $\textsc{Decode}^j$ remains undefined for now. In what follows we will define it such that
	$$\textsc{Decode}^j(\textsc{Sparsify}(H_s),q^j)\cong|s\cap q^{j}|.$$
	To estimate the size of $s \cap q^j$, we will observe the cut $E_s^H(S,\overline S) = E_s^H(S_s^j,\overline S_s^j)$ defined as follows:
	\begin{itemize}
		\item From $P$, $S$ contains the subset of vertices in $P_j$ corresponding to edges in $q^j$. Formally
		$$S\cap P = \{P\cap e \mid e\in M_j \text{ s.t.\ } \phi(e)\in q^j\}.$$
		\item From $Q$, $S$ contains all vertices except $Q_j$.
	\end{itemize}
	We will prove the the size of the cut $(S,\overline S)$ is closely related to the size of $s\cap q$, as long as $s$ satisfies some nice properties.

	Note that each hyperedge in $E_s^H$ corresponds to a vertex in $P$: for $u\in P$ we denote the hyperedge $\{u\}\cup\Gamma_s(u)$ as $e_u$. We will bound the contribution of $e_u$ to the cut $(S,\overline S)$ for all $u$ in each of the following three categories:
	\begin{enumerate}
		\item $u\in P_j\cap S$:\\
			Let the edge from $M_j$ adjacent on $u$ be $f_u$. For any such $u$, $e_u$ crosses the cut if and only if $s_{\phi(f_u)}=1$. Indeed, if $s_{\phi(f_u)}=1$, then $f\in E_s$ and $f\cap Q\in Q_j\subseteq \overline S$. On the other hand, if $s_{\phi(f_u)}=0$ then $f\not\in E_s$ and all edges adjacent on $u$ in $E_s$ correspond to matchings different from $M_j$ (that is $M_k$ for $k\neq j$). Since $M_j$ is induced by the property of RS-graphs, $\Gamma_s(u)\subseteq Q\setminus Q_j\subseteq S$.
			Therefore, the total amount of hyperedges crossing the cut from this category is exactly $|s\cap q^j|$.

		\item $u\in P_j\setminus S$:\\
			In this case $e_u$ crosses the cut unless $d_s(u)<2$. Indeed, if $d_s(u)\ge2$ then at least one edge adjacent on $u$ in $G_s$ does \textit{not} come from $M_j$. The other endpoint of this edge is in $Q\setminus Q_j\subseteq S$, whereas $u$ itself is in $\overline S$ by definition. In the the case where $d_s(u)<2$ we cannot say whether $e_u$ crosses the cut or not.
			Therefore, the number of hyperedges crossing the cut from this category is approximately $m-|q^j|$ (that is all of them), but with a possible error of
			$$|\{u\in P_j \mid d_s(u)<2\}|.$$
		\item $u\in P\setminus P_j$:\\
			In this case we cannot say anything about the number of edges crossing the cut, except that it is unlikely to deviate from its expectation when $s$ is considered to be uniformly random on $\{0,1\}^\ell$. Let
			$$Z_j=|\{u\in P\setminus P_j \mid \Gamma_s(u)\not\subseteq Q_j\}|,$$
			or the number of hyperedges in $E_s^H$ from this category crossing the cut.
	\end{enumerate}
	Overall, we can approximate the size of the cut $(S,\overline S)$ in $H_s$ by
	\begin{equation}\label{eq:cut-approx-a}
	|s\cap q^j| + (a-|q^j|) + \mathbb E_s Z_j,
	\end{equation}
	with an maximum additive error of
	\begin{equation}\label{eq:cut-approx-b}
	\big|\{u\in P_j \mid d_s(u)<2\}\big| + \big|Z_j-\mathbb E_s Z_j\big|.
	\end{equation}

	Conversely, this allows us to approximate $|s\cap q^j|$ using the size of the same cut in our $(2n,t+1,k,\varepsilon)$-HCSS\@.
    Therefore, we define $\textsc{Decode}^j$ as follows:
	\begin{equation}\label{eq:decode-def}
	\textsc{Decode}^j(\textsc{Sparsify}(H_s),q^j) = \textsc{Cut}(\textsc{Sparsify}(H_s),S)-(a-|q^j|)-\mathbb E_s Z_j.
	\end{equation}
	It remains to bound the total error introduced by the inaccuracies above.

	We will define the set of good input strings, $\mathcal G$ to be those where this additive error is small across all $j$'s, and we will prove that this contains a majority of possible input strings.

	\begin{claim}\label{cl:good-bound}
		Let $\mathcal G$ be the set of strings $s\in\{0,1\}^\ell$ such that
		$$\sum_{j=1}^t\left( \big|\{u\in P_j \mid d_s(u)<2\}\big| + \big|Z_j-\mathbb E_s Z_j\big|\right)\le8n + 100t\sqrt{n\log n}.$$
		Then $|\mathcal G|\ge2^{\ell-1}$.
	\end{claim}

	\begin{proof}
		Consider $s$ to be a random string, chosen uniformly on $\{0,1\}^\ell$.
        We will prove that $\mathbb P[s\in\mathcal G]\ge1/2$. We do this by considering the two bad events
        \begin{align*}
        	&\sum_{j=1}^t\left|\left\{u\in P_j|d_s(u)<2\right\}\right|>8n,\\
        	&\sum_{j=1}^t\left|Z_j-\mathbb EZ_j\right|>100t\sqrt{n\log n},
        \end{align*}
        and prove that neither happens with probability more than $1/4$.

		To bound the probability of the first event consider the expectation of the sum:
		\begin{align*}
		& \mathbb E\sum_{j=1}^t\big|\{u\in P_j \mid d_s(u)<2\}\big|=\mathbb E\sum_{j=1}^t\sum_{u\in P_j}\mathbbm1(d_s(u)<2)
		=\sum_{j=1}^t\sum_{u\in P_j}\mathbb P[d_s(u)<2]\\
		&=\sum_{u\in P}\sum_{j:u\in P_j}\mathbb P[d_s(u)<2]
		=\sum_{u\in P}|\{j \mid u\in P_j\}|\cdot\mathbb P[d_s(u)<2]\\
		&=\sum_{u\in P}d(u)\cdot(d(u)+1)\cdot2^{-d(u)}
		\le2n,
		\end{align*}
		as the function $d(d+1)\cdot2^{-d}$ is bounded by $2$ for all non-negative $d$.

		This means, that by the Markov inequality
		$$\mathbb P\left[\sum_{j=1}^t\big|\{u\in P_j \mid d_s(u)<2\}\big|>8n\right]\le \frac{1}{4}.$$

		Now, for the second bad event, we apply Chernoff bound (Theorem~\ref{thm:chernoff}).
        Note that $$Z_j=|\{u\in P\setminus P_j \mid \Gamma_s(u)\not\subseteq Q_j\},$$ is
		the sum of $n-m$ independent random variables bounded by one. Therefore,
		$$\mathbb P\left[|Z_j-\mathbb EZ_j|>\delta n\right]\le2\exp\left(-\frac{\delta^2 n}{3}\right).$$
		Setting $\delta$ to $100\sqrt{(\log n)/n}$ and taking union bound over $j = 1,\ldots,t$ gives us that
		$$\mathbb P\left[\sum_{j=1}^t|Z_j-\mathbb EZ_j|>100t\sqrt{n\log n}\right]\le \frac{1}{4}.$$

		Putting the bounds on the first and second event together gives us the statement of the claim.
	\end{proof}

	This $\mathcal G$ will be our set of good inputs in our $(\ell,k,1/10,1/2)$-SCS\@.
    Claim~\ref{cl:good-bound} essentially shows that the error in our estimate of $|s\cap q|$ would be at most $8n + 100t\sqrt{n\log n}$ without the inaccuracy introduced by our cut sparsifier. Since the size of the cut $(S,\overline S)$ is at most $n$ (the total number of hyperedges in the hypergraph $H_s$), this introduces an additional $\varepsilon n$ additive error.

	Formally, when $s\in\mathcal G$
	\begin{align*}
	&\big||s\cap q-\textsc{Decode}(\textsc{Encode}(s),q)\big|\\
	=&\left||s\cap q|-\sum_{j=1}^t\textsc{Decode}^j(\textsc{Sparsify}(H_s),q^j)\right|\\
	\le&\sum_{j=1}^t\big||s\cap q^j|-\textsc{Decode}^j(\textsc{Sparsify}(H_s),q^j)\big|\\
\le&\sum_{j=1}^t\big||s\cap q^j|-\textsc{Cut}(\textsc{Sparsify}(H_s),S)+(a-|q^j|)+\mathbb E_sZ_j\big|~~~~~~~~~~~~~~~~~~~~~~~~~~~~~~~~~~~\text{by equation~\ref{eq:decode-def}}\\
	\le&\sum_{j=1}^t\left(\big||s\cap q^j|-|E^H_s(S_s^j,\overline S_s^j)|+(a-|q^j|)+\mathbb E_s Z_j\big|+\big||E_s^H(S_s^j,\overline S_s^j)|-\textsc{Cut}(\textsc{Sparsify}(H_s),S_s^j)\big|\right)\\
	\le&\sum_{j=1}^t\Big(\big|\{u\in P_j \mid d_s(u)<2\}\big| + \big|Z_j-\mathbb E_s Z_j\big|+\varepsilon n\Big)~~~~~~~~~~~~~~~~~~~~~~~~~~~~~~~~~~~~~~~\text{by equations~\ref{eq:cut-approx-a} and~\ref{eq:cut-approx-b}}\\
	\le&\left(8n + 100t\sqrt{n\log n}\right) + \sum_{j=1}^t\varepsilon n~~~~~~~~~~~~~~~~~~~~~~~~~~~~~~~~~~~~~~~~~~~~~~~~~~~~~~~~~~~~~~~~~~~~~~~~~~~~\text{since $s\in\mathcal G$}\\
	\le&8n + 100t\sqrt{n\log n}+\varepsilon tn
	\end{align*}

	This is less than $\ell/20=at/20$ due to the theorem's assumptions on the parameters. Therefore, $(\textsc{Encode},\textsc{Decode})$ is a $(\ell,k,1/10,1/2)$-SCS with the set of good inputs being $\mathcal G$. By Corollary~\ref{cor:scs} $k$ must be at least $\Omega(\ell)=\Omega(at)$.
\end{proof}

We can now apply Theorem~\ref{thm:lower-bound-construction} to several RS-graph constructions known in the literature. Note that if there exists a $(t,a)$-RS graph, one can always reduce the parameters to get an $(t',a')$-RS graph for $t'\le t$ and $a'\le a$. We begin with Fischer et al.~\cite{Fischer2002}, which proves the existence of $(n^{\Omega(1/\log\log n)},n/3-o(1))$-Ruzsa-Szemer\'edi graphs, resulting in the following corollary.
\begin{corollary}\label{corollary1}
	Any $(n,r,k,\varepsilon)$-HCSS with $r=n^{O(1/\log\log n)}$ and small constant $\varepsilon$ requires $k=\Omega(nr)$ space.
\end{corollary}
In other words, any data structure that can provide a $(1+\varepsilon)$-approximation to the size of all cuts in an $r$-uniform hypergraph with $n$ vertices and $r=n^{O(1/\log\log n)}$ for small constant $\epsilon\in (0, 1)$ requires $\Omega(n r)$ bits of space.
This is tight due to the hypergraph cut sparsifier construction of~\cite{Chen20}. A different construction, also from~\cite{Fischer2002}, is able to achieve an $(n^c,n/O(\sqrt{\log\log n/\log n}))$-RS graph for some small enough constant $c$. This results in the following:
\begin{corollary}\label{corollary2}
	For some constant $c$, any $(n,r,k,\varepsilon)$-HCSS with $r=O(n^c)$ and $\varepsilon=O(\sqrt{\log\log n/\log n})$ requires $k=\Omega(nr/\sqrt{\log n/\log\log n})$ space.
\end{corollary}
Finally, the original construction of Ruzsa and Szemeredi~\cite{ruzsa1978triple} guarantees the existence of an $(n/3,n/2^{O(\sqrt{\log n})})$-RS graphs, implying:
\begin{corollary}\label{corollary3}
	Any $(n,r,k,\varepsilon)$-HCSS with $\varepsilon = 2^{-\Omega(\sqrt{\log n})}$ requires $k=nr/2^{O(\sqrt{\log n})}$ space.
\end{corollary}
These results imply that for any value of $r$, it is impossible to compress the cut structure of a hypergraph with $n$ vertices and maximum hyperedge size $r$, with significantly less than $nr$ space, and a polynomial scaling in the error (that is with $nr^{1-\Omega(1)}\varepsilon^{-O(1)}$ space).

%%% Local Variables:
%%% mode: latex
%%% TeX-master: "000-main"
%%% End:
%!TEX root=./000-main.tex

\section{Spectral Sparsification of Directed Hypergraphs}\label{sec:directed-hypergraph-sparsification}
In this section, we discuss spectral sparsification of directed hypergraphs.
First we introduce some notions and study basic properties of directed hypergraphs in Section~\ref{subsec:directed-preliminaries}.
Then, we discuss spectrally sparsifying directed hypergraphs with hyperedges having nearly equal overlap (a concept to be defined in Section~\ref{subsec:directed-preliminaries}).
Finally, we prove Theorem~\ref{thm:directed-hypergraph-sparsification} in Section~\ref{subsec:directed-hypergraph-sparsification}.

\subsection{Preliminaries}\label{subsec:directed-preliminaries}

A \emph{directed hypergraph} $G=(V,E)$ is a pair of a vertex set $V$ and a set $E$ of hyperarcs, where a \emph{hyperarc} $e \in E$ is an ordered pair of two disjoint vertex sets $h(e) \subseteq V$, the \emph{head}, and $t(e) \subseteq V$, the \emph{tail}.
The \emph{size} of a hyperarc $e \in E$ is $|h(e)| + |t(e)|$. We restrict ourselves to dealing with only \textit{simple} directed hypergraphs, that is, in Section~\ref{sec:directed-hypergraph-sparsification} $E$ is always considered to be a set as opposed to a multiset.

We say that a vertex set $S \subseteq V$ \emph{cuts} a hyperarc $e \in E$ if $S \cap t(e) \neq \emptyset$ and $(V \setminus S) \cap h(e) \neq \emptyset$.
The \emph{energy} of a hyperarc $e$ with respect to a vector $x \in \mathbb{R}^V$ is defined as
$$\max_{a\in t(e),b\in h(e)}(x_a-x_b)_+^2,$$
where $(\alpha)_+ = \max\{\alpha,0\}$.
The energies of a set of arcs, or of an entire vector with respect to $G$, is defined identically to the undirected case. So in particular the energy of $x$ with respect to $G$ is
\[
Q(x) = \sum_{e \in E}\max_{a\in t(e), b\in h(b)}{(x_a - x_b)}_+^2.
\]

Note that $Q(1_S)$, where $1_S \in \mathbb{R}^V$ is the characteristic vector of $S$, is equal to the number of hyperarcs cut by $S$.
Identically to Definition~\ref{def:spectral-sparsifier}, for $\epsilon > 0$, a weighted subgraph $\widetilde{G}$ of $G$ is said to be a $\epsilon$-spectral sparsifier of $G$ if
\[
\widetilde Q(x) = (1\pm\epsilon)Q(x),
\]
where $Q(x)$ and $\widetilde{Q}(x)$ are energy of $x$ with respect to $G$ and $\widetilde{G}$, respectively.

In constructing our sparsifier, a useful object to consider will be the \emph{clique graph} of $G$, the directed (ordinary) multigraph we get by replacing each hyperarc in $G$ with a directed bipartite clique. Formally, the clique of a hyperarc $e \in E$ is the set of arcs $C(e)=\{(a,b) \mid a\in t(e),b\in h(e)\}$.
The clique graph of a set of hyperarcs $E'\subseteq E$ is the multi-union of the individual cliques
$C(E')=\biguplus_{e\in E'}C(e)$.
Finally, the clique graph of $G$ itself is $C(G)=(V,C(E))$.
In the following, we make some observation about the multiplicities of arcs in the clique graph.

\begin{definition}
Given a hypergraph $G=(V,E)$, we say that a subset of hyperarcs $E'\subseteq E$ $k$-{\emph{overlapping}} if every arc in $C(E')$ appears with multiplicity at least $k$.
Furthermore, the \emph{overlap} $k(e)$ of a single hyperarc $e \in E$ is defined as the largest $k$ such that there exists a $k$-overlapping set of hyperarcs containing $e$.
\end{definition}

Informally, we will use the inverse overlap of each hyperarc as a sampling rate in constructing our sparsifier. Thus, the following lemma will be a useful bound on the sum of these rates:

\begin{lemma}\label{lem:overlap}
	Let $G=(V,E)$ be a directed hypergraph. Then, we have
	$$\sum_{e\in E}\frac1{k(e)}\le n^2.$$
\end{lemma}

\begin{proof}
	Consider the following simple algorithm:
	\begin{algorithm}[H]
		\caption{}\label{alg:overlap}
		\begin{algorithmic}[1]
			\Procedure{OverlapPeeling}{$G=(V,E)$}
			\State $E' \gets E$.
			\For{$k=1,\ldots,2^{n-2}$}\label{line:dp-for}
			\State $E'_k\gets E'$.
			\While{there exists $(u,v)\in C(E')$ with multiplicity at most $k$}\label{line:dp-while}
			\For{all hyperarcs $e \in E'$ such that $(u,v)\in C(e)$}
			\State $f(e)\gets(u,v)$.
			\State $E'\gets E'\setminus \{e\}$.
			\EndFor
			\EndWhile
			\EndFor
			\EndProcedure
		\end{algorithmic}
	\end{algorithm}
This algorithm iterates through all possible overlaps (from $1$ to $2^{n-2}$) and peels off all hyperarcs with this overlap, until no hyperarcs remain. The algorithm maintains several variables ($E'_k$ and $f(e)$) that are not used. However, these will be useful in proving the lemma.

\begin{claim}
	The set $E'_k$ has overlap $k$ for all $k$.
\end{claim}\label{claim:ek-level}
\begin{proof}
	Indeed, the variable $k$ is augmented in the for-loop at Line~\ref{line:dp-for} only after exiting the while-loop at Line~\ref{line:dp-while}. This means that there no longer existed any pairs $(u,v)$ in $C(E')$ with multiplicity at most $k-1$, and therefore $E'$ was $k$-overlapping. (The exception to this argument is $k=1$, however all sets are $1$-overlapping by definition.)
\end{proof}

\begin{claim}\label{cla:overlap-k^*}
	If a hyperarc $e$ is removed at a time when $k=k^*$, then it has overlap exactly $k^*$.
\end{claim}
\begin{proof}
It is easy to see that $e$ has overlap \textit{at least} $k^*$, since it was an element of $E'_{k^*}$ which is itself $k^*$-overlapping by Claim~\ref{claim:ek-level}.

We prove that $e$ has overlap \textit{at most} $k^*$ by induction.
By induction, we can assume that all hyperarcs removed before $e$ had overlap corresponding to the value of $k$ at the time, that is, at most $k^*$.
Let $E^*$ be the current value of $E'$ at the time just before $e$ is removed. Suppose for contradiction that $e$ is at least $(k^*+1)$-overlapping, or equivalently there exists a $(k^*+1)$-overlapping set containing $e$, say $\widetilde E_{k^*+1}$.
However, no hyperarc removed before $e$ could be in this set, since we know they are at most $k^*$-overlapping. So $\widetilde E_{k^*+1}\subseteq E^*$. But some arc in $C(e)$ has multiplicity only at most $k^*$ in $E^*$, which is a contradiction.
\end{proof}

\begin{claim}
	For any pair $(u,v)\in V^2$, we have
	$$\sum_{e:f(e)=(u,v)}\frac1{k(e)}\le1.$$
\end{claim}

\begin{proof}
	First note that all pairs $(u,v)$ are only considered once in the while-loop of Line~\ref{line:dp-while} throughout the whole algorithm. Indeed, once a pair is considered, all  hyperarcs containing it are removed and $(u,v)$ is no longer in $C(E')$.

	Suppose $(u,v)$ is removed in this way when $k=k^*$.
	Then all hyperarcs $e$ such that $f(e)=(u,v)$ have overlap at most $k^*$. On the other hand, there are at most $k^*$ such hyperarcs due to the condition in Line~\ref{line:dp-while}. This concludes the proof of the claim.
\end{proof}

From here the lemma statement follows simply:
\begin{align*}
\sum_{e\in E}\frac1{k(e)}&=\sum_{(u,v)\in V^2}\sum_{e:f(e)=(u,v)}\frac1{k(e)}
\le\sum_{(u,v)\in V^2}1
=n^2. \qedhere
\end{align*}
\end{proof}
\begin{remark}\label{rem:overlap-time-complexity}
	Note that, by Claim~\ref{cla:overlap-k^*}, we can compute overlaps of hyperarcs by running Algorithm~\ref{alg:overlap}.
	Furthermore, we can make it run in polynomial time by, instead of incrementing $k$ at Line~\ref{line:dp-for}, updating $k$ to be the smallest multiplicity of an edge in $C(E')$.
\end{remark}

\subsection{Nearly Equally Overlapping Directed Hypergraphs}\label{subsec:nearly-equal-densities}

In this section, we consider the simpler case where every hyperarc has a similar overlap.
\begin{lemma}\label{lem:directed-main}
	There is an algorithm that, given $0<\epsilon\le1/2$ and a directed hypergraph $G=(V,E)$ such that every hyperarc has overlap between $k$ and $2k$ for some $k \geq 1$, and each hyperarc has size at most $r\le\sqrt{\epsilon n}/11$, outputs in polynomial time a weighted subgraph $\widetilde G = (V,\widetilde E,w)$ of $G$ satisfying the following with probability $1-O(1/n)$:
	\begin{itemize}
	\item $\widetilde{G}$ is an $\epsilon$-spectral sparsifier of $G$,
	\item $|\widetilde{E}| = O(n^2r^2\log n/\epsilon^2)$.
	\end{itemize}
\end{lemma}

\paragraph{Construction}
Let us construct $\widetilde G=(V,\widetilde E)$ by sampling each hyperarc independently with the same probability $p=1000r^2\log n/(k\epsilon^2)$ and scaling them up by $1/p$. Let the weight of each hyperarc $e$ in $\widetilde G$ be denoted as $w_e$. Then $w_e$ is an independent random variable taking value $1/p$ with probability $p$ and value $0$ otherwise, for each $e$.

Clearly, we can compute the output in $O(m)$ time.
Also, we can bound the size of $\widetilde E$ easily:
\begin{lemma}\label{lem:d-size-bound}
	We have $\mathbb E[|\widetilde E|] = 2000n^2r^2 \log n/\epsilon^2$ and
	\[
		\mathbb P\left[|\widetilde E| > 4000n^2r^2 \log n/\epsilon^2\right]\le2\exp\left(-\frac{2pkn^2}{3}\right).
	\]
\end{lemma}
\begin{proof}
Note that since the overlap of each hyperarc is at most $2k$, there are at most $2kn^2$ hyperarcs in total (in $E$) by Lemma~\ref{lem:overlap}.
Each hyperarc is sampled with probability $p$ to be in $\widetilde E$, so $\mathbb E[|\widetilde E|] = 2pkn^2=2000n^2r^2 \log n/\epsilon^2$, as claimed. By Chernoff bounds (Theorem~\ref{thm:chernoff}), the claimed concentration inequality holds.
\end{proof}

\paragraph{Correctness}
We now examine the spectral properties of $\widetilde G$. Recall that $C(G)$ is the clique graph of $G$. Let us denote by $Q^C$ the energy with respect to the clique graph. We may assume without loss of generality that $Q^C(x)=1$, since whether $Q(x) = (1\pm \epsilon)\widetilde{Q}(x)$ holds or not is unaffected by scaling $x$. Define $\overline{\mathbb R^V}$ to be the set of vectors $x$ such that this is satisfied. Note that this means that $Q_x(E)\ge1/r^2$. Indeed
\begin{align*}
	Q_x(E)&=\sum_{e\in E}\max_{u\in t(e),\ v\in h(e)}{(x_u-x_v)}_+^2
	=\sum_{e\in E}\max_{f\in C(e)}Q^C_x(f)
	\ge\frac1{r^2}\sum_{e\in E}\sum_{f\in C(E)}Q_x^C(f)\\
	&=\frac1{r^2}\sum_{f\in C(E)}Q_x^C(f)
	=\frac{Q^C(x)}{r^2}
	=\frac{1}{r^2}.
\end{align*}

Let us categorize the arcs in $C(E)$ based on their contributions to the total energy $Q^C(x)=1$ in $C(G)$.
The categories are
$$C_i=\left\{f\in C(E)\ \middle|\ Q^C_x(f)\in\left(\frac{2^{-i}}{k},\frac{2^{-i+1}}{k}\right]\right\},$$
for $i=1,\ldots,i^*$ where $i^*:=\lceil3\log n\rceil$, as well as
$$C_*=\left\{f\in C(E)\middle|\ Q^C_x(f)\le \frac{2^{-i^*}}{k} \right\}.$$
Recall that $C(E)$ is a multiset, and consequently so are $C_i$ and $C_*$. Since each arc $f$ appears with multiplicity at least $k$, any single arc can contribute at most $1/k$ to the energy. Therefore, all arcs of $C(G)$ are covered by these categories.

We then partition the hyperarcs into similar categories: A hyperarc $e$ gets into category $i$ (or $E_i$) if $i$ is the smallest number for which $C(e)$ contains an arc in $C_i$. Formally
\begin{align*}
	E_i&=\left\{e\in E\ \middle|\ i=\max\{j \mid C(e)\cap C_j\neq\emptyset\}\right\}\;(i =1,\ldots,i^*), \text{ and}\\
	E_*&=\left\{e\in E\ \middle|\ C(e)\subseteq C_*\right\}.
\end{align*}

To prove that $\widetilde{G}$ is an $\epsilon$-spectral sparsifier, we will show that, for all $i$, $Q_x(E_i)\approx \widetilde Q_x(E_i)$. Similarly to the proof of Theorem~\ref{thm:expander-sparsification} we will introduce a discretization of $Q_x(E_i)$. However, unlike in the proof of Theorem~\ref{thm:expander-sparsification}, instead of rounding the vertex potentials $x_v$, we will round the energies of hyperarcs, that is, $Q_x(e)$ for $e \in E$.

Let us first define $Q^{C,(i)}_x(f)$, the rounding of $Q^C_x(f)$. Firstly, if $Q_x^C(f)\le2^{-i}/k$, that is the arc $f$ is not relevant to $E_i$, we define $Q_x^{C,(i)}(f)$ to be zero. Otherwise, let $Q^{C,(i)}_x(f)$ be the rounding of $Q^C_x(f)$ to the nearest integer multiple of $1/(kn^3)$. Analogously with the definition of $Q_x$, for $e\in E$ let
\begin{align*}
	Q^{(i)}_x(e)&=\max_{f\in C(e)}Q^{C,(i)}_x(f), \quad
	Q^{(i)}_x(E')=\sum_{e\in E'}Q^{(i)}_x(E'),\\
	\widetilde Q^{(i)}_x(e)&=w_e Q^{(i)}_x(e), \quad
	\widetilde Q^{(i)}_x(E')=\sum_{e\in E'}\widetilde Q^{(i)}_x(e).
\end{align*}

Informally, we prove the following chain of approximations for each $i$:
$$Q_x(E_i)\cong Q_x^{(i)}(E_i)\cong\widetilde Q_x^{(i)}(E_i)\cong\widetilde Q_x(E_i),$$
as well as
$$Q_x(E_*)\cong\widetilde Q_x(E_*).$$

We make this formal in the following claims:

\begin{restatable}{claim}{claimdone}\label{claim:d1}
	For all $x \in \mathbb{R}^V$ and all $i=1,\ldots,i^*$,
	$$Q_x^{(i)}(E_i)=Q_x(E_i)\pm\frac{2}{n}.$$
\end{restatable}

\begin{restatable}{claim}{claimdtwo}\label{claim:d2}
	For all $i=1,\ldots,i^*$,
	$$\mathbb P\left[\forall x \in \overline{\mathbb{R}^V}:\ \widetilde Q_x^{(i)}(E_i)=\left(1\pm\frac{\epsilon}{2}\right)Q_x^{(i)}(E_i)\pm\frac{\epsilon Q(x)}{10\log n}\right]\ge1-\frac1n.$$
\end{restatable}

\begin{restatable}{claim}{claimdthree}\label{claim:d3}
	For all $i=1,\ldots,i^*$,
	$$\mathbb P\left[\forall x \in \overline{\mathbb{R}^V}:\ \widetilde Q_x^{(i)}(E_i)=\widetilde Q_x(E_i)\pm\frac{4}{n}\right]\ge1-\frac1n.$$
\end{restatable}

\begin{restatable}{claim}{claimdfour}\label{claim:d4}
	$$\mathbb P\left[\forall x \in \overline{\mathbb{R}^V}:\ \widetilde Q_x(E_*)=Q_x(E_*)\pm\frac{6}{n}\right]\ge1-\frac1n.$$
\end{restatable}

Before proving the above claims, which we do in the next section, we conclude the analysis of correctness of the sparsifier.
\begin{lemma}\label{lem:d-spectral-property}
	The directed hypergraph $\widetilde{G}$ is an $\epsilon$-spectral sparsifier of $G$ with probability $1-O(1/n)$.
\end{lemma}
\begin{proof}
	The statements of Claims~\ref{claim:d2},~\ref{claim:d3}, and~\ref{claim:d4} all hold with high probability. Let us consider the event that they all hold simultaneously, then by Claims~\ref{claim:d1},~\ref{claim:d2}, and~\ref{claim:d3},
	\begin{align*}
		& \big|Q_x(E_i)-\widetilde Q_x(E_i)\big| \le\big|Q_x(E_i)-Q_x^{(i)}(E_i)\big|+\big|Q_x^{(i)}(E_i)-\widetilde Q^{(i)}_x(E_i)\big|+\big|\widetilde Q_x^{(i)}(E_i)-\widetilde Q_x(E_i)\big|\\
		&\le\frac{2}{n}+\frac{\epsilon}{2} Q_x^{(i)}(E_i)+\frac{Q(x)}{10\log n}+\frac4n
		\le \frac{\epsilon}{2} Q_x(E_i) + \frac{Q(x)}{10\log n}+\frac6n,
	\end{align*}
	using that $\epsilon\le1$. Summing this over $i=1\ldots,i^*=\lceil3\log n\rceil$ and adding Claim~\ref{claim:d4} we get
	\begin{align*}
		& \big|Q(x)-\widetilde Q(x)\big| \le\sum_{i=1}^{i^*}\big|Q_x(E_i)-\widetilde Q_x(E_i)\big|+\big|Q_x(E_*)-\widetilde Q_x(E_*)\big| \\
		& \le\sum_{i=1}^{i^*}\left[\frac{\epsilon}{2} Q_x(E_i)+\frac{\epsilon Q(x)}{10\log n}+\frac6n\right]+\frac{6}{n}
		 \le \frac{\epsilon}{2}Q(x)+\frac{4\epsilon }{10}Q(x) + \frac{12}{n}
		\le\epsilon Q(x),
	\end{align*}
	since $\epsilon Q(x)/10\ge\epsilon/(10r^2)\ge12/n$ because $11r\le\sqrt{\epsilon n}$.
\end{proof}
Lemma~\ref{lem:directed-main} follows by Lemmas~\ref{lem:d-size-bound} and~\ref{lem:d-spectral-property} and a union bound.

\subsection{Proofs of Claims~\ref{claim:d1},~\ref{claim:d2},~\ref{claim:d3}, and~\ref{claim:d4}}\label{subsubsec:nearly-equal-densities-claims}

We begin with a preliminary claim examining the difference between $Q_x$ and $Q_x^{(i)}$ on a single hyperarc.
\begin{claim}\label{claim:d0}
	For all $x \in \mathbb{R}^V$, all $i=1,\ldots,i^*$, and any hyperarc $e\in E_i$,
	$$Q_x^{(i)}(e) = Q_x(e)\pm\frac{1}{kn^3}.$$
\end{claim}

\begin{proof}
	Suppose first that $Q_x(e)\ge Q_x^{(i)}(e)$. Recall that $e\in E_i$ and by definition of $E_i$ there exist arcs in $C(e)\cap C_i$. In this case let $f=\argmax_{f\in C(e)}Q_x^C(f)$, guaranteeing that $f\in C_i$. Therefore, by definition $Q_x^{C,(i)}(f)$ is not zero, but a rounding to the nearest integer multiple of $1/(kn^3)$. Therefore,
	\begin{align*}
		Q_x(e)-Q_x^{(i)}(e)&\le Q_x^C(f)-Q_x^{C,(i)}(f)\le\frac1{kn^3}.
	\end{align*}

	Now suppose that $Q_x(e)< Q_x^{(i)}(e)$. In this case we define $f$ to be $\argmax_{f\in C(e)}Q_x^{C,(i)}(e)$. Now if $Q_x^{C,(i)}(f)$ is zero the claim holds trivially, so we may assume that $Q_x^{C,(i)}(f)$ is instead a rounding to the nearest integer multiple of $1/(kn^3)$:
	\begin{align*}
	Q_x^{(i)}(e)-Q_x(e)&\le Q_x^{C,(i)}(f)-Q_x^C(f)\le\frac1{kn^3}. \qedhere
	\end{align*}
\end{proof}

\claimdone*

\begin{proof}%[Proof of Claim~\ref{claim:d1}]
	We can simply sum over the hyperarcs in $E_i$. Since $|E_i|\le |E|\le 2kn^2$, we have that
	\[
	\big|Q_x(E_i)-Q_x^{(i)}(E_i)\big|\le\sum_{e\in E_i}\big|Q_x(e)-Q_x^{(i)}(e)\big|\le\frac{|E_i|}{kn^3}\le\frac2n. \qedhere
	\]
\end{proof}

\claimdtwo*

\begin{proof}%[Proof of Claim~\ref{claim:d2}]
	We prove the stronger claim
$$\mathbb P\left[\forall x \in \overline{\mathbb{R}^V}:\ \widetilde Q_x^{(i)}(E_i)=\left(1\pm\frac{\epsilon}{2}\right)Q_x^{(i)}(E_i)\pm\frac{\epsilon}{10r^2\log n}\right]\ge1-\frac1n,$$
	replacing $Q(x)$ by $1/r^2$ in the allowable additive error, which depends on $x$ only through $Q_x^{(i)}$ and $E_i$.

	We first consider a single setting of $x$ (and consequently $E_i$ and $Q_x^{(i)}$). Since $\mathbb E[\widetilde Q_x^{(i)}(e)]=Q_x^{(i)}(e)$, we can apply additive-multiplicative Chernoff (Theorem~\ref{thm:am-Chernoff}) to get the desired bound. Each independent random variable ($\widetilde Q_x^{(i)}(e)$ for $e\in E_i$) is in the range $[0,2^{-i+1}/(pk)]$ by definition of $E_i$. Therefore we get
	\begin{align*}
		\mathbb P\left[\big|\widetilde Q_x^{(i)}(E_i)-Q_x^{(i)}(E_i)\big|>\frac{\epsilon}{2} Q_x^{(i)}(E_i)+\frac{\epsilon}{10r^2\log n}\right]&\le2\exp\left(-\frac{\epsilon/2\cdot\frac{\epsilon}{10r^2\log n}}{3\cdot 2^{-i+1}/(pk)}\right)\\
		&=2\exp\left(-\frac{\epsilon^2pk\cdot2^i}{120r^2\log n}\right).
	\end{align*}
	We will now use a union bound to prove that this holds simultaneously for all possible settings of $E_i$ and $Q_x^{(i)}$. Recall that by definition $\bigcup_{j=1}^i C_j$ contains exactly arcs of $C(E)$ that contribute more than $2^{-i}/k$ energy to the total energy of $Q^C(x)=1$. There are at most $k\cdot2^i$ such arcs, but since each arc appears with multiplicity at least $k$, there are at most $2^i$ \textit{distinct} arcs. The $Q_x^{C,(i)}$-value of all arcs not in $\bigcup_{j=1}^i C_j$ is zero. To select this multiset of non-zero valued arcs there are
	$$\binom{n^2}{2^i}\le n^{2\cdot2^i}=\exp\left(2\cdot2^i\log n\right)$$
	options. Furthermore, for each relevant arc, we must choose its $Q_x^{C,(i)}$-value: This is an integer multiple of $1/kn^3$ in the range $[-1/k,1/k]$ and so there are $2n^3$ options per arc---for a total of
	$${\left(2n^3\right)}^{2^i}\le\exp\left(4\cdot2^i\log n\right)$$
	options. Finally, for each relevant arc, we must choose which category among $E_1,\ldots,E_i$ it belongs to (as this may not be evident from the value of $Q_x^{C,(i)}$). This is an additional $i\le3\log n+1$ options per arc---for a total of
	$${\left(3\log n+1\right)}^{2^i}\le\exp\left(2^i\log n\right),$$
	options among all arcs.

	Ultimately, there are
	$$\exp\left(2\cdot2^i\log n + 4\cdot2^i\log n+2^i\log n\right)=\exp\left(7\cdot2^i\log n\right)$$
	possible settings of $(E_1,\ldots,E_i,Q_x^{C,(i)})$.

	Combining the above Chernoff bound for a single setting of $x$ with this union bound we get the statement of the claim:
	\begin{align*}
		&\mathbb P\left[\forall x:\ \big|Q_x^{(i)}(E_i)-\widetilde Q_x^{(i)}(E_i)\big|> \frac{\epsilon}{2} Q_x^{(i)}(E_i)+\frac{\epsilon}{10r^2\log n}\right]\\
		\le&2\exp\left(7\cdot2^i\log n\right)\cdot\exp\left(-\frac{\epsilon^2pk\cdot2^{-i}}{120r^2}\right)\\
		=&2\exp\left(2^i\cdot\left(7\log n-\frac{\epsilon^2pk}{120r^2}\right)\right)\\
		\le & \frac1n,
	\end{align*}
	since $pk=1000r^2\log n/\epsilon^2$.
\end{proof}

\claimdthree*

\begin{proof}%[Proof of Claim~\ref{claim:d3}]
	We consider the high probability event that $|\widetilde E|\le4pkn^2$. (Lemma~\ref{lem:d-size-bound}). Similarly to the proof of Claim~\ref{claim:d1} we simply sum over all edges of $E_i$. Note that if $e\in\widetilde E$,
	\[
		\big|\widetilde Q_x(e)-\widetilde Q_x^{(i)}(e)\big|\le\sum_{e\in E_i}\big|\widetilde Q_x(e)-\widetilde Q^{(i)}_x(e)\big|=\sum_{e\in E_i\cap\widetilde E}\frac1p\big|Q_x(e)-Q_x^{(i)}(e)\big|\le\frac{|\widetilde E|}{pkn^3}\le\frac{4}{n}. \qedhere
	\]
\end{proof}

\claimdfour*

\begin{proof}%[Proof of Claim~\ref{claim:d4}]
	Note that
	$$\big|Q_x(E_*)-\widetilde Q_x(E_*)\big|\le Q_x(E_*)+\widetilde{Q}_x(E_*).$$
	We bound the two terms separately:
	$$Q_x(E_*)\le|E_*|\cdot \frac{1}{kn^3}\le|E|\cdot\frac1{kn^3}\le\frac2n,$$
	and
	$$\widetilde Q_x(E_*)\le|\widetilde E|\cdot\frac{1}{pkn^3}\le\frac{4}{n},$$
	with high probability by Lemma~\ref{lem:d-size-bound}.
\end{proof}

\subsection{Proof of Theorem~\ref{thm:directed-hypergraph-sparsification}}\label{subsec:directed-hypergraph-sparsification}
\begin{proof}[Proof of Theorem~\ref{thm:directed-hypergraph-sparsification}]
	Given the results of Lemma~\ref{lem:directed-main}, we only need to decompose $G$ into directed hypergraphs with their hyperedges having nearly the same overlap.
	We will repeatedly separate and sparsify hyperarcs of the highest overlap until no hyperarcs remain. This results in an $\epsilon$-spectral sparsifier of $G$, since the quality of being an $\epsilon$-spectral sparsifier is additive.

	Consider the following simple algorithm, where \textsc{UniformSamplingSparsify} denotes the sparsification algorithm given in Lemma~\ref{lem:directed-main}:
	\begin{algorithm}[H]\label{alg:directed}
		\caption{Directed hypergraph sparsification}
		\begin{algorithmic}[1]
			\Procedure{Sparsify}{$G=(V,E)$}
			\State $\widetilde E\gets\emptyset$.
			\While{$E\neq\emptyset$}
				\State $2k\gets$ the largest overlap among hyperedges in $E$.
				\State $E'\gets$ the maximal $k$-overlapping set in $E$.\label{line:maximal}
				\State $E\gets E\setminus E'$.
				\State $\widetilde E\gets\widetilde E\cup\textsc{UniformSamplingSparsify}(V,E')$.
			\EndWhile
			\State \Return $(V,\widetilde E)$.
			\EndProcedure
		\end{algorithmic}
	\end{algorithm}
	Note first that the maximal set of a certain multiplicity (in Line~\ref{line:maximal}) is indeed unique. It follows from definition that the union of hyperarc sets of multiplicity $k$ still has multiplicity $k$. Therefore, $E'$ contains all hyperarcs of overlap at least $k$ form (the current) $E$. Furthermore, removing $E'$ from $E$ reduces the maximum overlap of any hyperarc to below $k$, so by a factor of at least $2$.
	Since the maximum overlap started out is at most $n^{r-2}$, Algorithm~\ref{alg:directed} terminates in at most $r\log n$ iterations.
	Since the size of $\widetilde E$ increased by at most $O(n^2r^2\log n/\epsilon^2)$ in each iteration, by Lemma~\ref{lem:directed-main}, this results in $|\widetilde E| = O(n^2r^3\log^2 n/\epsilon^2)$, as claimed.

	The running time is polynomial because we can compute overlaps of hyperarcs in polynomial time by Remark~\ref{rem:overlap-time-complexity}, and hence can compute the largest overlap $k$ and the maximal $k$-overlapping set in polynomial time.
\end{proof}

\fi

\ifprocs
\begin{acks}
\else
\section*{Acknowledgement}
\fi
This project has received funding from the European Research Council (ERC) under the European Union's Horizon 2020 research and innovation programme (grant agreement No 759471).
\ifprocs
\end{acks}
\fi

\ifprocs
\bibliographystyle{alphaurl}
\bibliography{999-bibliography}
\else
\bibliographystyle{alphaurl}
\bibliography{999-bibliography}
\fi

\begin{appendix}
\section{Technical lemmas}\label{app:A}
\subsection{Concentration Inequalities}
The following concentration bound is standard.
\begin{theorem}[Chernoff bound, see e.g.~\cite{DBLP:books/daglib/0021015}]\label{thm:chernoff}
	Let $X_1,\ldots,X_n$ be independent random variables in the range $[0,a]$. Let $\sum_{i=1}^n X_i=S$. Then for any $\delta\in [0,1]$ and $\mu\ge\mathbb ES$,
	$$\mathbb P[|S-\mathbb ES|\ge\delta\mu]\le2\exp\left(-\frac{\delta^2\mu}{3a}\right).$$
\end{theorem}
The following slight variation, allowing for both multiplicative and additive error, will be the most convenient for our purposes throughout the paper.
\begin{theorem}[Additive-multiplicative Chernoff bounds \cite{Badanidiyuru2013}]\label{thm:am-Chernoff}
	Let $X_1,\ldots X_n$ be independent random variables in the range $[0,a]$.  Let $\sum_{i=1}^n X_i=S$. Then for all $\delta\in[0,1]$ and $\alpha>0$,
	$$\mathbb P[|S-\mathbb ES|\ge\delta\mathbb ES+\alpha]\le2\exp\left(-\frac{\delta\alpha}{3a}\right).$$
\end{theorem}

\subsection{Proof of Hypergraph Cheeger's Inequality}
\begin{proofof}{Theorem~\ref{thm:hypergraph-cheeger}}
	Recall that $\sum_{v\in V}x_vd(v)=0$. Suppose for contradiction that there exists a vector $x \in \mathbb{R}^V$ such that $Q(x)<\frac{r\Phi^2}{32}\sum_{v\in V}x_v^2d(v)$. Let $\widetilde x\in\mathbb R^V$ be such $x$ shifted such that $\sum_{v\in V}x_vd_x(v)=0$,	where $d_x(v)$ denotes the degree of $v$ in $G_x=(V,E_x)$.
	Then, we have
	\begin{align*}
	& Q(x)
	<\frac{r\Phi^2}{32}\sum_{v\in V}x_v^2d(v)
	\le\frac{r\Phi^2}{32}\sum_{v\in V}\widetilde x_v^2d(v)
	=\frac{r\Phi^2}{32}\sum_{e\in E}\sum_{v\in e}\widetilde x_v^2
	\\
	&\le\frac{r^2\Phi^2}{32}\sum_{e\in E}\max_{v\in e}\widetilde x_v^2
	\le \frac{r^2\Phi^2}{32}\sum_{(a,b)\in E_x}\left(\widetilde x_a^2+\widetilde x_b^2\right)
	=\frac{r^2\Phi^2}{32}\sum_{v\in V}\widetilde x_v^2d_x(v),
	\end{align*}

	The second inequality follows since $x$ is centered, that is $\sum_{v\in V}x_vd(v)=0$.

	This means, by Cheeger's inequality for ordinary graphs~\cite{Alon1985,Alon1986}, that there exists a vertex set $S$ of expansion $\frac{r\Phi}4$ in $G_x$.
	Moreover, $S$ can be chosen to be a \emph{sweep cut} with respect to $x$ (regardless of the degree vector) in the sense that $S$ is of the form $\{v \in V \mid x_v \leq \tau\}$ or $\{v \in V \mid x_v \geq \tau\}$ for some $\tau \in \mathbb{R}$.
	Let $S\subseteq V$ be the smaller side of the cut (in volume).
	Let $\eta := |E_x(S,V \setminus S)|$ and $\zeta := |E(S)|$.
	Then, we have
	$$\eta\le\frac{r\Phi}4\sum_{v\in S}d_x(v)=\frac{r\Phi}4(\eta+2\zeta).$$
	Since $\Phi\le\frac2r$, $\frac{r\Phi}{4}\le\frac12$ and so $\zeta\ge\frac{\eta}{r\Phi}$.
	Since $S$ is a sweep cut with respect to $x$, each edge of $G_x$ crossing the cut $(S,V\setminus S)$ corresponds to a distinct hyperedge of $G$ also crossing the cut, and each edge of $G_x$ fully inside $S$ translates to a hyperedge of $G$ fully inside $S$.
	Therefore, the number of edges crossing the cut $(S,V \setminus S)$ in $G$ is still $\eta$ and $\sum_{v\in S}d(v)>r\zeta\ge\frac\eta\Phi$.
	Similarly, $\sum_{v\in V \setminus S}d(v)>\frac\eta\Phi$. Therefore, the expansion of the cut $(S,V \setminus S)$ in $G$ is less than $\Phi$, which is a contradiction.
\end{proofof}

\end{appendix}

\end{document}